\documentclass[letterpaper,11pt]{article}

\usepackage[margin = 1in]{geometry}
\usepackage{fullpage}

\usepackage{amsmath,amssymb}

\usepackage{mathtools}
\mathtoolsset{centercolon} % center the colon in ":="

\usepackage[retainorgcmds]{IEEEtrantools}

\usepackage{tikz}
\usetikzlibrary{positioning,calc}

\usepackage[colorlinks = true]{hyperref}

\hypersetup{
  pdftitle = {Quantum walks can find a marked element on any graph},
  pdfauthor = {Hari Krovi, Frederic Magniez, Maris Ozols, Jeremie Roland}
}

\definecolor{darkred}  {rgb}{0.5,0,0}
\definecolor{darkblue} {rgb}{0,0,0.5}
\definecolor{darkgreen}{rgb}{0,0.5,0}
\hypersetup{
  urlcolor   = blue,         % color of external links
  linkcolor  = darkblue,     % color of internal links
  citecolor  = darkgreen,    % color of links to bibliography
  filecolor  = darkred       % color of file links
}

\newcommand{\ket}[1]{|#1\rangle}
\newcommand{\bra}[1]{\langle#1|}
\newcommand{\braket}[2]{\langle#1|#2\rangle}
\newcommand{\proj}[1]{|#1\rangle\langle#1|}

\DeclarePairedDelimiter{\set}{\lbrace}{\rbrace}
\DeclarePairedDelimiter{\abs}{\lvert}{\rvert}
\DeclarePairedDelimiter{\norm}{\lVert}{\rVert}
\DeclarePairedDelimiter{\floor}{\lfloor}{\rfloor}
\DeclarePairedDelimiter{\ceil}{\lceil}{\rceil}

\newcommand{\C}{\mathbb{C}}
\newcommand{\R}{\mathbb{R}}
\newcommand{\N}{\mathbb{N}}

\newcommand{\ct}{^{\dagger}}
\newcommand{\tp}{^{\mathsf{T}}}
\newcommand{\x}{\otimes}

\newcommand{\mx}[1]{\begin{pmatrix}#1\end{pmatrix}}
\newcommand{\smx}[1]{\bigl(\begin{smallmatrix}#1\end{smallmatrix}\bigr)}

\newcommand{\Hi}{\mathcal{H}}

\newcommand{\Reg}{\mathsf{R}}

\newcommand{\m}{m} %{\abs{M}} % Number of marked elements

\DeclareMathOperator{\spn}{span}
\DeclareMathOperator{\diag}{diag}

\DeclareMathOperator{\HT}{HT}
\DeclareMathOperator{\limHT}{\HT^{+}(\mathit{P,M})}

\newenvironment{algobox}[1]
{\begin{center}\begin{minipage}{#1\textwidth}\hrulefill\\}
{\vspace{-8pt}\hrulefill\end{minipage}\end{center}}

\newcommand{\alg}[1]{\textbf{#1}}

\newcommand{\RWAtxt}{\alg{Random Walk Algorithm}}      \newcommand{\RWA}{\hyperref[alg:RWA]{\RWAtxt}}
\newcommand{\ASAtxt}{\alg{Adiabatic Search Algorithm}} 
\newcommand{\Itptxt}{\alg{Interpolation}}              \newcommand{\Itp}{\hyperref[alg:Itp]{\Itptxt}}
\newcommand{\QEEtxt}{\alg{Eigenvalue Estimation}}      \newcommand{\QEE}{\hyperref[alg:QEE]{\QEEtxt}}
\newcommand{\QWStxt}{\alg{Search}}                     \newcommand{\QWS}{\hyperref[alg:QWS]{\QWStxt}}
\newcommand{\QWRtxt}{\alg{Incremental Search}}         \newcommand{\QWR}{\hyperref[alg:QWR]{\QWRtxt}}

\newcommand{\PFT}{\hyperref[thm:PF]{Perron--Frobenius Theorem}}

\newcommand{\Thm}[1]{\hyperref[thm:#1]{Theorem~\ref*{thm:#1}}}
\newcommand{\Lem}[1]{\hyperref[lem:#1]{Lemma~\ref*{lem:#1}}}
\newcommand{\Prop}[1]{\hyperref[prop:#1]{Prop.~\ref*{prop:#1}}}
\newcommand{\Cor}[1]{\hyperref[cor:#1]{Corollary~\ref*{cor:#1}}}
\newcommand{\Def}[1]{\hyperref[def:#1]{Definition~\ref*{def:#1}}}

\newcommand{\Sect}[1]{\hyperref[sect:#1]{Sect.~\ref*{sect:#1}}}
\newcommand{\Chap}[1]{\hyperref[chap:#1]{Chapter~\ref*{chap:#1}}}
\newcommand{\Apx}[1]{\hyperref[apx:#1]{Appendix~\ref*{apx:#1}}}

\newcommand{\Fig}[1]{\hyperref[fig:#1]{Fig.~\ref*{fig:#1}}}
\newcommand{\Tab}[1]{\hyperref[tab:#1]{Table~\ref*{tab:#1}}}

\newcommand{\EqRef}[1]{\hyperref[eq:#1]{(\ref*{eq:#1})}}
\newcommand{\Eq}[1]{Eq.~\hyperref[eq:#1]{(\ref*{eq:#1})}}
\newcommand{\step}[1]{\hyperref[step:#1]{step~\ref*{step:#1}}}

\def\eg{\textit{e.g.}}
\def\ie{\textit{i.e.}}
\def\etal{\textit{et al.}}

\newcommand{\X}{\mathcal{X}}
\newcommand{\reflex}{\mathrm{ref}}
\newcommand{\shift}{\textsc{Shift}}
\newcommand{\ancilla}{\ket{\bar{0}}}

\newcommand{\pstar}{p^*}
\newcommand{\sstar}{s^*}

\newcommand{\Order}{\mathrm{O}}

\renewcommand{\check}{\mathtt{Check}}
\newcommand{\setup}{\mathtt{Setup}}
\newcommand{\update}{\mathtt{Update}}
\newcommand{\checkingcost}{\mathsf{C}}
\newcommand{\setupcost}{\mathsf{S}}
\newcommand{\updatecost}{\mathsf{U}}

\newcommand{\find}{\mbox{\textsc{Find}}}
\newcommand{\detect}{\mbox{\textsc{Detect}}}
\newcommand{\sample}{\mbox{\textsc{Sample-marked}}}

\usepackage{amsthm}

\usepackage{thmtools}
\usepackage{thm-restate}

\theoremstyle{definition}

\newtheorem{definition}{Definition}

\theoremstyle{plain}

\newtheorem*{theorem*}{Theorem}
\newtheorem{theorem}[definition]{Theorem}
\newtheorem{lemma}[definition]{Lemma}
\newtheorem{corollary}[definition]{Corollary}
\newtheorem{proposition}[definition]{Proposition}

\begin{document}

% Title
\title{Quantum walks can find a marked element on any graph\footnote{A preliminary version of this work appeared in the \textit{Proceedings of the 37th International Colloquium on Automata, Languages and Programming}, volume 6198 of Lecture Notes in Computer Science, pages 540--551, Springer, 2010.}}

\newcommand{\s}[1]{\textsuperscript{#1}}

% Authors
\author{
  Hari Krovi\footnote{
    Quantum Information Processing Group,
    Raytheon BBN Technologies, Cambridge,
    Massachusetts 02138, USA} \and
  Fr\'ed\'eric Magniez\footnote{
    CNRS, LIAFA, Univ Paris Diderot,
    Sorbonne Paris-Cit\'e, 75205 Paris, France} \and
  Maris Ozols\footnote{
    DAMTP, Centre for Mathematical Sciences,
    University of Cambridge,
    Cambridge CB3 0WA, UK} \and
  J\'er\'emie Roland\footnote{
    QuIC, Ecole Polytechnique de Bruxelles,
    Universit\'e Libre de Bruxelles (ULB),
    1050 Brussels, Belgium}
}

\date{January 30, 2014}

\maketitle

\begin{abstract}
We solve an open problem by constructing quantum walks that not only detect but also find marked vertices in a graph. In the case when the marked set $M$ consists of a single vertex, the number of steps of the quantum walk is quadratically smaller than the classical hitting time $\HT(P,M)$ of any reversible random walk $P$ on the graph. In the case of multiple marked elements, the number of steps is given in terms of a related quantity $\limHT$ which we call extended hitting time.

Our approach is new, simpler and more general than previous ones. We introduce a notion of interpolation between the random walk $P$ and the absorbing walk $P'$, whose marked states are absorbing. Then our quantum walk is simply the quantum analogue of this interpolation. Contrary to previous approaches, our results remain valid when the random walk $P$ is not state-transitive. We also provide algorithms in the cases when only approximations or bounds on parameters $p_M$ (the probability of picking a marked vertex from the stationary distribution) and $\limHT$ are known.
\end{abstract}

\setcounter{tocdepth}{4}
\tableofcontents

%%%%%%%%%%%%%%%%%%%%%%
\section{Introduction}
%%%%%%%%%%%%%%%%%%%%%%

Many randomized classical algorithms rely heavily on random walks or Markov chains. This technique has been extended to the quantum case and is called \emph{quantum walk}. Ambainis~\cite{Ambainis04} was the first to solve a natural problem---the element distinctness problem---using a quantum walk. Following this, many other quantum walk algorithms were discovered, for example,~\cite{MagniezSS05,BuhrmanS06,MagniezN05}.

A common class of problems that are typically solved using a random walk are the so-called \emph{spatial search problems}. In such problems, the displacement constraints are modelled by edges of an undirected graph $G$, which has some desired subset of vertices $M$ that are marked. The goal of a spatial search problem is to find one of the marked vertices by traversing the graph along its edges. Classically, a simple strategy for finding a marked vertex is to perform a random walk on $G$, by repeatedly applying some stochastic matrix $P$ until one of the marked vertices is reached (see \Sect{Classical hitting time} for more details). The expected running time of this algorithm is called the \emph{hitting time} of $P$ and is denoted by $\HT(P,M)$.

Quantum walk algorithms for the spatial search problem were studied in~\cite{AaronsonA05}. This problem has also been considered for several specific graphs, such as the hypercube~\cite{ShenviKW03} and the grid~\cite{childs2,AKR}. The notion of the hitting time has been carried over to the quantum case in~\cite{AKR,Kempe,Sze,KB,MNRS,MNRS2,VKB} by generalizing the classical notion in different ways. Usually, the quantum hitting time has a quadratic improvement over the classical one. However, 
several serious restrictions were imposed for this to be the case. A quantum algorithm could only solve the \emph{detection problem} of deciding whether there are marked vertices or not~\cite{Sze}, but for being able to \emph{find} them, the Markov chain had to be reversible, state-transitive, and with a unique marked vertex~\cite{Tulsi,MNRS2}. The detection algorithm is quite intuitive and well understood, whereas the finding algorithm requires an elaborate proof whose intuition is not clear. This is due in part to a modification of the quantum walk, so that the resulting walk is not a quantum analogue of a Markov chain anymore.

Whether this quadratic speed-up for finding a marked element also holds for all reversible Markov chains 
was an open question. We give a positive answer to this question by providing a quantum algorithm for solving this problem.
The case of multiple marked elements still remains open, because of a possible gap between the so-called extended hitting time $\limHT$, which characterizes the cost of our quantum algorithms, and the standard hitting time $\HT(P,M)$ (see \Sect{extended-hitting-time} and \Apx{HT comparison} for more details\footnote{Note that in the preliminary version of this work~\cite{KMOR}, a subtle error led to the wrong conclusion that $\limHT=\HT(P,M)$ for all $M$ and reversible $P$, while in general this only holds when $\abs{M}=1$.}).

%-----------------------%
\subsection{Related work}
%-----------------------%

Inspired by Ambainis' quantum walk algorithm for solving the element distinctness problem~\cite{Ambainis04}, Szegedy~\cite{Sze} has introduced a powerful way of constructing quantum analogues of Markov chains which led to new quantum walk-based algorithms. He showed that for any symmetric Markov chain a quantum walk could detect the presence of marked vertices in at most the square root of the classical hitting time. However, showing that a marked vertex could also be found in the same time (as is the case for the classical algorithm) proved to be a very difficult task. Magniez~\etal~\cite{MNRS} extended Szegedy's approach to the larger class of ergodic Markov chains, and proposed a quantum walk-based algorithm to find a marked vertex, but its complexity may be larger than the square root of the classical hitting time. A typical example where their approach fails to provide a quadratic speed-up is the 2D grid, where their algorithm has complexity $\Theta(n)$, whereas the classical hitting time is $\Theta(n \log n)$. Ambainis~\etal~\cite{AKR} and Szegedy's~\cite{Sze} approaches yield a complexity of $\Theta(\sqrt{n} \log n)$ in this special case, for a unique marked vertex. Childs and Goldstone~\cite{childs1,childs2} also obtained a similar result using a continuous-time quantum walk.

However, whether a full quadratic speed-up was possible in the 2D grid case remained an open question, until Tulsi~\cite{Tulsi} proposed a solution involving a new technique. Magniez~\etal~\cite{MNRS2} extended Tulsi's technique to any reversible state-transitive Markov chain, showing that for such chains, it is possible to find a unique marked vertex with a full quadratic speed-up over the classical hitting time. However, the state-transitivity is a strong symmetry condition, and furthermore their technique cannot deal with multiple marked vertices. Recently \cite{ABNOR} have suggested to modify the original~\cite{AKR} algorithm in the case of the 2D grid with a single marked element, by replacing amplitude amplification with classical search in a neighbourhood of the final vertex. This results in a $\sqrt{\log n}$ speed-up over the original algorithm from~\cite{AKR} and yields complexity $\Order(\sqrt{n \log n})$ as in the case of~\cite{Tulsi,MNRS2}.

It seems implausible that one has to rely on involved techniques to solve the finding problem under such restricted conditions in the quantum case, while the classical random walk algorithm (see \Sect{Classical hitting time}) is conceptually simple and works under general conditions. The classical algorithm simply applies \emph{absorbing} walk $P'$ obtained from $P$ by turning all outgoing transitions from marked states into self-loops (see \Apx{Semi-absorbing}). Each application of $P'$ results in more probability being absorbed in marked states.

Previous attempts at providing a quantum speed-up over this classical algorithm have followed one of these two approaches:
\begin{itemize}
  \item Combining a quantum version of $P$ with a reflection through marked vertices to mimic a Grover operation~\cite{AKR,Ambainis04,MNRS}.
  \item Directly applying a quantum version of $P'$~\cite{Sze,MNRS2}.
\end{itemize}
The problem with these approaches is that they would only be able to find marked vertices in very restricted cases. We explain this by the different nature of random and quantum walks: while both have a stable state, \ie, the stationary distribution for the random walk and the eigenstate with eigenvalue $1$ for the quantum walk, the way both walks act on other states is dramatically different.

Indeed, an ergodic random walk will converge to its stationary distribution from any initial distribution. This apparent robustness may be attributed to the inherent randomness of the walk, which will smooth out any initial perturbation. After many iterations of the walk, non-stationary contributions of the initial distribution will be damped and only the stationary distribution will survive (this can be attributed to the thermodynamical irreversibility\footnote{Reversibility of Markov chains (see \Apx{Reversibility}) is not related to thermodynamical reversibility. Actually, even a ``reversible" Markov chain is  thermodynamically irreversible.} of ergodic random walks).

On the other hand, this is not true for quantum walks, because in the absence of measurements a unitary evolution is deterministic (and in particular thermodynamically reversible): the contributions of the other eigenstates will not be damped but just oscillate with different frequencies, so that the overall evolution is quasi-periodic. As a consequence, while iterations of $P'$ always lead to a marked vertex, it may happen that iterations of the quantum analogue of $P'$ will never lead to a state with a large overlap over marked vertices, unless the walk exhibits a strong symmetry (as is the case for a state-transitive walk with only one marked element, which could be addressed by previous approaches).

%-----------------------------------------%
\subsection{Our approach and contributions}
%-----------------------------------------%

Our main result is that a quadratic speed-up for finding a marked element via quantum walk holds for any reversible Markov chain with a single marked element. We provide several algorithms for different versions of the problem. Compared to previous results, our algorithms are more general and conceptually clean. The intuition behind our main algorithm is based on the adiabatic algorithm from~\cite{KOR}. However, all algorithms presented here are circuit-based and thus do not suffer from the drawbacks of the adiabatic algorithm in~\cite{KOR}.

We choose an approach that is different from the ones described above: first, we directly modify the original random walk $P$, and then construct a quantum analogue of the modified walk. We choose the modified walk to be the interpolated Markov chain $P(s) = (1-s) P + s P'$ that interpolates between $P$ and the absorbing walk $P'$ whose outgoing transitions from marked vertices have been replaced by self-loops. Thus, we can still use our intuition from the classical case, but at the same time also get simpler proofs and more general results in the quantum case.

All of our quantum walk algorithms are based on eigenvalue estimation performed on operator $W(s)$, a quantum analogue of Markov chain $P(s)$. We consider the $(+1)$-eigenstate $\ket{\Psi_n(s)}$ of $W(s)$ that plays the role of the stationary distribution in the quantum case. We use the interpolation parameter $s$ to tune the length of projections of $\ket{\Psi_n(s)}$ onto marked and unmarked vertices. If both projections are large, our algorithm succeeds with large probability in $\Order \bigl( \! \sqrt{\limHT} \bigr)$ steps (\Thm{Search with known pm and HT}), where $\limHT$ is a quantity we call the extended hitting time (see \Def{HT(s)}, in particular, $\limHT=\HT(P,M)$ when $\abs{M}=1$).

We also provide several modifications of the main algorithm. In particular, we show how to make a suitable choice of $s$ to balance the overlap of $\ket{\Psi_n(s)}$ on marked and unmarked vertices even if some of the parameters required by the main algorithm are unknown and the rest are either approximately known (\Thm{Search with known HT} and \Thm{Search with unknown HT}) or bounded (\Thm{Search with bound on pmin and HTmax} and \Thm{Search with bound on HTmax}). In all cases a marked vertex is found in $\Order \bigl( \! \sqrt{\limHT} \bigr)$ steps.

In \Sect{Search preliminaries} we introduce several variations of the spatial search problem and provide preliminaries on random and quantum walks and their hitting times.
\Sect{Search algorithms} describes our quantum algorithms and contains the main results. The main algorithm is presented in \Sect{Main algorithm} and is followed by several modifications that execute the main algorithm many times with different parameters.

Technical and background material is provided in several appendices. In \Apx{Semi-absorbing} we describe basic properties of the interpolated Markov chain $P(s)$ and the extended hitting time $\limHT$, which is crucial for the analysis of the algorithms in \Sect{Search algorithms}. In \Apx{W(s)} we compute the spectrum of the walk operator $W(s)$ and show how it can be implemented for any $s$. In \Apx{HT comparison} we discuss limitations of our results for the case of multiple marked elements.

%%%%%%%%%%%%%%%%%%%%%%%
\section{Preliminaries} \label{sect:Search preliminaries}
%%%%%%%%%%%%%%%%%%%%%%%

%---------------------------------%
\subsection{Classical random walks} \label{sect:Random walks}
%---------------------------------%

A Markov chain\footnote{We will use terms ``random walk'', ``Markov chain'', and ``stochastic matrix'' interchangeably. The same holds for ``state'', ``vertex'', and ``element''.} on a discrete state space $X$ of size $n := \abs{X}$ is described by an $n \times n$ \emph{row-stochastic matrix} $P$ where $P_{xy} \in [0,1]$ is the transition probability from state $x$ to $y$ and
\begin{equation}
  \forall x \in X: \sum_{y \in X} P_{xy} = 1.
\end{equation}
Such Markov chain has a corresponding \emph{underlying directed graph} with $n$ vertices labelled by elements of $X$, and directed arcs labelled by \emph{non-zero} probabilities $P_{xy}$ (see \Fig{Graph}).

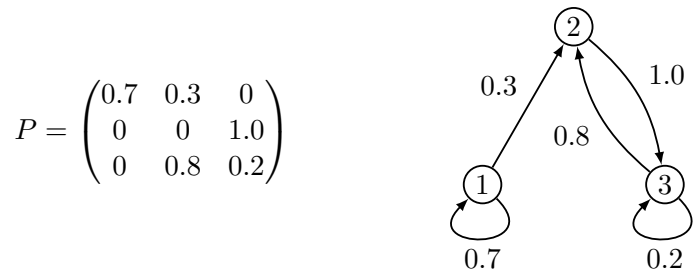
\begin{figure}[th]
  \centering
  %\input{fig-graph.tex}
  % !TeX root = KMOR.tex

\def\dl{45} % delta angle for loops

\begin{tikzpicture}
  [line width = 0.7pt,
   crc/.style = {circle, draw = black, minimum size = 5.0mm, inner sep = 1pt},
   arc/.style = {-latex},
   lop/.style = {loop below, min distance = 10mm, in = -180 + \dl, out = -\dl}]

  \def\r{1.4} % radius

  % Matrix

  \path (-4*\r,0) node {$P =
    \mx{ 0.7 & 0.3 & 0   \\
         0   & 0   & 1.0 \\
         0   & 0.8 & 0.2 }$};

  % Vertices

  \path (-150:\r) node [crc] (v1) {1};
  \path (  90:\r) node [crc] (v2) {2};
  \path (- 30:\r) node [crc] (v3) {3};

  % Arcs

  \def\da{20} % delta angle for arcs

  \draw [arc] (v1) to node [label = 150:0.3] {} (v2);
  \draw [arc] (v2) to [out = -60 + \da, in = 120 - \da] node [label =  30:1.0] {} (v3);
  \draw [arc] (v3) to [out = 120 + \da, in = -60 - \da] node [label = 210:0.8] {} (v2);

  % Self-loops

  \draw [arc] (v1) to [lop] node {0.7} (v1);
  \draw [arc] (v3) to [lop] node {0.2} (v3);

\end{tikzpicture}

  \caption[Markov chain $P$ and the corresponding graph]{Markov chain $P$ and the corresponding graph with transition probabilities.}
  \label{fig:Graph}
\end{figure}

We represent probability distributions by \emph{row} vectors whose entries are real, non-negative, and sum to one. When one step of Markov chain $P$ is applied to a given distribution $p$, the resulting distribution is $pP$. A probability distribution $\pi$ that satisfies $\pi P = \pi$ is called a \emph{stationary distribution} of $P$. For more background on Markov chains see, \eg,~\cite{GrinsteadSnell,KemenySnell,KoralovSinai}.

%........................%
\subsubsection{Ergodicity} \label{sect:Ergodicity}
%........................%

Let us consider Markov chains with some extra structure.

\begin{definition}\label{def:Ergodicity}
A Markov chain is called
\begin{itemize}
\item \emph{irreducible}, if any state in the underlying directed graph can be reached from any other by a finite number of steps (\ie, the graph is strongly connected);
\item \emph{aperiodic}, if there is no integer $k > 1$ that divides the length of every directed cycle of the underlying directed graph;
\item \emph{ergodic\index{Markov chain!ergodic}}, if it is both irreducible and aperiodic.
\end{itemize}
\end{definition}

Equivalently, a Markov chain $P$ is ergodic if there exists some integer $k_0 \geq 1$ such that all entries of $P^{k_0}$ (and, in fact, of $P^k$ for any $k \geq k_0$) are strictly positive. Some authors call such chains \emph{regular} and use the term ``ergodic'' already for irreducible chains~\cite{GrinsteadSnell,KemenySnell}. From now on we will almost exclusively consider only ergodic Markov chains.

Even though some of the Markov chain properties in \Def{Ergodicity} are independent from each other (such as irreducibility and aperiodicity), usually they are imposed in a specific order which is summarized in \Fig{Hierarchy}. As we impose more conditions, more can be said about the spectrum of $P$ as discussed in the next section.

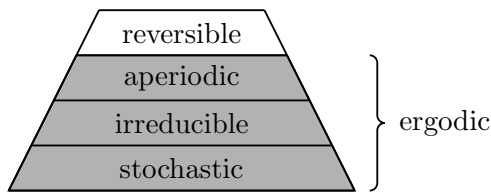
\begin{figure}[th]
  \centering
%  \input{fig-hierarchy.tex}
% !TeX root = KMOR.tex

\begin{tikzpicture}[line width = 0.7pt]

  \def\w{2.3} % width
  \def\h{0.6} % height
  \def\c{0.5} % slope

  % Fill
  \draw [fill = black!30]
        (+\w,0) --
        (-\w,0) --
        (-\w+3*\h*\c,3*\h) --
        (+\w-3*\h*\c,3*\h) -- cycle;

  % Draw sides
  \draw (-\w,0) -- (-\w+4*\h*\c,4*\h);
  \draw (+\w,0) -- (+\w-4*\h*\c,4*\h);

  % Horizontal lines and labels
  \foreach \i/\txt in {0/stochastic, 1/irreducible, 2/aperiodic, 3/reversible, 4/} {
    \draw (-\w+\i*\h*\c,\i*\h) --
          (+\w-\i*\h*\c,\i*\h);
    \draw (0,{(\i+1/2)*\h}) node {\txt};
  }

  \def\r{0.1} % arc radius

  \draw (\w+2*\r,0)    arc(-90:0:\r) to +(0, 1.5*\h-2*\r) arc(180: 90:\r);
  \draw (\w+2*\r,3*\h) arc( 90:0:\r) to +(0,-1.5*\h+2*\r) arc(180:270:\r);

  \draw (\w+1.2,3*\h/2) node {ergodic};

\end{tikzpicture}

  \caption[Summary of Markov chain properties]{The order in which Markov chain properties from~\Def{Ergodicity} are typically imposed. Reversibility will be defined in \Apx{Reversibility}.}
  \label{fig:Hierarchy}
\end{figure}

%.......................................%
\subsubsection{Perron--Frobenius theorem}
%.......................................%

The following theorem will be very useful for us. It is essentially the standard Perron--Frobenius theorem~\cite[Theorem~8.4.4, p.~508]{HornJohnson}, but adapted for Markov chains. (This theorem is also known as the ``Ergodic Theorem for Markov chains'' \cite[Theorem~5.9, p.~72]{KoralovSinai}.) The version presented here is based on the extensive overview of Perron--Frobenius theory in~\cite[Chapter~8]{Meyer}.

\begin{theorem}[Perron--Frobenius]\label{thm:PF}
Let $P$ be a stochastic matrix. Then
\begin{itemize}
  \item all eigenvalues of $P$ are at most $1$ in absolute value and $1$ is an eigenvalue of $P$;
  \item if $P$ is irreducible, then the $1$-eigenvector is unique and strictly positive (\ie, it is of the form $c \pi$, where $c \neq 0$ and $\pi$ is a probability distribution that is non-zero everywhere);
  \item if in addition to being irreducible, $P$ is also aperiodic (\ie, $P$ is ergodic), then the remaining eigenvalues of $P$ are strictly smaller than $1$ in absolute value.
\end{itemize}
\end{theorem}

If $P$ is irreducible but not aperiodic, it has some complex eigenvalues on the unit circle (which can be shown to be roots of unity)~\cite[Chapter~8]{Meyer}. However, when in addition we also impose aperiodicity (and hence ergodicity), we are guaranteed that there is a unique eigenvalue of absolute value $1$ and, in fact, it is equal to $1$.

%-----------------------------------%
\subsection{Spatial search on graphs}
%-----------------------------------%

We fix an undirected graph $G = (X, E)$ with $n := \abs{X}$ vertices and a set of edges $E$. Let $M \subseteq X$ be a set of marked vertices of size $\m := \abs{M}$. We insist that during the traversing of the graph the current vertex is stored in a distinguished \emph{vertex register}. Our goal is to find any of the marked vertices in $M$ using only evolutions that preserve the locality of $G$ on the vertex register, \ie, to perform a \emph{spatial search} on $G$~\cite{AaronsonA05} (here we define an even more restricted notion of locality than the ones in~\cite{AaronsonA05}, but it is more intuitive and sufficiently powerful for our purpose).

We allow two types of operations on the vertex register:
\begin{itemize}
  \item \emph{static transformations}, that can be conditioned on the state of the vertex register, but do not modify it;
  \item \shift{}, that exchanges the value of the vertex register and another register.
\end{itemize}
To impose locality, we want to restrict the execution of \shift{} only to the edges of $G$.

\begin{definition}[$\shift$ operation]\label{def:Shift}
Let
\begin{equation}
  \shift: (x,y) \mapsto
  \begin{cases}
    (y,x), & \text{if } (x,y) \in E, \\
    (x,y), & \text{otherwise.}
  \end{cases}
  \label{eq:Shift}
\end{equation}
In the first case we say that $\shift$ \emph{succeeds}, but in the second case it \emph{fails} (we assume that $\shift$ always succeeds if $x = y$).
\end{definition}

\begin{definition}[Search problems]
Under the restriction that only static transformations and $\shift$ are allowed, consider the following problems:
\begin{itemize}
  \item $\detect(G)$: Detect if there is a marked vertex in $G$;
  \item $\find(G)$: Find any marked vertex in $G$, with the promise that $M \neq \emptyset$.
\end{itemize}
We also define the following variations of the above problems:
\begin{itemize}
  \item $\detect^{(k)}(G)$: problem $\detect(G)$ with the promise that either $\m = 0$ or $\m = k$;
  \item $\find^{(k)}(G)$: problem $\find(G)$ with the promise that $\m = k$.
\end{itemize}
Similarly, let $\detect^{(\geq k)}(G)$ and $\find^{(\geq k)}(G)$ denote the corresponding problems with equality $m = k$ replaced by inequality $\m \geq k$.
\end{definition}

Note that an algorithm for $\find$ (or its variations) should output a marked element, but there is no additional constraint on its output. Our quantum algorithms will solve a slightly stronger version of $\find$, which we call $\sample$, where it is necessary to sample marked elements from a specific distribution (see \Sect{extended-hitting-time}).

%---------------------------------%
\subsection{Search via random walk}
%---------------------------------%

A natural approach to searching on a graph involves using a random walk. Intuitively, a random walk is an alternation of coin flips and shifts. More precisely, a coin is flipped according to the current state $x \in X$ of the vertex register, its value describes the target vertex $y$, and \shift{} performs a move from $x$ to $y$. Let $P_{xy}$ be the probability that $x$ is shifted to $y$. Then \shift{} always succeeds if $P_{xy} = 0$ whenever $(x,y) \notin E$. In such case, we say that $P = (P_{xy})_{x, y \in X}$ is a \emph{Markov chain on graph $G$}.

We assume from now on that $P$ is an ergodic Markov chain (see \Def{Ergodicity}). Therefore, by the \PFT{}, $P$ has a unique stationary distribution $\pi$. We also assume that $P$ is reversible: $\pi_x P_{xy} = \pi_y P_{yx}$, for all $x, y \in X$ (see \Def{Reversibility}).

To measure the complexity of implementing a random walk corresponding to $P$, we introduce the following black-box operations:
\begin{itemize}
  \item $\check(M)$: check if a given vertex is marked;
  \item $\setup(P)$: draw a sample from the stationary distribution $\pi$ of $P$;
  \item $\update(P)$: perform one step of $P$.
\end{itemize}
Each of these black-box operations have the corresponding associated implementation cost, which we denote by $\checkingcost$, $\setupcost$, and $\updatecost$, respectively.

%----------------------------------%
\subsection{Search via quantum walk} \label{sect:Quantum walk preliminaries}
%----------------------------------%

The setup in the quantum case is as follows. As in~\cite{KOR}, the evolution takes place in space $\Hi \x \Hi$ where $\Hi := \spn \set{\ket{x} : x \in X}$ is the $n$-dimensional complex Euclidean space spanned by elements of set $X$. The first register stores the current vertex of the walk and is called \emph{vertex register}. We call a unitary transformation \emph{static} if it is controlled by this register, \ie{}, it is of the form $\sum_{x \in X} \proj{x} \x U_x$ for some unitaries $U_x$. The quantum version of the $\shift$ operation is obtained by extending the expression in \Eq{Shift} by linearity.

A \emph{quantum walk} on $G$ is a composition of static unitary transformations and \shift{}. In addition, we require that it respects the local structure of $G$, \ie{}, whenever \shift{} is applied to a state, the state must completely lie within the subspace of $\Hi \x \Hi$ where \shift{} is guaranteed to succeed.

We will only consider quantum walks built from quantum analogues of reversible Markov chains, so we extend the operations $\check$, $\setup$, and $\update$ to the quantum setting as follows:
\begin{itemize}
  \item $\check(M)$: map $\ket{x} \ket{b}$ to $\ket{x} \ket{b}$ if $x \notin M$ and $\ket{x} \ket{b \oplus 1}$ if $x \in M$, where $\ket{x}$ is the vertex register and $b \in \set{0,1}$;
  \item $\setup(P)$: construct the superposition $\ket{\pi} := \sum_{x \in X} \sqrt{\pi_x} \ket{x}$;
  \item $\update(P)$: apply any of $V(P)$, $V(P)\ct$, or $\shift$, where
$V(P)$ is a unitary operation that satisfies
\begin{equation}
  V(P) \ket{x} \ancilla
  := \ket{x} \ket{p_x}
  := \ket{x} \sum_{y \in X} \sqrt{P_{xy}} \ket{y}
  \label{eq:V(P)}
\end{equation}
for all $x \in X$ and some fixed reference state $\ancilla \in \Hi$.
\end{itemize}
Implicitly, we also allow any controlled version of $\check(M)$, $\setup(P)$, and $\update(P)$, on which we access via oracle.

In terms of the number of applications of $\shift$, $\update$ has complexity $1$ while $\setup$ has complexity equal to the diameter of graph $G$. Nonetheless, in many algorithmic applications, the situation is more complex and the number of applications of $\shift$ is not the only relevant cost; see for instance~\cite{Ambainis04,MagniezSS05}.

To define a quantum analogue of a reversible Markov chain $P$, we follow the construction of Szegedy~\cite{Sze}. Let $\X : = \Hi \x \spn \set{\ancilla} = \spn \set{\ket{x} \ancilla : x \in X}$ and
\begin{equation}
  \reflex_\X
  := 2 \sum_{x \in X} \proj{x} \x \proj{\bar{0}} - I \x I
   = I \x (2 \proj{\bar{0}} - I)
  \label{eq:refX}
\end{equation}
be the reflection in $\Hi \x \Hi$ with respect to the subspace $\X$. The \emph{quantum walk operator} corresponding to Markov chain $P$ is\footnote{In~\cite{Sze} the quantum walk operator corresponding to $P$ is defined as $\bigl( V(P) \, W(P) \, V(P)^\dagger \bigr)^2$ where $W(P)$ is defined in \Eq{W(P)}.}
\begin{equation}
  W(P) := V(P)\ct \, \shift \, V(P) \cdot \reflex_{\X}.
  \label{eq:W(P)}
\end{equation}
Notice that $W(P)$ requires $3$ calls to $\update(P)$.

Since we always choose an initial state that lies in subspace $\X$, we can simplify the analysis by restricting the action of $W(P)$ to the smallest subspace that contains $\X$ and is invariant under $W(P)$. We call this subspace the \emph{walk space} of $W(P)$. We show in \Apx{W(s)} that this subspace is spanned by $\X$ and $W(P) \X$, and that $\shift$ is guaranteed to succeed when $W(P)$ is applied to a state in the walk space.

%---------------------------------%
\subsection{Classical hitting time} \label{sect:Classical hitting time}
%---------------------------------%

We define the hitting time of $P$ based on a simple classical random walk algorithm for finding a marked element in the state space $X$.

\begin{definition}\label{def:HT}
Let $P$ be an ergodic Markov chain, and $M$ be a set of marked states. The \emph{hitting time} of $P$ with respect to $M$, denoted by $\HT(P,M)$, is the expected number of executions of the last step of the \RWA{}, conditioned on the initial vertex being unmarked.
\begin{algobox}{0.9}
\RWAtxt{}\label{alg:RWA}
\begin{enumerate}
 \item Generate $x \in X$ according to the stationary distribution $\pi$ of $P$ using $\setup(P)$.
 \item Check if $x$ is marked using $\check(M)$. If $x$ is marked, output $x$ and exit.\label{step:check}
 \item Otherwise, update $x$ according to $P$ using $\update(P)$ and go back to \step{check}.\label{step:walk}
\end{enumerate}
\end{algobox}
\end{definition}

It is straightforward to bound the classical complexity of the $\detect$ and $\find$ problems in terms of the hitting time.

\begin{proposition}\label{prop:random-walk-search}
Let $k \geq 1$. $\detect^{(\geq k)}(G)$ can be solved with high probability and classical complexity of order
\begin{equation}
  \setupcost + T \cdot (\updatecost+\checkingcost),
  \quad\textrm{where}\quad
  T = \max_{\abs{M'}=k} \HT(P,M').
\end{equation}
$\find(G)$ can be solved with high probability and expected classical complexity of order
\begin{equation}
  \setupcost + T \cdot (\updatecost+\checkingcost),
  \quad\textrm{where}\quad
  T = \HT(P,M).
\end{equation}
\end{proposition}

Note that since the \RWA{} consists in applying the random walk $P$ until hitting a marked vertex, it may be seen as repeated applications of the \emph{absorbing} walk $P'$.
\begin{definition}
 Let $P$ be an ergodic Markov chain, and $M$ be a set of marked states. The \emph{absorbing} walk $P'$ is the walk obtained from $P$ by replacing all outgoing transitions from marked vertices by self-loops, that is $P'_{xy}=P_{xy}$ for all $x\notin M$, and $P'_{xy}=\delta_{xy}$ for all $x\in M$ ($\delta_{xy}$ being the Kronecker delta).
\end{definition}

The hitting time $\HT(P,M)$ may be obtained from the spectral properties of the \emph{discriminant matrix} of $P'$, which was introduced by Szegedy in~\cite{Sze,Sze-arXiv}.
\begin{definition}\label{def:Discriminant}
 The \emph{discriminant matrix} $D(P)$ of a Markov chain $P$ is
\begin{equation}
  D(P) := \sqrt{P \circ P\tp},
  \label{eq:Discriminant}
\end{equation}
where the Hadamard product ``$\circ$'' and the square root are computed entry-wise. 
\end{definition}

\begin{restatable}{proposition}{hittingtime}\label{prop:HT(P,M)}
The hitting time of Markov chain $P$ with respect to marked set $M$ is given by
\begin{equation}
  \HT(P,M) = \sum_{k=1}^{n-\abs{M}}\frac{\abs{\braket{v_k'}{U}}^2}{1-\lambda'_k},
  \label{eq:HT(P,M)}
\end{equation}
where $\lambda'_k$ are the eigenvalues of the discriminant matrix $D'=D(P')$ in nondecreasing order, $\ket{v_k'}$ are the corresponding eigenvectors, and $\ket{U}$ is the unit vector
\begin{equation*}
 \ket{U}:=\frac{1}{\sqrt{1-p_M}}\sum_{x\notin M}\sqrt{\pi_x}\ket{x},
\end{equation*}
$p_M$ being the probability to draw a marked vertex from the stationary distribution $\pi$ of $P$.
\end{restatable}
This proposition is proved in \Apx{HT}.

%-------------------------------%
\subsection{Quantum hitting time}
%-------------------------------%

Quantum walks have been successfully used for detecting the presence of marked vertices quadratically faster than random walks~\cite{Sze}. Nonetheless, very little is known about the problem of finding a marked vertex. Below, we describe the current understanding of this problem.

\begin{theorem}[\cite{Sze}]\label{thm:Detection}
Let $k\geq 1$. $\detect^{(\geq k)}(G)$ can be solved with high probability and quantum complexity of order
\begin{equation}
  \setupcost + T \cdot (\updatecost+\checkingcost),
  \quad\textrm{where}\quad
  T = \max_{\abs{M'}=k} \sqrt{\HT(P,M')}.
\end{equation}
\end{theorem}

When $P$ is state-transitive and there is a unique marked vertex $z$ (\ie, $\m = 1$), $\HT(P,\{z\})$ is independent of $z$ and one can also find $z$:

\begin{theorem}[\cite{Tulsi,MNRS2}]\label{thm:Finding}
Assume that $P$ is state-transitive. $\find^{(1)}(G)$ can be solved with high probability and quantum complexity of order
\begin{equation}
  \setupcost + T \cdot (\updatecost+\checkingcost),
  \quad\textrm{where}\quad
  T = \sqrt{\HT(P,\set{z})}.
\end{equation}
\end{theorem}

Using standard techniques, such as in~\cite{AaronsonA05}, \Thm{Finding} can be generalized to any number of marked vertices, with an extra logarithmic multiplicative factor. Nonetheless, the complexities of the corresponding algorithms do not decrease when the size of $M$ increases, contrary to the random walk search algorithm (\Prop{random-walk-search}) and the quantum walk detecting algorithm (\Thm{Detection}).

\begin{corollary}\label{cor:tulsimulti}
Assume that $P$ is state-transitive. $\find(G)$ can be solved with high probability and quantum complexity of order
\begin{equation}
  \log(n) \cdot \bigl( \setupcost + T \cdot (\updatecost + \checkingcost) \bigr),
  \quad\textrm{where}\quad
  T = \sqrt{\HT(P,\set{z})}, \text{ for any $z$.}
\end{equation}
\end{corollary}

%--------------------------------%
\subsection{Extended hitting time} \label{sect:extended-hitting-time}
%--------------------------------%

The quantum algorithms leading to the results in the previous subsection are based on quantum analogues of either the Markov chain $P$ or the corresponding absorbing walk $P'$. However, the algorithms proposed in the present article are based on a quantum analogue of the following \emph{interpolated} Markov chain.
\begin{definition}
 Let $P$ be a Markov chain, $M$ be a set of marked elements and $P'$ be the corresponding absorbing walk. We define the \emph{interpolated} Markov chain $P(s)$ as
\begin{align*}
 P(s):=(1-s)P+sP', \quad 0\leq s\leq 1.
\end{align*}
We also denote by $D(s)$ the discriminant matrix $D(P(s))$, by $\lambda_k(s)$ (for $k\in[n]$) its eigenvalues (in nondecreasing order) and by $\ket{v_k(s)}$ (for $k\in[n]$) its corresponding eigenvectors.
\end{definition}
Some properties of $P(s)$ are proven in \Apx{Semi-absorbing basics}, in particular, we note that $P(s)$ is ergodic for any $0\leq s<1$ as soon as $P$ is (\Prop{Ergodicity}). Moreover, just as $P(s)$ interpolates between $P$ and $P'$, the stationary distribution $\pi(s)$ of $P(s)$ interpolates between the stationary distribution $\pi$ of $P$ and its restriction to the set of marked vertices, i.e. a stationary distribution for $P'$ (\Prop{Stationary}).

This implies that $P(s)$ may be used to solve the following strong version of the $\find$ problem.

\begin{definition}[Sampling problem]
Let $P$ be an ergodic Markov chain on graph $G$. Under the restriction that only static transformations and $\shift$ are allowed, consider the following problems:
\begin{itemize}
  \item $\sample(P)$: Sample marked vertices in $G$ according to the restriction to set $M$ of the stationary distribution of $P$, with the promise that $M \neq \emptyset$.
  \item $\sample^{(k)}(P)$: problem $\sample(P)$ with the promise that $\m = k$.
\end{itemize}
\end{definition}

Indeed, since the stationary distribution of $P(s)$ precisely interpolates between $\pi$ and its restriction to $M$, we can solve the $\sample$ problem by applying Markov chain $P(s)$ for a sufficient number of steps $t$ to approach its stationary distribution, then outputting the current vertex if it is marked, otherwise starting over.

Our new quantum algorithms can be seen as quantum analogues of this classical algorithm, and their cost will be expressed in terms of a quantity which we call the \emph{extended} hitting time.

\begin{restatable}{definition}{deflimHT}\label{def:HT(s)}
The \emph{extended hitting time} of $P$ with respect to $M$ is
\begin{equation}
  \limHT := \lim_{s \to 1} \HT(s),
  \label{eq:HT lim}
\end{equation}
where the \emph{interpolated hitting time} $\HT(s)$ is defined for any $s\in[0,1)$\footnote{Note that in the case of multiple marked elements this expression cannot be used for $s = 1$, since the numerator and denominator vanish for terms with $k > n - \abs{M}$. We analyze the $s \to 1$ limit in \Apx{HT comparison}.} as
\begin{equation}
  \HT(s) := \sum_{k=1}^{n-1} \frac{\abs{\braket{v_k(s)}{U}}^2}{1-\lambda_k(s)}.
  \label{eq:HT(s) series}
\end{equation}
\end{restatable}

The name \emph{extended hitting time} is justified by comparing \Eq{HT(s) series} to \Eq{HT(P,M)}, and noting that $\braket{v_k'}{U}=0$ for $k>n-\abs{M}$. In general, the extended hitting time $\limHT$ can be larger than the hitting time $\HT(P,M)$, but they happen to be equal in the case of a single marked element. This implies that when $\abs{M}=1$, the cost of our quantum algorithms can be expressed in terms of the usual hitting time, which might be attributed to the fact that the $\sample$ problem is equivalent to the usual $\find$ problem in that case.

\begin{restatable}{proposition}{htcontinuity}\label{prop:Continuity}
If $\abs{M} = 1$ then $\limHT = \HT(P,M)$. However, there exists $P$ and $\abs{M} > 1$ such that $\limHT > \HT(P,M)$.
\end{restatable}
This proposition is proved in \Apx{HT(s)}. An alternative expression for $\limHT$ is provided in \Apx{HT comparison}; it allows for an easier comparison with $\HT(P,M)$. The following theorem holds for any number of marked elements and it relates $\HT(s)$ to $\limHT$.
\begin{restatable}{theorem}{HTintermsoflimHT}\label{thm:HT}
For $s < 1$, the interpolated hitting time $\HT(s)$ is related to $\limHT$ from \Eq{HT lim} as follows:
\begin{equation}
  \HT(s) = \frac{p_M^2}{(1-s(1-p_M))^2} \limHT
  \label{eq:HT(s)}
\end{equation}
where $p_M$ is the probability to pick a marked state from the stationary distribution $\pi$ of $P$. When $\abs{M} = 1$, $\limHT$ in \Eq{HT(s)} can be replaced by $\HT(P,M)$.
\end{restatable}

The proof is provided in \Apx{relation between HT(s) and HT(1)}.

%%%%%%%%%%%%%%%%%%%%%%%%%%%%%%%%%%%
\section{Quantum search algorithms} \label{sect:Search algorithms}
%%%%%%%%%%%%%%%%%%%%%%%%%%%%%%%%%%%

In this section we provide several quantum search algorithms. They are all based on a procedure known as \emph{eigenvalue estimation} and essentially run it different numbers of times with different values of parameters. Here is a formal statement of what eigenvalue estimation does.

\begin{theorem}[Eigenvalue estimation~\cite{Kitaev95,CEMM98}]\label{thm:eigenvalue-estimation}\label{alg:QEE}
For any unitary operator $W$ and precision $t \in \N$, there exists a quantum circuit \QEE$(W,t)$ that uses $2^t$ calls to the controlled-$W$ operator and $\Order(t^2)$ additional gates, and acts on eigenstates $\ket{\Psi_k}$ of $W$ as
\begin{equation}
  \ket{\Psi_k} \mapsto
  \ket{\Psi_k} \frac{1}{2^t}
  \sum_{l,m=0}^{2^t-1}
  e^{-\frac{2 \pi i l m}{2^t}}
  e^{i \varphi_k l} \ket{m},
  \label{eq:QEE}
\end{equation}
where $e^{i \varphi_k}$ is the eigenvalue of $W$ corresponding to $\ket{\Psi_k}$.
\end{theorem}

By linearity, \QEE$(W,t)$ resolves any state as a linear combination of the eigenstates of $W$ and attaches to each term a second register holding an approximation of the first $t$ bits of the binary decomposition of $\frac{1}{2\pi} \varphi_k$, where $\varphi_k$ is the phase of the corresponding eigenvalue. We will mostly be interested in the component along the eigenvector $\ket{\Psi_n}$ which corresponds to phase $\varphi_n = 0$. In that case, the second register is in the state $\ket{0^t}$ and the estimation is exact.

Our search algorithms will be based on \QEE$(W(s),t)$ for some values of parameters $s$ and $t$. Here, $W(s):=W(P(s))$ is the quantum analogue of the interpolated Markov chain $P(s)$, following Szegedy's construction as described in \Sect{Quantum walk preliminaries} (a quantum circuit implementing $W(s)$ is also provided by \Lem{Update} in \Apx{Circuit for W(s)}). The value of the interpolation parameter $s \in [0,1]$ will be related to $p_M$, the probability to pick a marked vertex from the stationary distribution $\pi$ of $P$. Precision $t \in \N$, or the number of binary digits in eigenvalue estimation, will be related to $\limHT$, the extended hitting time of $P$.

We consider several scenarios where different knowledge of the values of parameters $p_M$ and $\limHT$ is available, and for each case we provide an algorithm. The list of all results and the corresponding assumptions is given in \Tab{Search algorithms}.

Throughout the rest of this section we assume that all eigenvalues of $P$ are between $0$ and $1$. If this is not the case, we can guarantee it by making $P$ ``lazy'', which affects the hitting time only by a constant factor (see \Prop{Lazy}).

\newcommand{\known}{known}
\newcommand{\apprx}{approximation known}
\newcommand{\noinf}{not known}
\newcommand{\bound}{bound known}

\begin{table}
\begin{center}
\begin{tabular}{c|r|r}
Result & \multicolumn{1}{c|}{$p_M$} & \multicolumn{1}{c}{$\limHT$} \\
\hline
\Thm{Search with known pm and HT} & \known & \known \\
\Thm{Search with known HT}        & \apprx & \known \\
\Thm{Search with unknown HT}      & \apprx & \noinf \\
\Thm{Search with bound on pmin and HTmax}                   & \bound & \bound \\
\Thm{Search with bound on HTmax}  & \noinf & \bound \\
\end{tabular}
\end{center}
\caption[Summary of results on quantum search algorithms]{Summary of results on quantum search algorithms. Assumptions on $p_M$ and $\limHT$ are listed in the last two columns.}
\label{tab:Search algorithms}
\end{table}

%----------------------------------------------------%
\subsection{Algorithm with known values of
\texorpdfstring{$p_M$ and $\limHT$}{pM and HT(P,M)}} \label{sect:Main algorithm}
%----------------------------------------------------%

For simplicity, let us first assume that the values of $p_M$ and $\limHT$ are known. In this case we provide a quantum algorithm that solves $\find(G)$ (\ie, outputs a marked vertex if there is any) with success probability and running time that depends on two parameters $\varepsilon_1$ and $\varepsilon_2$.

Let us first recall how the classical \RWA{} from \Sect{Classical hitting time} works. It starts with the stationary distribution $\pi$ of $P$ and applies the absorbing walk $P'$ until most of the probability is absorbed in marked vertices and thus the state is close to a stationary distribution of $P'$.

In the quantum case a natural starting state is $\ket{\pi} \ancilla = \ket{v_n(0)} \ancilla$, which is a stationary state $W(P)$ (see \Lem{W(s) spectrum} below). By analogy, we would like to end up in its projection onto marked vertices, namely $\ket{M} \ancilla$, where
\begin{equation*}
\ket{M}:=\frac{1}{\sqrt{p_M}}\sum_{x\in M}\ket{x}, 
\end{equation*}
which is also a stationary state for $W(P')$. However, at this point the analogy breaks down, since we do not want to apply $W(P')$ to reach the final state. The reason is that in many cases, including the 2D grid, every iteration of $W(P')$ on $\ket{\pi} \ancilla$ may remain far from $\ket{M} \ancilla$. Instead, our approach consists of quantizing a new random walk, namely an interpolation $P(s)$ between $P$ and $P'$. This technique is drastically different from the approach of~\cite{Tulsi,MNRS2} and, to our knowledge, new.

\begin{figure}[th]
  \centering
%  \input{fig-UM.tex}
% !TeX root = KMOR.tex

\begin{tikzpicture}
  [line width = 0.7pt,
   arc/.style = {-latex}]

  \def\a{28} % angle
  \def\r{3.0} % radius

  \draw [arc] (0,0) to (\r, 0); % x axis
  \draw [arc] (0,0) to (0, \r); % y axis
  \draw [arc] (0,0) to (\a:\r); % vector

  \draw[dashed] (\a:\r) -- ($(0,0)!(\a:\r)!(\r,0)$);
  \draw[dashed] (\a:\r) -- ($(0,0)!(\a:\r)!(0,\r)$);

  \def\c{0.15} % scaling factor for the angle

  \draw (\c*\r,0) to [out = 90, in = -90+\a] (\a:\c*\r);
  \draw (\a/2:0.5*\r)+(-0.12,-0.09) node {$\theta(s)$};

  \def\z{1.15} % scaling factor for labels

  \draw (\z*\r, 0) node {$\ket{U}$};
  \draw (0, \z*\r) node {$\ket{M}$};
  \draw (\a:\z*\r)+(0,0.2) node {$\ket{v_n(s)}$};

  \draw ($(0,0)!(\a:\r/2)!(\r,0)$)+(0,-0.5) node {$\cos \theta(s) \geq \sqrt{\varepsilon_1}$};
  \draw ($(0,0)!(\a:\r/2)!(0,\r)$)+(-1.7,0) node {$\sin \theta(s) \geq \sqrt{\varepsilon_1}$};

\end{tikzpicture}

  \caption[Vectors $\ket{U}$, $\ket{M}$, and $\ket{v_n(s)}$]{Vectors $\ket{U}$, $\ket{M}$, and $\ket{v_n(s)} = \cos \theta(s) \ket{U} + \sin \theta(s) \ket{M}$. We want to choose $s$ so that $\braket{U}{v_n(s)} = \cos \theta(s) \geq \sqrt{\varepsilon_1}$ and $\braket{M}{v_n(s)} = \sin \theta(s) \geq \sqrt{\varepsilon_1}$.}
  \label{fig:UM}
\end{figure}
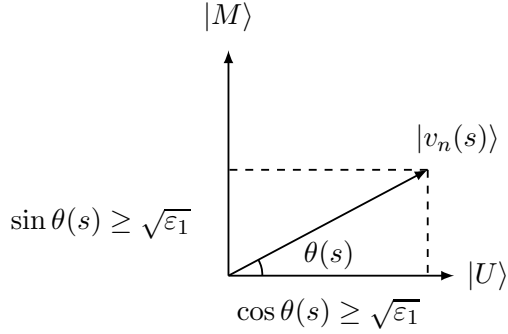

Intuitively, our quantum algorithm works as follows. We fix some value of $s \in [0,1]$ and map $\ket{U}$ to $\ket{v_n(s)}$ using a quantum walk based on $P(s)$, and then measure $\ket{v_n(s)}$ in the standard basis to get a marked vertex. For this to work with a good probability of success, we have to choose the interpolation parameter $s$ so that $\ket{v_n(s)}$ has a large overlap on both $\ket{U}$ and $\ket{M}$ (see \Fig{UM}). In that context, the following proposition, proved in \Apx{Rotation}, will be useful.

\begin{restatable}{proposition}{propvn}\label{prop:vn(s)}
$\ket{v_n(s)} = \cos \theta(s) \ket{U} + \sin \theta(s) \ket{M}$ where
\begin{align}
  \cos \theta(s) &= \sqrt{\frac{(1-s)(1-p_M)}{1-s(1-p_M)}}, &
  \sin \theta(s) &= \sqrt{\frac{p_M}{1-s(1-p_M)}}.
  \label{eq:cos and sin theta}
\end{align}
\end{restatable}

Therefore, for $\ket{v_n(s)}$ to have a large overlap on both $\ket{U}$ and $\ket{M}$, we will demand that $\cos \theta(s) \sin \theta(s) \geq \varepsilon_1$ for some parameter $\varepsilon_1$. A second parameter $\varepsilon_2$ controls the precision of phase estimation.

\begin{theorem}\label{thm:Search with known pm and HT}
Assume that the values of $p_M$ and $\limHT$ are known, and let $s \in [0,1)$, $T \geq 1$, and $\frac{1}{2} \geq \varepsilon_1 \geq \varepsilon_2 \geq 0$ be some parameters. If
\begin{align}
  \cos \theta(s) \sin \theta(s) \geq \varepsilon_1
  &&\text{and}&&
  T \geq \frac{\pi}{\sqrt{2} \varepsilon_2} \sqrt{\HT(s)}
  \label{eq:conditions}
\end{align}
where $\cos \theta(s)$ and $\sin \theta(s)$ are defined in \Eq{cos and sin theta} and $\HT(s)$ is the interpolated hitting time (see \Def{HT(s)}),
then \QWS$(P, M, s, \ceil{\log T})$ solves $\find(G)$ with success probability at least
\begin{equation}
  p_M + (1-p_M) (\varepsilon_1 - \varepsilon_2)^2
\end{equation}
and complexity of order $\setupcost + T \cdot (\updatecost + \checkingcost)$.
\end{theorem}

The proof of this theorem relies on the following lemma, originally due to Szegedy and proved in \Apx{Spectrum of W(s)}, which provides the spectral decomposition of the quantum walk operator $W(s)$ in terms of that of the discriminant $D(s)$.

\begin{restatable}[\cite{Sze}]{lemma}{lemWspectrum}\label{lem:W(s) spectrum}
Let $\mathcal{B}_k(s)$ for $k = 1, \dotsc, n$ be the subspaces from \Def{Subspaces}.
Assume that all eigenvalues $\lambda_k(s)$ of $D(s)$ are between $0$ and $1$, and let $\varphi_k(s) \in [0,\pi]$ be such that
\begin{equation}
  \lambda_k(s) = \cos \varphi_k(s).
  \label{eq:omega_k}
\end{equation}
Then $W(s)$ has the following eigenvalues and eigenvectors.
\begin{align}
& \text{On $\mathcal{B}_k(s)$:}
& & e^{\pm i \varphi_k(s)},
& \ket{\Psi^\pm_k(s)} &:= \frac{\ket{v_k(s), \bar{0}} \pm i \ket{v_k(s), \bar{0}}^\perp}{\sqrt{2}}.
  \label{eq:Bk} \\
& \text{On $\mathcal{B}_n(s)$:}
& & 1,
& \ket{\Psi_n(s)} &:= \ket{v_n(s), \bar{0}}.
  \label{eq:Bn}
\end{align}
In particular, $\bigcup_{k=1}^n \mathcal{B}_k(s)$ is the walk space of $W(s)$ and the remaining eigenvectors of $W(s)$ lie in the orthogonal complement $\mathcal{B}^\perp(s)$.
\end{restatable}

\newcommand{\psuccess}{q}%{p_{\text{success}}}

We can now prove \Thm{Search with known pm and HT}.

\begin{proof}[Proof of \autoref{thm:Search with known pm and HT}]
Let $t = \ceil{\log T}$ be the precision in the eigenvalue estimation. Our algorithm uses two registers: $\Reg_1$ and $\Reg_2$ with underlying state space $\Hi$ each. Occasionally we will attach the third register $\Reg_3$ initialized in $\ket{0} \in \C^2$ to check if the current vertex is marked.

\begin{algobox}{0.9}
\QWStxt$(P,M,s,t)$\label{alg:QWS}
\begin{enumerate}
  \item Use $\setup(P)$ to prepare the state $\ket{\pi} \ancilla$.\label{step:Preparation}
  \item Attach $\Reg_3$, apply $\Check(M)$ to $\Reg_1 \Reg_3$, and measure $\Reg_3$.
  \item If $\Reg_3 = 1$, measure $\Reg_1$ (in the vertex basis) and output the outcome.
  \item Otherwise, discard $\Reg_3$ and:\label{step:4}
  \begin{enumerate}
    \item Apply \QEE$(W(s),t)$ on $\Reg_1 \Reg_2$.\label{step:QEE}
    \item Attach $\Reg_3$, apply $\Check(M)$ to $\Reg_1 \Reg_3$, and measure $\Reg_3$.
    \item If $\Reg_3 = 1$, measure $\Reg_1$ (in the vertex basis) and output the outcome.
          Otherwise, output: \texttt{No marked vertex}.\label{step:Check}
  \end{enumerate}
\end{enumerate}
\end{algobox}

Notice that \step{Preparation} has complexity $\setupcost$, but \QEE$(W(s),t)$ in \step{QEE} has complexity of the order $2^t \cdot (\updatecost + \checkingcost)$ according to \Thm{eigenvalue-estimation} and \Lem{Update}. Thus, the total complexity is of the order $\setupcost + T \cdot (\updatecost + \checkingcost)$, and it only remains to bound the success probability.

Observe that the overall success probability is of the form $p_M + (1-p_M) \psuccess$ where $\psuccess$ is the probability to find a marked vertex in \step{4}. Thus, it remains to show that $\psuccess \geq (\varepsilon_1 - \varepsilon_2)^2$.

We assume that \QWS$(P,M,s,t)$ reaches \step{QEE}, otherwise a marked vertex is already found. At this point the state is $\ket{U} \ancilla$. Let us expand the first register of this state in the eigenbasis of the discriminant matrix $D(s)$. From now on we will omit the explicit dependence on $s$ when there is no ambiguity. Let
\begin{equation}
  \alpha_k := \braket{v_k}{U}
\end{equation}
and observe from \Lem{W(s) spectrum} that $\ket{v_k} \ancilla = \frac{1}{\sqrt{2}} (\ket{\Psi^+_k} + \ket{\Psi^-_k})$. Then
\begin{equation}
  \ket{U} \ancilla
  = \alpha_n \ket{v_n} \ancilla + \sum_{k=1}^{n-1} \alpha_k \ket{v_k} \ancilla
  = \alpha_n \ket{\Psi_n} + \frac{1}{\sqrt{2}} \sum_{k=1}^{n-1} \alpha_k
    \bigl( \ket{\Psi^+_k} + \ket{\Psi^-_k} \bigr).
\end{equation}
According to \Lem{W(s) spectrum}, the eigenvalues corresponding to $\ket{\Psi_n}$ and $\ket{\Psi^{\pm}_k}$ are $1$ and $e^{\pm i \varphi_k}$, respectively. From \Eq{QEE} we see that \QEE$(W(s),t)$ in \step{QEE} acts as follows:
\begin{align}
  \ket{\Psi_n} & \mapsto \ket{\Psi_n} \ket{0^t}, \\
  \ket{\Psi^{\pm}_k} & \mapsto \ket{\Psi^{\pm}_k} \ket{\xi^{\pm}_k},
\end{align}
where $\ket{\xi^{\pm}_k}$ is a $t$-qubit state that satisfies
\begin{equation}
  \braket{0^t}{\xi^{\pm}_k}
  = \frac{1}{2^t} \sum_{l=0}^{2^t-1} e^{\pm i \varphi_k l}
  =: \delta^{\pm}_k.
  \label{eq:delta_k}
\end{equation}
Thus, the state after eigenvalue estimation lies in $\Hi \x \Hi \x \C^{2^t}$ and is equal to
\begin{equation}
  \ket{\Phi} :=
  \alpha_n \ket{\Psi_n} \ket{0^t}
+ \frac{1}{\sqrt{2}} \sum_{k=1}^{n-1} \alpha_k
  \bigl( \ket{\Psi^+_k} \ket{\xi^+_k} + \ket{\Psi^-_k} \ket{\xi^-_k} \bigr).
\end{equation}

Recall that $\psuccess$ denotes the probability to obtain a marked vertex by measuring the first register of $\ket{\Phi}$ in \step{Check}. To lower bound $\psuccess$, we require that the last register of $\ket{\Phi}$ is in the state $\ket{0^t}$ (\ie, the phase is estimated to be $0$). Then
\begin{align}
  \sqrt{\psuccess}
  &=    \norm{(\Pi_M \x I \x I) \ket{\Phi}} \\
  &\geq \norm{(\Pi_M \x I \x \proj{0^t}) \ket{\Phi}} \\
  &\geq \norm{\alpha_n (\Pi_M \x I) \ket{\Psi_n}}
      - \frac{1}{\sqrt{2}} \norm[\Big]{
          (\Pi_M \x I) \sum_{k=1}^{n-1} \alpha_k
          \bigl( \delta^+_k \ket{\Psi^+_k} + \delta^-_k \ket{\Psi^-_k} \bigr)
        } \\
  &\geq \norm{\alpha_n (\Pi_M \x I) \ket{\Psi_n}}
      - \frac{1}{\sqrt{2}} \norm[\Big]{
         \sum_{k=1}^{n-1} \alpha_k
           \bigl( \delta^+_k \ket{\Psi^+_k} + \delta^-_k \ket{\Psi^-_k} \bigr)
        }. \label{eq:Difference of norms}
\end{align}
From \Lem{W(s) spectrum} and \Prop{vn(s)} we know that $\ket{\Psi_n} = \ket{v_n} \ancilla = (\cos \theta \ket{U} + \sin \theta \ket{M}) \ancilla$. Hence, we find that $\alpha_n = \braket{v_n}{U} = \cos \theta$ and $\norm{(\Pi_M \x I) \ket{\Psi_n}} = \sin \theta$. Moreover, from \Lem{W(s) spectrum} we also know that vectors $\ket{\Psi^\pm_1}, \dotsc, \ket{\Psi^\pm_k}$ are mutually orthogonal. Thus we can simplify \Eq{Difference of norms} as follows:
\begin{equation}
  \sqrt{\psuccess}
  \geq \cos \theta \sin \theta
     - \sqrt{\sum_{k=1}^{n-1} \abs{\alpha_k}^2 \delta_k^2}
  \label{eq:psuccess-difference}
\end{equation}
where $\delta_k := \abs{\delta^+_k} = \abs{\delta^-_k}$ (note from \Eq{delta_k} that $\delta^+_k$ and $\delta^-_k$ are complex conjuagtes). Now we will bound the second term in \Eq{psuccess-difference}.

Let us compute the sum of the geometric series in \Eq{delta_k}:
\begin{equation}
  \delta_k^2
  = \abs*{\frac{1}{2^t} \sum_{l=0}^{2^t-1} e^{i \varphi_k l}}^2
  = \frac{1}{2^{2t}}
    \abs*{\frac{1 - e^{i \varphi_k 2^t}}
              {1 - e^{i \varphi_k    }}}^2
  = \frac{1}{2^{2t}}
    \abs*{\frac{e^{-i \frac{\varphi_k}{2} 2^t} - e^{i \frac{\varphi_k}{2} 2^t}}
              {e^{-i \frac{\varphi_k}{2}    } - e^{i \frac{\varphi_k}{2}    }}}^2
  = \frac{       \sin^2 (\frac{\varphi_k}{2} 2^t)}
         {2^{2t} \sin^2 (\frac{\varphi_k}{2}    )}.
\end{equation}
We can upper bound the numerator in the final expression by one. To bound the denominator, we use $\sin \frac{x}{2} \geq \frac{x}{\pi}$ for $x \in [0,\pi]$. Hence, we get
\begin{equation}
  \delta_k^2
  \leq \frac{\pi^2}{2^{2t} \varphi^2_k}
  \leq \frac{\pi^2}{T^2 \varphi^2_k}
  \label{eq:delta_k^2}
\end{equation}
since we chose $t = \ceil{\log T}$.

The interpolated hitting time is given by \Def{HT(s)}:
\begin{equation}
  \HT(s) = \sum_{k=1}^{n-1} \frac{\abs{\braket{v_k(s)}{U}}^2}{1 - \lambda_k(s)}.
\end{equation}
If we substitute $\braket{v_k(s)}{U} = \alpha_k(s)$ and $\lambda_k(s) = \cos \varphi_k(s)$ from Eqs.~\EqRef{delta_k} and~\EqRef{omega_k}, and omit the dependence on $s$, we get
\begin{equation}
  \HT(s)
     =   \sum_{k=1}^{n-1} \frac{\abs{\alpha_k}^2}{1 - \cos \varphi_k}
     =   \sum_{k=1}^{n-1} \frac{\abs{\alpha_k}^2}{2 \sin^2 (\frac{\varphi_k}{2})}
  \geq 2 \sum_{k=1}^{n-1} \frac{\abs{\alpha_k}^2}{\varphi^2_k}
  \label{eq:HT(s) bound}
\end{equation}
since $x \geq \sin x$ for $x \in [0,\pi]$.

By combining Eqs.~\EqRef{delta_k^2} and~\EqRef{HT(s) bound} we get
\begin{equation}
       \sum_{k=1}^{n-1} \abs{\alpha_k}^2 \delta_k^2
  \leq \sum_{k=1}^{n-1} \abs{\alpha_k}^2 \frac{\pi^2}{T^2 \varphi^2_k}
     = \frac{\pi^2}{T^2} \sum_{k=1}^{n-1} \frac{\abs{\alpha_k}^2}{\varphi^2_k}
  \leq \frac{\pi^2}{2} \frac{\HT(s)}{T^2}.
\end{equation}
Thus, \Eq{psuccess-difference} becomes
\begin{equation}
  \sqrt{\psuccess}
  \geq \cos \theta(s) \sin \theta(s)
     - \frac{\pi}{\sqrt{2}} \frac{\sqrt{\HT(s)}}{T}
  \geq \varepsilon_1 - \varepsilon_2,
  \label{eq:psuccess}
\end{equation}
where the last inequality follows from our assumptions. Thus $\psuccess \geq (\varepsilon_1 - \varepsilon_2)^2$.
\end{proof}

%-------------------------------------------------------------------------%
\subsection{Algorithms with approximately known \texorpdfstring{$p_M$}{pM}}
%-------------------------------------------------------------------------%

In this section we show that a good approximation $\pstar$ of $p_M$ suffices to guarantee that the constraint $\cos \theta(s) \sin \theta(s) \geq \varepsilon_1$ in \Thm{Search with known pm and HT} is satisfied. Our strategy is to make a specific choice of the interpolation parameter $s$, based on $\pstar$.

Intuitively, we want to choose $s$ so that $\cos \theta(s) \sin \theta(s)$ is large (recall \Fig{UM}), since this will increase the success probability according to \Eq{psuccess}, and make it easier to satisfy the constraint on $\varepsilon_1$ in \Thm{Search with known pm and HT}. The maximal value of $\cos \theta(s) \sin \theta(s)$ is achieved when $\sin \theta(s) = \cos \theta(s) = 1/\sqrt{2}$, and from \Eq{cos and sin theta} we get that the optimal value of $s$ as a function of $p_M$ is
\begin{equation}
  s(p_M) := 1 - \frac{p_M}{1-p_M}.
\end{equation}
Thus, when only an approximation $\pstar$ of $p_M$ is known, we will choose the interpolation parameter to be
\begin{equation}
  \sstar := s(\pstar) = 1 - \frac{\pstar}{1-\pstar}.
  \label{eq:sstar}
\end{equation}
Since we want $\sstar \geq 0$, we have to always make sure that $\pstar \leq 1/2$. In fact, from now we will also assume that $p_M \leq 1/2$. This is without loss of generality, since one can always prepare the initial state $\ket{\pi}$ at cost $\setupcost$ and measure it in the standard basis. If $p_M \geq 1/2$, this yields a marked vertex with probability at least $1/2$.

\begin{proposition}\label{prop:pstar}
If $p_M, \varepsilon_1 \in [0,\frac{1}{2}]$ and $\pstar$ satisfy
\begin{equation}
  2 \varepsilon_1 p_M \leq \pstar \leq 2 (1 - \varepsilon_1) p_M,
  \label{eq:pstar bounds}
\end{equation}
then $\cos \theta(\sstar) \sin \theta(\sstar) \geq \varepsilon_1$ where $\sstar := 1 - \frac{\pstar}{1-\pstar}$.
\end{proposition}

\begin{proof}
To get the desired result, we will show that the two inequalities in \Eq{pstar bounds} imply that $\cos^2 \theta(\sstar) \geq \varepsilon_1$ and $\sin^2 \theta(\sstar) \geq \varepsilon_1$, respectively, where
\begin{align}
  \cos^2 \theta(\sstar) &= \frac{(1 - p_M) \pstar}{p_M + \pstar - 2 p_M \pstar}, &
  \sin^2 \theta(\sstar) &= \frac{p_M (1 - \pstar)}{p_M + \pstar - 2 p_M \pstar}
  \label{eq:cos and sin sstar}
\end{align}
according to \Eq{cos and sin theta}.

From \Eq{cos and sin sstar}, we have $\sin^2 \theta(\sstar) \geq \varepsilon_1$ if and only if
\begin{equation}
  \pstar
  \leq \frac{(1 - \varepsilon_1) p_M}{\varepsilon_1 + p_M - 2 \varepsilon_1 p_M}.
  \label{eq:sin sstar}
\end{equation}
Since $p_M, \varepsilon_1 \leq 1/2$, the denominator is upper bounded as
\begin{equation}
  \varepsilon_1 + (1 - 2 \varepsilon_1) p_M
  \leq \varepsilon_1 + \frac{1 - 2 \varepsilon_1}{2}
  = \frac{1}{2}.
\end{equation}
Therefore, $\pstar \leq 2 (1 - \varepsilon_1) p_M$ implies \Eq{sin sstar}, which is equivalent to $\sin^2 \theta(\sstar) \geq \varepsilon_1$.

Similarly from \Eq{cos and sin sstar} we have $\cos^2 \theta(\sstar) \geq \varepsilon_1$ if and only if
\begin{equation}
  \pstar
  \geq \frac{\varepsilon_1 p_M}{1 - \varepsilon_1 - p_M + 2 \varepsilon_1 p_M},
  \label{eq:cos sstar}
\end{equation}
where the denominator is lower bounded as
\begin{equation}
  1 - \varepsilon_1 - (1 - 2 \varepsilon_1) p_M
  \geq 1 - \varepsilon_1 - \frac{1 - 2 \varepsilon_1}{2}
  = \frac{1}{2}.
\end{equation}
Therefore, $\pstar \geq 2 \varepsilon_1 p_M$ implies \Eq{cos sstar}, which is equivalent to $\cos^2 \theta(\sstar) \geq \varepsilon_1$.
\end{proof}

%........................................................%
\subsubsection{Known \texorpdfstring{$\limHT$}{HT(P,M)}}
%........................................................%

Now we will use \Prop{pstar} to show how an approximation $\pstar$ of $p_M$ can be used to make a specific choice of the parameters $\varepsilon_1$, $\varepsilon_2$, $s$, and $T$ in \Thm{Search with known pm and HT}, so that our quantum search algorithm succeeds with constant probability.

To be more specific, we assume that we have an approximation $\pstar$ of $p_M$ such that
\begin{equation}
  \abs{\pstar - p_M} \leq \frac{1}{3} p_M,
  \label{eq:approximation}
\end{equation}
where the constant $1/3$ is an arbitrary choice. Notice that
\begin{align}
  \frac{1}{3} p_M
  & \geq \pstar - p_M
    \quad \Longleftrightarrow \quad
    \frac{4}{3} p_M \geq \pstar, \\
  \frac{1}{3} p_M
  & \geq p_M - \pstar
    \quad \Longleftrightarrow \quad
    \pstar \geq \frac{2}{3} p_M,
\end{align}
so \Eq{approximation} is equivalent to
\begin{equation}
  \frac{2}{3} p_M \leq \pstar \leq \frac{4}{3} p_M.
  \label{eq:approximation2}
\end{equation}
If we are given such $\pstar$ and we choose $\sstar$ according to \Eq{sstar}, then our algorithm succeeds with constant probability if $T$ is sufficiently large.

\begin{theorem}\label{thm:Search with known HT}
Assume that we know the value of $\limHT$ and an approximation $\pstar$ of $p_M$ such that $\abs{\pstar - p_M} \leq p_M/3$. If $T \geq 14 \sqrt{\limHT}$ then \QWS$(P, M, \sstar, \ceil{\log T})$ solves $\find(G)$ with probability at least $1/36$ and complexity of order $\setupcost + T \cdot (\updatecost + \checkingcost)$.
\end{theorem}

\begin{proof}
We are given $\pstar$ that satisfies \Eq{approximation2}. This is equivalent to \Eq{pstar bounds} if we choose $\varepsilon_1 := 1/3$. Without loss of generality $p_M \leq 1/2$, so from \Prop{pstar} we get that $\cos \theta(\sstar) \sin \theta(\sstar) \geq \varepsilon_1$. Thus, the first condition in \Eq{conditions} of \Thm{Search with known pm and HT} is satisfied.

Next, we choose $\varepsilon_2 := 1/6$ somewhat arbitrarily. According to \Thm{HT}, $\HT(\sstar) \leq \limHT$. Thus
\begin{equation}
  \frac{\pi}{\sqrt{2}} \frac{1}{\varepsilon_2} \sqrt{\HT(\sstar)}
    \leq \pi \, 3 \sqrt{2} \sqrt{\limHT}
    \leq 14 \sqrt{\limHT}
    \leq T,
\end{equation}
so the second condition in \Eq{conditions} is also satisfied.

Hence, according to \Thm{Search with known pm and HT}, \QWS$(P, M, \sstar, \ceil{\log T})$ solves $\find(G)$ with success probability at least
\begin{equation}
  p_M + (1 - p_M) (\varepsilon_1 - \varepsilon_2)^2
  \geq (\varepsilon_1 - \varepsilon_2)^2
  = \biggl( \frac{1}{3} - \frac{1}{6} \biggr)^2
  = \frac{1}{36}
\end{equation}
and complexity of order $\setupcost + T \cdot (\updatecost + \checkingcost)$.
\end{proof}

%..........................................................%
\subsubsection{Unknown \texorpdfstring{$\limHT$}{HT(P,M)}}
%..........................................................%

Recall from \Thm{Search with known HT} in previous section that a marked vertex can be found if $\pstar$, an approximation of $p_M$, and $\limHT$ are known. In this section we show that a marked vertex can still be found (with essentially the same expected complexity), even if the requirement to know $\limHT$ is relaxed.

\begin{theorem}\label{thm:Search with unknown HT}
Assume that we are given $\pstar$ such that $\abs{\pstar - p_M} \leq p_M/3$, then \QWR$(P,M,\sstar,50)$ solves $\find(G)$ with expected quantum complexity of order
\begin{equation}
  \log(T) \cdot \setupcost + T \cdot (\updatecost + \checkingcost),
  \quad\textrm{where}\quad
  T = \sqrt{\limHT}.
\end{equation}
\end{theorem}

\begin{proof}
The idea is to repeatedly use \QWS$(P,M,\sstar,t)$ with increasing accuracy of the eigenvalue estimation. We start with $t = 1$ and in every iteration increase it by one. Once $t$ is above some threshold $t_0$, any subsequent iteration outputs a marked element with probability that is at least a certain constant. To boost the success probability of the \QWS$(P,M,\sstar,t)$ subroutine, for each value of $t$ we call it $k = 50$ times.

\begin{algobox}{0.8}
\QWRtxt$(P,M,\sstar,k)$\label{alg:QWR}
\begin{enumerate}
 \item Let $t=1$.
 \item\label{step:loop} Call $k$ times \QWS$(P,M,\sstar,t)$.
 \item\label{step:end} If no marked vertex is found, set $t \leftarrow t+1$ and go back to \step{loop}.
\end{enumerate}
\end{algobox}

\newcommand{\pfail}{p_{\text{fail}}}

Let $t_0$ be the smallest integer that satisfies
\begin{equation}
  14 \sqrt{\limHT} \leq 2^{t_0}.
  \label{eq:t0}
\end{equation}
Assume that variable $t$ has reached value $t \geq t_0$, but \QWR$(P,M,\sstar,50)$ has not terminated yet. By \Thm{Search with known HT}, each execution of \QWS$(P,M,\sstar,t)$ outputs a marked vertex with probability at least $1/36$. Let $\pfail$ be the probability that none of the $k = 50$ executions in \step{loop} succeeds. Notice that
\begin{equation}
  \pfail \leq (1 - 1/36)^{50} \leq 1/4.
\end{equation}

Let us assume that \QWR$(P,M,\sstar,50)$ terminates with the final value of $t$ equal to $t_f$. Recall from \Thm{Search with known pm and HT} that \QWS$(P,M,\sstar,t)$ has complexity of order $\setupcost + 2^t \cdot (\updatecost + \checkingcost)$, so the expected complexity of \QWR$(P,M,\sstar,50)$ is of order
\begin{equation}
  N_1 \cdot \setupcost + N_2 \cdot (\updatecost + \checkingcost),
  \label{eq:N1 and N2}
\end{equation}
where $N_1$ is the expectation of $t_f$, and $N_2$ is the expectation of $2 + 4 + \dotsb + 2^{t_f}$.

To upper bound $N_1$, we assume that the first $t_0 - 1$ iterations fail. Since each of the remaining iterations fails with probability at most $\pfail$, we get
\begin{align}
  N_1 &\leq (t_0 - 1) + \sum_{t=t_0}^\infty \pfail^{1+(t-t_0)} \\
      &   = (t_0 - 1) + \frac{\pfail}{1-\pfail} \\
      &\leq (t_0 - 1) + \frac{1/4}{3/4} \\
      &\leq t_0.
\end{align}
We use the same strategy to upper bound $N_2$:
\begin{align}
  N_2 &\leq \sum_{t=1}^{t_0-1} 2^t + \sum_{t=t_0}^\infty \pfail^{1+(t-t_0)} 2^t \\
      &   = (2^{t_0} - 2) + \pfail \cdot \sum_{t=0}^\infty \pfail^t 2^{t+t_0} \\
      &\leq (2^{t_0} - 2) + \frac{1}{4} \cdot \sum_{t=0}^\infty
            \Bigl( \frac{1}{4} \cdot 2 \Bigr)^t \cdot 2^{t_0} \\
      &   = (2^{t_0} - 2) + \frac{1}{4} \cdot 2 \cdot 2^{t_0} \\
      &\leq 2 \cdot 2^{t_0}.
\end{align}
We plug the bounds on $N_1$ and $N_2$ in \Eq{N1 and N2} and get that the expected complexity is of order $t_0 \cdot \setupcost + 2^{t_0 + 1} \cdot (\updatecost + \checkingcost)$. Since $t_0$ satisfies \Eq{t0}, this concludes the proof.
\end{proof}

%--------------------------------------------------%
\subsection{Algorithms with a given bound on
\texorpdfstring{$p_M$ or $\limHT$}{pM or HT(P,M)}}
%--------------------------------------------------%

In previous section, we considered the case when we know a \emph{relative} approximation of $p_M$, \ie, a value $\pstar$ such that $\abs{\pstar - p_M} \leq p_M/3$. In this section, we consider the case when we are given an \emph{absolute} lower bound $p_{\min}$ such that $p_{\min} \leq p_M$, or an \emph{absolute} upper bound $\HT_{\max} \geq \limHT$, or both. In particular, for problem $\find(G)^{(\geq k)}$ we can set $p_{\min} := \min_{M':\abs{M'}=k} p_{M'}$ and $\HT_{\max} := \max_{M':\abs{M'}\geq k} \HT^+(P,M')$.

%............................................................%
\subsubsection{Assuming a bound on \texorpdfstring{$p_M$}{pM}}
%............................................................%

\begin{theorem}\label{thm:Search with bound on pmin and HTmax}
Assume that we are given $p_{\min}$ such that $p_{\min} \leq p_M$, $\find(G)$ can be solved with expected quantum complexity of order
\begin{equation}
  \sqrt{\log(1/p_{\min})} \cdot
  \bigl[ \log(T) \cdot \setupcost + T \cdot (\updatecost + \checkingcost) \bigr],
  \quad\textrm{where}\quad
  T = \sqrt{\limHT}.
\end{equation}
Moreover, given $\HT_{\max}$ such that $\HT_{\max} \geq \limHT$, we can solve  $\find(G)$ with quantum complexity of order
\begin{equation}
  \sqrt{\log(1/p_{\min})} \cdot
  \bigl[ \setupcost + T \cdot (\updatecost + \checkingcost) \bigr],
  \quad\textrm{where}\quad
  T = \sqrt{\HT_{\max}}.
\end{equation}
\end{theorem}

\begin{proof}
We prove the first part of the theorem. The second part is similar except one has to use \QWS$(P,M,\sstar,T)$ instead of \QWR$(P,M,\sstar,50)$.

To apply \Thm{Search with unknown HT}, it is enough to obtain an approximation $\pstar$ of $p_M$ such that $\abs{\pstar - p_M} \leq p_M/3$. Recall from \Eq{approximation2} that this is equivalent to finding $\pstar$ such that
\begin{equation}
  \frac{2}{3} p_M \leq \pstar \leq \frac{4}{3} p_M.
  \label{eq:approximation2'}
\end{equation}
Let $l$ be the largest integer such that $p_M \leq 2^{-l}$. Then
\begin{equation}
  \frac{1}{2} \cdot 2^{-l} \leq p_M \leq 2^{-l}
\end{equation}
and hence
\begin{equation}
  \frac{2}{3} p_M
  \leq \frac{2}{3} \cdot 2^{-l}
     = \frac{4}{3} \cdot \biggl( \frac{1}{2} \cdot 2^{-l} \biggr)
  \leq \frac{4}{3} p_M.
\end{equation}
We can make sure that \Eq{approximation2'} is satisfied by choosing $\pstar := \frac{2}{3} \cdot 2^{-l}$. Unfortunately, we do not know the value of $l$. However, we know that $p_{\min} \leq p_M$ and without loss of generality we can assume that $p_M \leq 1/2$. Thus, it only suffices to check all values of $l$ from $1$ to $\floor{\log(1/p_{\min})}$.

To find a marked vertex, we replace \step{loop} in the \QWR{} algorithm by a loop over the $\floor{\log(1/p_{\min})}$ possible values of $\pstar$:
\begin{algobox}{0.5}
  For $l = 1$ to $\floor{\log(1/p_{\min})}$ do: \vspace{2pt}
  \begin{itemize}
    \item Let $\pstar := \frac{2}{3} \cdot 2^{-l}$.
    \item Call $k$ times \QWS$(P, M, s(\pstar), t)$.
  \end{itemize}
\end{algobox}
Recall from \Thm{Search with known pm and HT} that the complexity of \QWS$(P, M, \sstar, t)$ depends only on $t$. Hence, the analysis of the modified algorithm is the same, except that now the complexity of \step{loop} is multiplied by a factor of order $\log(1/p_{\min})$. In fact, this is the only non-trivial step of the \QWR{} algorithm, so the overall complexity increases by this multiplicative factor. Finally, note that instead of trying all possible values of $\pstar$, we can search for the right value using Grover's algorithm, following the approach of~\cite{HMdW03}, therefore reducing the multiplicative factor to $\sqrt{\log(1/p_{\min})}$.
\end{proof}

%......................................................................%
\subsubsection{Assuming a bound on \texorpdfstring{$\limHT$}{HT(P,M)}}
%......................................................................%

\begin{theorem}\label{thm:Search with bound on HTmax}
Assume that we are given $\HT_{\max}$ such that $\HT_{\max} \geq \limHT$, $\find(G)$ can be solved with expected quantum complexity of order
\begin{equation}
  \log(1/p_M) \cdot \bigl[ \setupcost + T \cdot (\updatecost + \checkingcost) \bigr],
  \quad\textrm{where}\quad
  T = \sqrt{\HT_{\max}}.
\end{equation}
\end{theorem}

\begin{proof}
We use \QWS$(P,M,\sstar,t)$ with $t = \ceil[\big]{\log \sqrt{\HT_{\max}}}$ and perform a dichotomic search for an appropriately chosen value of $\pstar$. This dichotomic search uses backtracking, since the branching in the dichotomy is with bounded error, similarly to the situation in~\cite{FeigeRPU94}.

Let us first describe the robust binary search of~\cite{FeigeRPU94}.
Let $x\neq 0^n$ be a $n$-bit string of $0$'s followed by some $1$'s.
An algorithm can only access $x$ by querying its bits as follows.
The answer to a query $i\in[n]$ to $x$ is a random and independent bit which takes value $x_i$ with probability at least $2/3$.

When there is no error, finding the largest $i$ such that $x_i=0$ can be done using the usual  binary search. Start with $a=1$ and $b=n$.
At each step, query $x_i$ with $i=\lceil (a+b)/2\rceil$. Then set $a=i$ if $x_i=0$, and $b=i$ otherwise.
The procedure stops when $x_a=0$ and $x_b=1$.

In our error model, the above algorithm can be made robust by adding a sanity check. 
Before querying $x_i$, bits $x_a$ and $x_b$ are also queried.
If one of the two answers is inconsistent, that is either the answer to query $a$ is $1$ or the answer to query $b$ is $0$,
the algorithm backtracks to the previous values of $a$ and $b$.
It is proven in~\cite{FeigeRPU94} that this procedure converges with expected time $\Theta(\log n)$ and outputs a correct value with high probability, say at least $2/3$.

For our problem, we are going to test each candidate value $\pstar$ using the following procedure for $k=50$.
\begin{algobox}{0.7}
\textbf{Test}$(P,M,\pstar,k)$\label{alg:test}
  \begin{enumerate}
    \item Call  $k$ times  \QWS$(P,M,s(\pstar),1)$;\\
    if a marked vertex is found, output it and stop.
    \item Call $k$ times  \QEE$(W(s(\pstar)),1)$;\\
    if a minority of $0$s is found output ``$p_M\le \pstar$'',\\
    else  output ``$p_M\ge \pstar$''.
  \end{enumerate}
\end{algobox}
This procedure satisfies the following:
\begin{itemize}
\item If $p_M/3\le\pstar\leq 2p_M/3$, then \textbf{Test}$(P,M,\pstar,50)$ outputs a marked element with probability at least $2/3$;
\item If $\pstar\leq p_M/3$, then \textbf{Test}$(P,M,\pstar,50)$ outputs ``$p_M\ge \pstar$'' with probability at most $2/3$;
\item If $\pstar\geq 2p_M/3$, then \textbf{Test}$(P,M,\pstar,50)$ outputs ``$p_M\le \pstar$'' with probability at most $2/3$.
\end{itemize}

Now we conduct a search similarly as in~\cite{FeigeRPU94}, starting with $a=0$ and $b=1$.
The only difference is that the search stops when a marked element is found.
At each step, we check the consistency of $a$ and $b$ by running \textbf{Test}$(P,M,a,50)$ and \textbf{Test}$(P,M,b,50)$.
If there is a contradiction, we backtrack to the previous values of $a$ and $b$. Otherwise
we conduct the dichotomy search by running \textbf{Test}$(P,M,\pstar,k)$ with $\pstar = (a + b)/2$
(in order to set either $a = \pstar$ or $b = \pstar$). 
The search stops when a marked element is found.

Our procedure behaves similarly to the one of~\cite{FeigeRPU94}.
Indeed, it converges even faster since it stops with probability at least $2/3$ when $\pstar\in[p_M/3,2p_M/3]$.
Therefore our procedure ends after $\Order(\log(1/p_M))$ expected iterations.
\end{proof}

%------------------------%
\section*{Acknowledgments}
%------------------------%

MO would like to acknowledge Andrew Childs for many helpful discussions. The authors would also like to thank Andris Ambainis for useful comments.
Part of this work was done while HK, MO, and JR were at NEC Laboratories America in Princeton. MO also was affiliated with University of Waterloo and Institute for Quantum Computing (supported by QuantumWorks) and IBM TJ Watson Research Center (supported by DARPA QUEST program under Contract No. HR0011-09-C-0047) during this project. Presently FM, MO and JR are supported by the European Union Seventh Framework Programme (FP7/2007-2013) under grant agreement no. 600700 (QALGO).
FM is also supported by the French ANR Blanc project ANR-12-BS02-005 (RDAM).
Last, JR acknowledges support from the Belgian ARC project COPHYMA.

 %%%%%%%
%%%%%%%%%
\appendix
%%%%%%%%%
 %%%%%%%

\newcommand{\rv}[2]{(#1\;\,#2)} % row vector
\newcommand{\Rv}[2]{\bigl(#1\;\,#2\bigr)} % row vector
\newcommand{\co}[2]{[#1,#2)}

\newcommand{\ofs}{} % {(s)}

%%%%%%%%%%%%%%%%%%%%%%%%%%%%%%%%%%%%%%
\section{Semi-absorbing Markov chains} \label{apx:Semi-absorbing}
%%%%%%%%%%%%%%%%%%%%%%%%%%%%%%%%%%%%%%

In this appendix we study a special type of Markov chains described by a one-parameter family $P(s)$ corresponding to convex combinations of $P$ and the associated absorbing chain $P'$. Intuitively, some states of $P(s)$ are hard to escape and the interpolation parameter $s$ controls how absorbing they are. For this reason we call such chains \emph{semi-absorbing}. In this appendix we consider various properties of semi-absorbing Markov chains as a function of the interpolation parameter $s$. The main result of this appendix is \Thm{HT} which is of central importance in \Sect{Search algorithms}.

We discussed some preliminaries on Markov chains and defined basic concepts such as ergodicity in \Sect{Random walks}. Here we begin by defining the interpolated Markov chain $P(s)$ and considering various its properties, such as the stationary distribution and reversibility (\Apx{Semi-absorbing basics}). We proceed by applying these concepts to define and study the discriminant matrix of $P(s)$ which encodes all relevant properties of $P(s)$, such as eigenvalues and the principal eigenvector, but has a much more convenient form (\Apx{Discriminant}). Finally, we define the hitting time $\HT$ and the interpolated hitting time $\HT(s)$ and relate the two in the case of a single marked element via \Thm{HT}, which is our main result regarding semi-absorbing Markov chains (\Apx{HT}).

Results from this appendix will be used in \Sect{Search algorithms} to construct quantum search algorithms based on discrete-time quantum walks.

%-----------------------------------------------------------%
\subsection{Basic properties of semi-absorbing Markov chains} \label{apx:Semi-absorbing basics}
%-----------------------------------------------------------%

Assume that a subset $M \subset X$ of size $m := \abs{M}$ of the states are marked (throughout this chapter we assume that $M$ is not empty).    (see~\cite[Chapter~III]{KemenySnell} and~\cite[Sect.~11.2]{GrinsteadSnell}). Note that $P'$ differs from $P$ only in the rows corresponding to the marked states (where it contains all zeros on non-diagonal elements, and ones on the diagonal). If we arrange the states of $X$ so that the unmarked states $U := X \setminus M$ come first, matrices $P$ and $P'$ have the following block structure:
\begin{align}
  P &:= \mx{
    P_{UU} & P_{UM} \\
    P_{MU} & P_{MM}
  }, &
  P' &:= \mx{
    P_{UU} & P_{UM} \\
    0      & I
  },
  \label{eq:P}
\end{align}
where $P_{UU}$ and $P_{MM}$ are square matrices of size $(n-\m) \times (n-\m)$ and $\m \times \m$, respectively, while $P_{UM}$ and $P_{MU}$ are matrices of size $(n-\m) \times \m$ and $\m \times (n-\m)$, respectively.

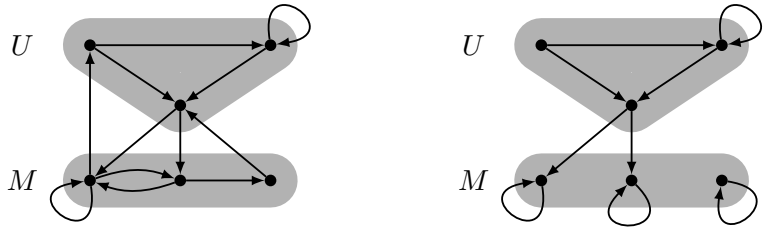
\begin{figure}[th]
  \centering
%  \input{fig-absorbing.tex}
% !TeX root = KMOR.tex

\begin{tikzpicture}
  [line width = 0.7pt,
   dot/.style = {circle, draw = black, fill = black, inner sep = 0mm, minimum size = 1.3mm},
   crc/.style = {circle, draw = black, minimum size = 5.0mm, inner sep = 1pt},
   arc/.style = {-latex}]

  \def\dl{45} % delta angle for loops
  \def\da{20} % delta angle for arcs

  \newcommand{\selfloop}[2]{
    \draw [arc] (#1) to [min distance = 10mm, out = #2 + \dl, in = #2 - \dl] (#1);
  }

  \def\w{12mm} % width
  \def\h{ 8mm} % height
  \def\d{10mm} % distance between U and M

  \newcommand{\drawgraph}[1]{
    % Coordinates
    \coordinate (u0#1) at (0,0);
    \coordinate (u1#1) at (-\w,+\h);
    \coordinate (u2#1) at (+\w,+\h);
    \coordinate (m1#1) at (-\w,-\d);
    \coordinate (m2#1) at (  0,-\d);
    \coordinate (m3#1) at (+\w,-\d);
    % Marked and unmarked sets
    \draw [draw = black!30, line width = 0.6*\w, line join = round] (u0#1) -- (u1#1) -- (u2#1) -- cycle;
    \draw [draw = black!30, line width = 0.6*\w, line cap  = round] (m1#1) -- (m3#1);
    % Text
    \path (u1#1)+(-0.9,0.0) node (U) {$U$};
    \path (m1#1)+(-0.9,0.0) node (M) {$M$};
    % Vertices
    \foreach \s in {u0, u1, u2, m1, m2, m3} {
      \path (\s#1) node [dot] (v\s#1) {};
    }
    % Arcs
    \foreach \from / \to in {u1/u0, u1/u2, u0/m1, u0/m2, u2/u0} {
      \draw [arc] (v\from#1) to (v\to#1);
    }
    % Self-loops
    \selfloop{vu2#1}{ 90-36}
    \selfloop{vm1#1}{-90-36}
  }

  % Graph for matrix P

  \begin{scope}
    \drawgraph{}
    \draw [arc] (vm1) to (vu1);
    \draw [arc] (vm2) to (vm3);
    \draw [arc] (vm3) to (vu0);
    \draw [arc] (vm1) to [out =       \da, in = 180 - \da] (vm2);
    \draw [arc] (vm2) to [out = 180 + \da, in =     - \da] (vm1);
  \end{scope}

  % Graph for matrix P'

  \begin{scope}[xshift = 5*\w]
    \drawgraph{'}
    \selfloop{vm2'}{-90}
    \selfloop{vm3'}{-90+36}
  \end{scope}

\end{tikzpicture}

  \caption[Markov chain $P$ and the corresponding absorbing chain $P'$]{Directed graphs underlying Markov chain $P$ (left) and the corresponding absorbing chain $P'$ (right). Outgoing arcs from vertices in the marked set $M$ have been turned into self-loops in $P'$.}
  \label{fig:Absorbing}
\end{figure}

Recall that we have defined an \emph{interpolated} Markov chain that interpolates between $P$ and $P'$:
\begin{equation}
  P(s) := (1-s) P + s P', \quad 0 \leq s \leq 1.
  \label{eq:P(s)}
\end{equation}
This expression has some resemblance with adiabatic quantum computation where similar interpolations are usually defined for quantum Hamiltonians~\cite{Farhi}. Indeed, the interpolated Markov chain $P(s)$ was used in \cite{KOR} to construct an adiabatic quantum search algorithm. Note that $P(0) = P$, $P(1) = P'$, and $P(s)$ has the following block structure:
\begin{equation}
  P(s) = \mx{
    P_{UU}      & P_{UM} \\
    (1-s)P_{MU} & (1-s)P_{MM} + s I
  }.
  \label{eq:Block P(s)}
\end{equation}

\begin{proposition}
If $P$ is ergodic then so is $P(s)$ for $s \in \co{0}{1}$. $P(1)$ is not ergodic.
\label{prop:Ergodicity}
\end{proposition}

\begin{proof}
Recall from \Def{Ergodicity} that ergodicity of a Markov chain can be established just by looking at its underlying graph. A non-zero transition probability in $P$ remains non-zero also in $P(s)$ for $s \in \co{0}{1}$. Thus the ergodicity of $P$ implies that $P(s)$ is also ergodic for $s \in \co{0}{1}$. However, $P(1)$ is not irreducible, since states in $U$ are not reachable from $M$. Thus $P(1)$ is \emph{not} ergodic.
\end{proof}

\begin{proposition}\label{prop:P' to t}
$(P'^{\,t})_{UU} = P_{UU}^t$.
\end{proposition}

\begin{proof}
Let us derive an expression for $P'^{\,t}$, the matrix of transition probabilities corresponding to $t$ applications of $P'$. Notice that $\smx{a & b \\ 0 & 1} \smx{c & d \\ 0 & 1} = \smx{ac & ad + b \\ 0 & 1}$. By induction,
\begin{equation}
  P'^{\,t} = \mx{P_{UU}^t & \sum_{k=0}^{t-1} P_{UU}^k P_{UM} \\ 0 & I}.
  \label{eq:P't}
\end{equation}
When restricted to $U$, it acts as $P_{UU}^t$.
\end{proof}

\begin{proposition}[{\cite[Theorem~11.3, p.~417]{GrinsteadSnell}}]\label{prop:PUU limit}
If $P$ is irreducible then $\lim_{k \to \infty} P_{UU}^k = 0$.
\end{proposition}

Intuitively this means that the sub-stochastic process defined by $P_{UU}$ eventually dies out or, equivalently, that the unmarked states of $P'$ eventually get absorbed (by \Prop{P' to t}).

\begin{proof}
Let us fix an unmarked initial state $x$. Since $P$ is irreducible, we can reach a marked state from $x$ in a finite number of steps. Note that this also holds true for $P'$. Let us denote the smallest number of steps by $l_x$ and the corresponding probability by $p_x > 0$. Thus in $l := \max_x l_x$ steps of $P'$ we are guaranteed to reach a marked state with probability at least $p := \min_x p_x > 0$, independently of the initial state $x \in U$. Notice that the probability to still be in an unmarked state after $kl$ steps is at most $(1-p)^k$ which approaches zero as we increase $k$.
\end{proof}

\begin{proposition}[{\cite[Theorem~3.2.1, p.~46]{KemenySnell}}]\label{prop:Invertible}
If $P$ is irreducible then $I - P_{UU}$ is invertible.
\end{proposition}

\begin{proof}
Notice that
\begin{equation}
  (I - P_{UU}) \cdot (I + P_{UU} + P_{UU}^2 + \dotsb + P_{UU}^{k-1}) = I - P_{UU}^k
\end{equation}
and take the determinant of both sides. From \Prop{PUU limit} we see that $\lim_{k \to \infty} \det(I - P_{UU}^k) = 1$. By continuity, there exists $k_0$ such that $\det(I - P_{UU}^{k_0}) > 0$, so the determinant of the left-hand side is non-zero as well. Using multiplicativity of the determinant, we conclude that $\det(I - P_{UU}) \neq 0$ and thus $I - P_{UU}$ is invertible.
\end{proof}
In the Markov chain literature $(I - P_{UU})^{-1}$ is called the \emph{fundamental matrix} of $P$.

%.....................................%
\subsubsection{Stationary distribution} \label{apx:Stationary distribution}
%.....................................%

From now on let us demand that $P$ is ergodic. Then according to the \PFT{} it has a unique stationary distribution $\pi$ that is non-zero everywhere. Let $\pi_U$ and $\pi_M$ be row vectors of length $n-\m$ and $\m$ that are obtained by restricting $\pi$ to sets $U$ and $M$, respectively. Then
\begin{align}
  \pi  &  = \mx{\pi_U & \pi_M}, &
  \pi' & := \mx{  0_U & \pi_M}
  \label{eq:Stationary}
\end{align}
where $0_U$ is the all-zeroes row vector indexed by elements of $U$ and $\pi'$ satisfies $\pi' P' = \pi'$.

Let $p_M := \sum_{x \in M} \pi_x$\label{math:pM} be the probability to pick a marked element from the stationary distribution. In analogy to the definition of $P(s)$ in \Eq{P(s)}, let $\pi(s)$ be a convex combination of $\pi$ and $\pi'$, appropriately normalized:
\begin{equation}
  \pi(s) := \frac{(1-s) \pi + s \pi'}{(1-s) + s p_M}
          = \frac{1}{1 - s (1-p_M)} \mx{(1-s) \pi_U & \pi_M}.
  \label{eq:pi(s)}
\end{equation}

\begin{proposition}\label{prop:Stationary}
$\pi(s)$ is the unique stationary distribution of $P(s)$ for $s \in \co{0}{1}$. At $s = 1$ any distribution with support only on marked states is stationary, including $\pi(1)$.
\end{proposition}

\begin{proof}
Notice that
\begin{equation}
  (\pi - \pi') (P - P')
  = \mx{\pi_U & 0} \mx{0 & 0 \\ P_{MU} & P_{MM} - I}
  = 0
\end{equation}
which is equivalent to
\begin{equation}
  \pi P' + \pi' P = \pi P + \pi' P'.
  \label{eq:primes}
\end{equation}
Using this equation we can check that $\pi(s) P(s) = \pi(s)$ for any $s \in [0,1]$:
\begin{align}
 &  \bigl( (1-s) \pi + s \pi' \bigr)
    \bigl( (1-s) P + s P' \bigr) \\
 &= (1-s)^2 \pi P + (1-s)s (\pi P' + \pi' P) + s^2 \pi' P' \\
 &= (1-s)^2 \pi + (1-s)s (\pi + \pi') + s^2 \pi' \\
 &= \bigl( (1-s) \pi + s \pi' \bigr) \bigl( (1 - s) + s \bigl) \\
 &= (1-s) \pi + s \pi'.
\end{align}
Recall from \Prop{Ergodicity} that $P(s)$ is ergodic for $s \in \co{0}{1}$ so $\pi(s)$ is the unique stationary distribution by \PFT{}. Since $P'$ acts trivially on marked states, any distribution with support only on marked states is stationary for $P(1)$.
\end{proof}

%...........................%
\subsubsection{Reversibility} \label{apx:Reversibility}
%...........................%

\begin{definition}\label{def:Reversibility}
Markov chain $P$ is called \emph{reversible\index{Markov chain!reversible}} if it is ergodic and satisfies the so-called \emph{detailed balance condition}
\begin{equation}
  \forall x,y \in X: \pi_x P_{xy} = \pi_y P_{yx}
  \label{eq:Detailed balance}
\end{equation}
where $\pi$ is the unique stationary distribution of $P$.
\end{definition}

Intuitively this means that the net flow of probability in the stationary distribution between every pair of states is zero. Note that \Eq{Detailed balance} is equivalent to
\begin{equation}
  \diag(\pi) \, P = P\tp \diag(\pi) = \bigl( \diag(\pi) P \bigr)\tp
\end{equation}
where $\diag(\pi)$ is a diagonal matrix whose diagonal is given by vector $\pi$. Thus \Eq{Detailed balance} is equivalent to saying that matrix $\diag(\pi) P$ is symmetric.

\begin{proposition}\label{prop:Reversibility}
If $P$ is reversible then so is $P(s)$ for any $s \in [0,1]$. Hence, $P(s)$ satisfies the interpolated detailed balance equation
\begin{equation}
  \forall s \in [0,1], \, \forall x,y \in X: \pi_x(s) P_{xy}(s) = \pi_y(s) P_{yx}(s).
  \label{eq:Interpolated detailed balance}
\end{equation}
\end{proposition}

\begin{proof}
First, notice that the absorbing walk $P'$ is reversible\footnote{Strictly speaking, the definition of reversibility also includes ergodicity for the stationary distribution to be uniquely defined. However, we will relax this requirement for $P'$ since, by continuity, $\pi'$ is the natural choice of the ``unique'' stationary distribution.} since $\diag(\pi') P'$ is a symmetric matrix:
\begin{equation}
  \diag(\pi') P'
  = \mx{0 & 0 \\ 0 & \diag(\pi_M)}
    \mx{P_{UU} & P_{UM} \\ 0 & I}
  = \mx{0 & 0 \\ 0 & \diag(\pi_M)}
  = \diag(\pi').
\end{equation}
Next, notice that
\begin{equation}
  \diag(\pi - \pi') (P - P')
  = \mx{\diag(\pi_U) & 0 \\ 0 & 0} \mx{0 & 0 \\ P_{MU} & P_{MM} - I}
  = 0
\end{equation}
which gives us an analogue of \Eq{primes}:
\begin{equation}
  \diag(\pi') P + \diag(\pi) P' = \diag(\pi) P + \diag(\pi') P'.
  \label{eq:primes-diag}
\end{equation}
Here the right-hand side is symmetric due to reversibility of $P$ and $P'$, thus so is the left-hand side. Using this we can check that $P(s)$ is reversible:
\begin{align}
&  \diag \bigl( (1-s) \pi + s \pi' \bigr) \bigl( (1-s) P + s P' \bigr) \\
&= (1-s)^2 \diag(\pi) P + (1-s)s \bigl( \diag(\pi) P' + \diag(\pi') P \bigr) + s^2 \diag(\pi') P'
\end{align}
where the first and last terms are symmetric since $P$ and $P'$ are reversible, but the middle term is symmetric due to \Eq{primes-diag}.
\end{proof}

%------------------------------%
\subsection{Discriminant matrix} \label{apx:Discriminant}
%------------------------------%

Recall from \Def{Discriminant} that the \emph{discriminant matrix\index{Markov chain!discriminant matrix of}} of a Markov chain $P(s)$ is
\begin{equation}
  D(s) := \sqrt{P(s) \circ P(s)\tp},
  \label{eq:D(s)}
\end{equation}
where the Hadamard product ``$\circ$'' and the square root are computed entry-wise. This matrix was introduced by Szegedy in~\cite{Sze,Sze-arXiv}. We prefer to work with $D(s)$ rather than $P(s)$ since the matrix of transition probabilities is not necessarily symmetric while its discriminant matrix is.

\begin{proposition}\label{prop:Discriminant}
If $P$ is reversible then
\begin{align}
  D(s) &= \diag \bigl( \! \sqrt{\pi(s)} \, \bigr) \: P(s) \:
          \diag \bigl( \! \sqrt{\pi(s)} \, \bigr)^{-1},
  \quad \quad \forall s \in \co{0}{1};
  \label{eq:D(s) and P(s)}
  \\[5pt]
  D(1) &= \mx{
    \diag \bigl( \! \sqrt{\pi_U} \, \bigr) \: P_{UU} \:
    \diag \bigl( \! \sqrt{\pi_U} \, \bigr)^{-1} & 0 \\
    0 & I
  }.
  \label{eq:D(1)}
\end{align}
\end{proposition}

Here the square roots are also computed entry-wise and $M^{-1}$ denotes the matrix inverse of $M$. Notice that for $s \in \co{0}{1}$ the right-hand side of \Eq{D(s) and P(s)} is well-defined, since $P(s)$ is ergodic by \Prop{Ergodicity} and thus according to the \PFT{} has a unique and non-vanishing stationary distribution. However, recall from \Prop{Stationary} that $\pi(1)$ vanishes on $U$, so the right-hand side of \Eq{D(s) and P(s)} is no longer well-defined at $s = 1$. For this reason we have an alternative expression for $D(1)$.

\begin{proof}[Proof (of {\Prop{Discriminant}})]
For a reversible Markov chain $P$ the interpolated detailed balance condition in \Eq{Interpolated detailed balance} implies that $D_{xy}(s) = \sqrt{P_{xy}(s) P_{yx}(s)} = P_{xy}(s) \sqrt{\pi_x(s) / \pi_y(s)}$. This is equivalent to \Eq{D(s) and P(s)}.

At $s=1$ from \Eq{D(s)} we have:
\begin{equation}
  D(1) = \sqrt{P(1) \circ P(1)\tp}
       = \sqrt{\mx{P_{UU} \circ P_{UU}\tp & 0 \\ 0 & I}}
       = \mx{\sqrt{P_{UU} \circ P_{UU}\tp} & 0 \\ 0 & I}.
  \label{eq:D(1) circ}
\end{equation}
It remains to verify that the upper left block of $D(1)$ agrees with \Eq{D(1)}. Using \Eq{D(s)} we compute that
\begin{equation}
  D_{UU}(s)
  = \sqrt{P_{UU} \circ P_{UU}\tp}
  = D_{UU}(0)
  = \diag \bigl( \! \sqrt{\pi_U} \, \bigr) \: P_{UU} \:
    \diag \bigl( \! \sqrt{\pi_U} \, \bigr)^{-1}
  \label{eq:DUU}
\end{equation}
where the last equality follows from \Eq{D(s) and P(s)} at $s = 0$. Together with \Eq{D(1) circ} this gives us the desired expression in \Eq{D(1)}.
\end{proof}

%....................................%
\subsubsection{Spectral decomposition} \label{apx:D(s) spectrum}
%....................................%

Recall from \Eq{D(s)} that $D(s)$ is real and symmetric. Therefore, its eigenvalues are real and it has an orthonormal set of real eigenvectors. Let
\begin{equation}
  D(s) = \sum_{i=1}^n \lambda_i(s) \ket{v_i(s)} \bra{v_i(s)}
  \label{eq:D(s) spectrum}
\end{equation}
be the spectral decomposition of $D(s)$ with eigenvalues $\lambda_i(s)$ and eigenvectors\footnote{There is no need to use bra-ket notation at this point; nevertheless we adopt it since vectors $\ket{v_i(s)}$ later will be used as quantum states.} $\ket{v_i(s)}$. Moreover, let us arrange the eigenvalues so that
\begin{equation}
  \lambda_1(s) \leq \lambda_2(s) \leq \dots \leq \lambda_n(s).
\end{equation}

From now on we will assume that $P$ is reversible (and hence ergodic) without explicitly mentioning it. Under this assumption the matrices $P(s)$ and $D(s)$ are similar (see \Prop{Eigenvalues} below). This means that $D(s)$ essentially has the same properties as $P(s)$, but in addition it also admits a spectral decomposition with orthogonal eigenvectors. This will be very useful in \Apx{Spectrum of W(s)}, where we find the spectral decomposition of the quantum walk operator $W(s)$ in terms of that of $D(s)$, and use it to relate properties of $W(s)$ and $P(s)$.

\begin{proposition}\label{prop:Eigenvalues}
Assume $P$ is reversible. The matrices $P(s)$ and $D(s)$ are similar for any $s \in [0,1]$ and therefore have the same eigenvalues. In particular, the eigenvalues of $P(s)$ are real.
\end{proposition}

\begin{proof}
From \Eq{D(s) and P(s)} we see that the matrices $D(s)$ and $P(s)$ are similar for $s \in \co{0}{1}$. From \Eq{D(1)} we see that $D(1)$ is similar to $\tilde{P} := \smx{P_{UU}&0\\0&I}$. To verify that $\tilde{P}$ and $P(1) = \smx{P_{UU} & P_{UM} \\ 0 & I}$ are similar, let $M := \smx{P_{UU}-I & P_{UM} \\ 0 & I}$. One can check that $M P(1) M^{-1} = \tilde{P}$ where $M^{-1} = \smx{(P_{UU}-I)^{-1} & -(P_{UU}-I)^{-1} P_{UM} \\ 0 & I}$ exists, since $P_{UU} - I$ is invertible according to \Prop{Invertible}. By transitivity, $D(1)$ is also similar to $P(1)$.
\end{proof}

\begin{proposition}\label{prop:Multiplicity}
The largest eigenvalue of $D(s)$ is $1$. It has multiplicity $1$ when $s \in \co{0}{1}$ and multiplicity $m$ when $s = 1$. In other words,
\begin{align}
  \lambda_{n-1}(s) < \lambda_n(s) = 1&, \quad \forall s \in \co{0}{1}, \label{eq:lambdas} \\
  \lambda_{n-\m}(1) < \lambda_{n-\m+1}(1) = \dots = \lambda_n(1) = 1&.
\end{align}
\end{proposition}

\begin{proof}
Let us argue about $P(s)$, since it has the same eigenvalues as $D(s)$ by \Prop{Eigenvalues}. From the \PFT{} we have that $\forall i: \lambda_i(s) \leq 1$ and $\lambda_n(s) = 1$. In addition, by \Prop{Ergodicity} the Markov chain $P(s)$ is ergodic for any $s \in \co{0}{1}$, so $\forall i \neq n: \lambda_i(s) < 1$. Finally, note by \Eq{D(1)} that for $s = 1$ eigenvalue $1$ has multiplicity at least $m$. Recall from \Eq{DUU} that $D_{UU}(1)$ and $P_{UU}$ are similar. From \Prop{Invertible} we conclude that all eigenvalues of $P_{UU}$ are strictly less than $1$. Thus the multiplicity of eigenvalue $1$ of $D(1)$ is exactly $m$.
\end{proof}

%...................................%
\subsubsection{Principal eigenvector} \label{apx:Rotation}
%...................................%

Let us prove an analogue of \Prop{Stationary} for the matrix $D(s)$.

\begin{proposition}\label{prop:Principal eigenvector}
$\sqrt{\pi(s)\tp}$ is the unique $(+1)$-eigenvector of $D(s)$ for $s \in \co{0}{1}$. At $s = 1$ any vector with support only on marked states is a $(+1)$-eigenvector, including $\sqrt{\pi(1)\tp}$.
\end{proposition}

\begin{proof}
Since $P(s)$ is row-stochastic, $P(s) \: 1_X\tp = 1_X\tp$ where $1_X$\label{math:1X} is the all-ones row vector. Thus we can check that for $s \in \co{0}{1}$,
\begin{align}
  D(s) \sqrt{\pi(s)\tp}
  &= \diag \Bigl( \! \sqrt{\pi(s)} \, \Bigr) \: P(s) \:
     \diag \Bigl( \! \sqrt{\pi(s)} \, \Bigr)^{-1} \sqrt{\pi(s)\tp} \\
  &= \diag \Bigl( \! \sqrt{\pi(s)} \, \Bigr) \: P(s) \: 1_X\tp \\
  &= \diag \Bigl( \! \sqrt{\pi(s)} \, \Bigr) \: 1_X\tp \\
  &= \sqrt{\pi(s)\tp}. \label{eq:1-eigenvector of D(s)}
\end{align}
Uniqueness for $s \in \co{0}{1}$ follows by the uniqueness of $\pi(s)$ and \Prop{Eigenvalues}. For the $s = 1$ case, notice from \Eq{D(1)} that $D(1)$ acts trivially on marked elements and recall from \Eq{pi(s)} that $\pi(1) = \rv{0_U}{\pi_M} / p_M$.
\end{proof}

According to the above Proposition, for any $s \in [0,1]$ we can choose the principal eigenvector $\ket{v_n(s)}$ in the spectral decomposition of $D(s)$ in \Eq{D(s) spectrum} to be
\begin{equation}
  \ket{v_n(s)} := \sqrt{\pi(s)\tp}.
  \label{eq:vn(s)}
\end{equation}
We would like to have an intuitive understanding of how $\ket{v_n(s)}$ evolves as a function of $s$. Let us introduce some useful notation that we will also need later.

Let $0_U$ and $1_U$ (respectively, $0_M$ and $1_M$) be the all-zeros and all-ones row vectors of dimension $n-\m$ (respectively, $\m$) whose entries are indexed by elements of $U$ (respectively, $M$). Furthermore, let
\begin{align}
  \tilde{\pi}_U &:= \pi_U/(1-p_M), &
  \tilde{\pi}_M &:= \pi_M/p_M
  \label{eq:piUM}
\end{align}
be the normalized row vectors describing the stationary distribution $\pi$ restricted to unmarked and marked states. Let us also define the following unit vectors in $\R^n$:
\begin{align}
  \ket{U}
   &:= \sqrt{\rv{\tilde{\pi}_U}{0_M}\tp}
     = \frac{1}{\sqrt{1-p_M}}\sum_{x \in U} \sqrt{\pi_x} \ket{x}, \label{eq:U} \\
  \ket{M}
   &:= \sqrt{\rv{0_U}{\tilde{\pi}_M}\tp}
     = \frac{1}{\sqrt{p_M}}\sum_{x \in M} \sqrt{\pi_x} \ket{x}. \label{eq:M}
\end{align}
Then we can express $\ket{v_n(s)}$ as a linear combination of $\ket{U}$ and $\ket{M}$.
\propvn*

\begin{proof}
By substituting $\pi(s)$ from \Eq{pi(s)} into \Eq{vn(s)} we get
\begin{equation}
  \ket{v_n(s)}
  = \sqrt{\pi(s)\tp}
  = \sqrt{\frac{\Rv{(1-s)\pi_U}{\pi_M}\tp}{1 - s (1-p_M)}}
  = \sqrt{\frac{\Rv{(1-s)(1-p_M)\tilde{\pi}_U}{p_M\tilde{\pi}_M}\tp}{1 - s (1-p_M)}}
\end{equation}
which is the desired expression.
\end{proof}

Thus $\ket{v_n(s)}$ lies in the two-dimensional subspace $\spn \set{\ket{U}, \ket{M}}$ and is subject to a rotation as we change the parameter $s$ (see \Fig{Rotation}). In particular,
\begin{align}
  \ket{v_n(0)} &= \sqrt{1-p_M} \ket{U} + \sqrt{p_M} \ket{M}, &
  \ket{v_n(1)} &= \ket{M}.
  \label{eq:vn(0) and vn(1)}
\end{align}

\begin{figure}[th]
  \centering
%  \input{fig-rotation.tex}
% !TeX root = KMOR.tex

\begin{tikzpicture}
  [line width = 0.7pt,
   arc/.style = {-latex}]

  \def\a{8} % angle
  \def\r{3.0} % radius

  \draw [arc] (0,0) to (\r, 0); % x axis
  \draw [arc] (0,0) to (0, \r); % y axis
  \draw [arc] (0,0) to (\a:\r); % vector

  \def\c{0.8} % scaling factor for dashed circle

  \draw [arc, dashed] (\a:\c*\r) to [out = 90+\a, in = 0] (0, \c*\r);

  \def\z{1.15} % scaling factor for labels

  \draw (\z*\r, 0) node {$\ket{U}$};
  \draw (0, \z*\r)+(0.98,0) node {$\ket{M} = \ket{v_n(1)}$};
  \draw (\a:\z*\r)+(0,0.30) node {$\ket{v_n(0)}$};

\end{tikzpicture}

  \caption[Rotation of $\ket{v_n(s)}$ in a two-dimensional subspace]{As $s$ changes from zero to one, the evolution of the principal eigenvector $\ket{v_n(s)}$ corresponds to a rotation in the two-dimensional subspace $\spn \set{\ket{U}, \ket{M}}$.}
  \label{fig:Rotation}
\end{figure}
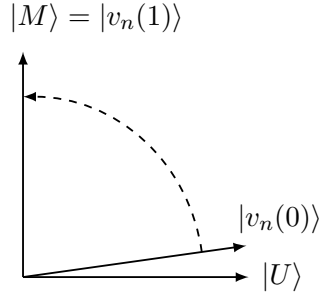

\begin{proposition}\label{prop:theta(s) dot}
$\theta(s)$ and its derivative $\dot{\theta}(s) := \frac{d}{ds} \theta(s)$ are related as follows:
\begin{equation}
  2 \dot{\theta}(s) = \frac{\sin \theta(s) \cos \theta(s)}{1-s}.
  \label{eq:theta(s) dot}
\end{equation}
\end{proposition}

\begin{proof}
Notice that
\begin{equation}
  \frac{d}{ds} \bigl( \sin^2 \theta(s) \bigr) = 2 \dot{\theta}(s) \sin \theta(s) \cos \theta(s).
\end{equation}
On the other hand, according to \Eq{cos and sin theta} we have
\begin{equation}
  \frac{d}{ds} \bigl( \sin^2 \theta(s) \bigr)
= \frac{d}{ds} \biggl( \frac{p_M}{1-s(1-p_M)} \biggr)
= \frac{p_M (1-p_M)}{(1-s(1-p_M))^2}
= \frac{\sin^2 \theta(s) \cos^2 \theta(s)}{1-s}.
\end{equation}
By comparing both equations we get the desired result.
\end{proof}

%........................%
\subsubsection{Derivative}
%........................%

\begin{proposition}\label{prop:D(s) dot}
$D(s)$ and its derivative $\dot{D}(s) := \frac{d}{ds} D(s)$ are related as follows:
\begin{equation}
  \dot{D}(s) = \frac{1}{2(1-s)} \bigl\{ \Pi_M, I-D(s) \bigr\}
\end{equation}
where $\{X,Y\} := XY + YX$ is the anticommutator of $X$ and $Y$, and $\Pi_M := \sum_{x \in M} \proj{x}$ is the projector onto the $m$-dimensional subspace spanned by marked states $M$.
\end{proposition}

\begin{proof}
Recall from \Eq{D(s)} that $D(s) = \sqrt{P(s) \circ P(s)\tp}$. The block structure of $P(s)$ is given in \Eq{Block P(s)}. First, let us derive an expression for $D_{MM}(s)$, the lower right block of $D(s)$:
\begin{align}
  D_{MM}(s)
  &= \sqrt{P_{MM}(s) \circ P_{MM}(s)\tp} \\
  &= \sqrt{\bigl( (1-s) P_{MM}    + s I \bigr) \circ
           \bigl( (1-s) P_{MM}\tp + s I \bigr)}.
\end{align}
Let us separately consider the diagonal and off-diagonal entries of $D_{MM}(s)$. For $x, y \in M$ we have
\begin{equation}
  D_{xy}(s) =
  \begin{cases}
    (1-s) \sqrt{P_{xy} P_{yx}} & \text{if $x \neq y$}, \\
    (1-s) P_{xx} + s & \text{if $x = y$}.
  \end{cases}
\end{equation}
Thus we can write $D_{MM}(s)$ as
\begin{equation}
  D_{MM}(s) = (1-s) \sqrt{P_{MM}\circ P\tp_{MM}} + s I.
\end{equation}

Expressions for the remaining blocks of $D(s)$ can be derived in a straightforward way. By putting all blocks together we get
\begin{equation}
  D(s) = \mx{
    \sqrt{P_{UU} \circ P\tp_{UU}} & \sqrt{(1-s) (P_{UM} \circ P\tp_{MU})} \\
    \sqrt{(1-s) (P_{MU} \circ P\tp_{UM})} & (1-s) \sqrt{P_{MM}\circ P\tp_{MM}} + s I }.
  \label{eq:D(s) blocks}
\end{equation}
When we take the derivative with respect to $s$ we find
\begin{equation}
  \dot{D}(s) = \mx{
    0 & -\frac{1}{2\sqrt{1-s}} \sqrt{P_{UM} \circ P\tp_{MU}} \\
    -\frac{1}{2\sqrt{1-s}} \sqrt{P_{MU} \circ P\tp_{UM}} & I - \sqrt{P_{MM}\circ P\tp_{MM}}}.
  \label{eq:D(s) dot}
\end{equation}

To relate $\dot{D}(s)$ and the original matrix $D(s)$, observe that
\begin{equation}
  \Pi_M D(s) + D(s) \Pi_M
  = \mx{
    0 & \sqrt{(1-s) (P_{UM} \circ P\tp_{MU})} \\
    \sqrt{(1-s) (P_{MU} \circ P\tp_{UM})} & 2 (1-s) \sqrt{P_{MM}\circ P\tp_{MM}} + 2 s I }
\end{equation}
which can be seen by overlaying the second column and row of $D(s)$ given in \Eq{D(s) blocks}. When we rescale this by an appropriate constant, we get
\begin{equation}
  - \frac{1}{2(1-s)} \{\Pi_M, D(s)\} = \mx{
      0 &
      -\frac{1}{2\sqrt{1-s}} \sqrt{P_{UM} \circ P\tp_{MU}} \\
      -\frac{1}{2\sqrt{1-s}} \sqrt{P_{MU} \circ P\tp_{UM}} &
      -\sqrt{P_{MM}\circ P\tp_{MM}} - \frac{s}{1-s} I
    }.
\end{equation}
This is very similar to the expression for $\dot{D}(s)$ in \Eq{D(s) dot}, except for a slightly different coefficient for the identity matrix in the lower right corner. We can correct this by adding $\Pi_M$ with an appropriate constant: $-\frac{1}{2(1-s)} \{\Pi_M,D(s)\} + \frac{1}{1-s} \Pi_M = \dot{D}(s)$.
\end{proof}

%-----------------------%
\subsection{Hitting time} \label{apx:HT}
%-----------------------%

From now on we assume that $P$ is ergodic and reversible. Recall from \Def{HT} that $\HT(P,M)$ is the expected number of steps it takes for the \RWA{} to find a marked vertex, starting from the stationary distribution of $P$ restricted to unmarked vertices.
We now prove \Prop{HT(P,M)} which expresses the hitting time of $P$ in terms of the spectral properties of the discriminant matrix of the absorbing walk $P'.$

\hittingtime*

\begin{proof}
The expected number of iterations in the \RWA{} is
\begin{align}
  \HT(P,M)
  &:= \sum_{l=1}^\infty l \cdot \Pr[\text{need \emph{exactly} $l$ steps}] \\
  & = \sum_{l=1}^\infty \sum_{t=1}^l \Pr[\text{need \emph{exactly} $l$ steps}] \label{eq:lt} \\
  & = \sum_{t=1}^\infty \sum_{l=t}^\infty \Pr[\text{need \emph{exactly} $l$ steps}] \label{eq:tl} \\
  &   \tikz[remember picture] \path (0,0) coordinate (O);
    = \sum_{t=1}^\infty \Pr[\text{need \emph{at least} $t$ steps}] \\
  & = \sum_{t=0}^\infty \Pr[\text{need \emph{more} than $t$ steps}]. \label{eq:HT sum}
\end{align}
The region corresponding to the double sums in Eqs.~\EqRef{lt} and \EqRef{tl} is shown in \Fig{Sums}.

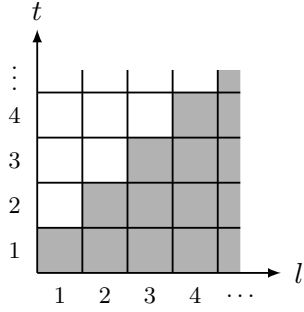
\begin{figure}
\begin{center}
%  \input{fig-sums.tex}
% !TeX root = KMOR.tex

\begin{tikzpicture}
  [line width = 0.7pt,
   arc/.style = {-latex}]

  % Number of squares:
  \def\W{4} % width
  \def\H{4} % height

  \def\d{0.6} % dimensions

  % Total dimensions:
  \def\rx{\W*\d} % width
  \def\ry{\H*\d} % height

  % Shaded region

  \path[fill = black!30]
    (0*\d,0*\d) -- (0*\d,1*\d) --
    (1*\d,1*\d) -- (1*\d,2*\d) --
    (2*\d,2*\d) -- (2*\d,3*\d) --
    (3*\d,3*\d) -- (3*\d,4*\d) --
    (4*\d,4*\d) -- (4*\d,4.5*\d) --
    (4.5*\d,4.5*\d) -- (4.5*\d,0) -- cycle;

  % Axis

  \def\a{1.4} % scaling factor for axis

  \draw [arc] (0,0) to (\rx+\a*\d, 0); % x axis
  \draw [arc] (0,0) to (0, \ry+\a*\d); % y axis

  \def\z{1.8} % scaling factor for labels

  \draw (\rx+\z*\d, 0) node {$l$};
  \draw (0, \ry+\z*\d) node {$t$};

  % Vertical lines

  \foreach \x in {1,...,\W} {
    \draw (\x*\d,0) -- (\x*\d,\ry+\d/2);
    \draw ({(\x-1/2)*\d},-\d/2) node {\footnotesize$\x$};
  }
  % Dots:
    \draw ({(\W+1/2)*\d},-\d/2) node {\footnotesize$\ldots$};

  % Horizontal lines

  \foreach \y in {1,...,\H} {
    \draw (0,\y*\d) -- (\rx+\d/2,\y*\d);
    \draw (-\d/2,{(\y-1/2)*\d}) node {\footnotesize$\y$};
  }
  % Dots:
    \draw (-\d/2,{(\H+1/2)*\d}) node {\footnotesize$\vdots$};

\end{tikzpicture}

\end{center}
\caption{Range of variables $l$ and $t$ in the double sums of Eqs.~\EqRef{lt} and \EqRef{tl}.\label{fig:Sums}}
\end{figure} 

It remains to determine the probability that no marked vertex is found after $t$ steps, starting from an unmarked vertex distributed according to $\tilde{\pi}_U = \pi_U/(1-p_M)$. The distribution of vertices at the first execution of \step{walk} of the \RWA{} is $\rv{\tilde{\pi}_U}{0_M}$, hence
\begin{equation}
  \Pr[\text{need \emph{more} than $t$ steps}]
  = \rv{\tilde{\pi}_U}{0_M} P'^{\,t} \rv{1_U}{0_M}\tp.
  \label{eq:Pr(more than t steps)}
\end{equation}
Recall from \Prop{P' to t} that $(P'^{\,t})_{UU} = P_{UU}^t$ so we can simplify \Eq{Pr(more than t steps)} as follows:
\begin{align}
 \Pr[\text{need \emph{more} than $t$ steps}] 
&= \rv{\tilde{\pi}_U}{0_M} P'^{\,t} \rv{1_U}{0_M}\tp \\
&= \frac{\pi_U}{1-p_M} P_{UU}^t 1_U\tp \\
&= \sqrt{\tfrac{\pi_U}{1-p_M}}
   \diag \bigl( \! \sqrt{\pi_U} \, \bigr) P_{UU}^t
   \diag \bigl( \! \sqrt{\pi_U} \, \bigr)^{-1}
   \sqrt{\tfrac{\pi_U\tp}{1-p_M}} \\
&= \bra{U} D'^t \ket{U},
\end{align}
where the last equality follows from the expression for the discriminant matrix $D'=D(1)$ in \Eq{D(1)}. By plugging this back in \Eq{HT sum} we get
\begin{equation}
  \HT(P,M) = \sum_{t=0}^\infty \bra{U} D'^t \ket{U}.
  \label{eq:HT(P,M) series}
\end{equation}
From the spectral decomposition $D'=\sum_{k=1}^n\lambda_k'\ket{v_k'}$, this may be rewritten as
\begin{equation}
  \HT(P,M) = \sum_{t=0}^\infty \sum_{k=1}^n
             \lambda_k'^t \abs{\braket{v_k'}{U}}^2.
  \label{eq:HT(P,M) double sum}
\end{equation}
Let $m := \abs{M}$ be the number of marked elements. Recall from \Eq{D(1)} that $D'=D(1)$ is block-diagonal and acts as identity matrix in the $m$-dimensional marked subspace. Furthermore, all $1$-eigenvectors of $D'$ lie in the marked subspace, since eigenvalue $1$ has multiplicity $m$ (recall from \Prop{Multiplicity} that $\lambda_k' = 1$ when $k > n - m$). Therefore, the terms in \Eq{HT(P,M) double sum} with $k > n - m$ disappear since $\braket{v'_k}{U} = 0$, and we get the desired expression by exchanging the two sums in \Eq{HT(P,M) double sum} and using the expansion $(1-x)^{-1} = \sum_{t=0}^\infty x^t$ where $\abs{x} < 1$.
\end{proof}

Note that the two sums in \Eq{HT(P,M) double sum} may not be exchanged before removing the terms with $k> n - m$: they do not commute in the presence of these extra terms since $\lambda_k'=1$ for $k> n - m$ and therefore $\sum_{t=0}^\infty \abs{\lambda_k'}^t$ diverges. This subtlety had unfortunately been overlooked in~\cite{KOR,KMOR}, and is at the source of the distinction between the hitting time $\HT(P,M)$ and the extended hitting time $\limHT$ (see \Apx{HT comparison}).

%...................................%
\subsubsection{Extended hitting time} \label{apx:HT(s)}
%...................................%

Recall the definition of the extended hitting time.

\deflimHT*

We now prove that the extended hitting time reduces to the usual hitting time in the case of a single marked element, even though they may differ in general.

\htcontinuity*

\begin{proof}
The fact that $\limHT=\HT(P,M)$ when $|M|=1$ follows immediately from the expression for $\HT(P,M)$ in \Prop{HT(P,M)} and \Def{HT(s)}.

For the second part, choose
\begin{equation}
  P = \frac{1}{4}
      \mx{3 & 1 & 0 \\
          1 & 2 & 1 \\
          0 & 1 &3}
\end{equation}
and let the last two elements be marked. If we explicitly compute the eigenvalues and eigenvectors of $D(s)$, then from \Def{HT(s)} we get that $\HT(s) = \frac{20}{(3-s)^2}$ for $s \in [0,1)$ and thus $\limHT = 5$. However, $\HT(P,M) = 4$. One can also use the formulas from \Lem{HT comparison} in \Apx{HT comparison} to verify this.
\end{proof}

This proposition implies that in the case of a single marked element, the quantum search algorithms in \Sect{Search algorithms} provide a quadratic speedup over the classical hitting time. In the general case of multiple marked elements, these quantum algorithms still solve the search problems but their cost is given in terms of the extended hitting time rather than the standard one.

%.......................%
\subsubsection{Lazy walk} \label{apx:Lazy}
%.......................%

For technical reasons, in \Sect{Search algorithms} it is important that all eigenvalues of $P(s)$ are non-negative. We can guarantee this using a standard trick---replacing the original Markov chain $P$ with a ``lazy'' walk $(P+I)/2$ where $I$ is the $n \times n$ identity matrix. In fact, we can assume without loss of generality that the original Markov chain already is ``lazy'', since this affects the hitting time only by a constant factor, as shown below.

\begin{proposition}\label{prop:Lazy}
Let $P$ be an ergodic and reversible Markov chain. Then for any $s \in [0,1]$ the eigenvalues of $(P(s)+I)/2$ are between $0$ and $1$. Moreover, if the interpolated hitting time of $P$ is $\HT(s)$, then the interpolated hitting time of $(P+I)/2$ is $2 \HT(s)$.
\end{proposition}

\begin{proof}
Since $P$ is reversible, so is $P(s)$ by \Prop{Reversibility}. Thus the eigenvalues of $P(s)$ are real by \Prop{Eigenvalues}. If $\lambda_k(s)$ is an eigenvalue of $P(s)$ then $\lambda_k(s) \in [-1,1]$ according to \PFT{}. Thus, the eigenvalues of $(P(s)+I)/2$ satisfy $(\lambda_k(s)+1)/2 \in [0,1]$.

Recall from \Prop{Eigenvalues} that $P(s)$ and $D(s)$ are similar. Thus, the discriminant matrix of $(P(s)+I)/2$ is $(D(s)+I)/2$, which has the same eigenvectors as $D(s)$. By \Def{HT(s)}, the interpolated hitting time of $(P(s)+I)/2$ is
\begin{equation}
  \sum_{k=1}^{n-1} \frac{\abs{\braket{v_k(s)}{U}}^2}{1 - \frac{\lambda_k(s)+1}{2}}.
\end{equation}
Since $1 - \frac{\lambda_k(s)+1}{2} = \frac{1-\lambda_k(s)}{2}$, the above expression is equal to $2 \HT(s)$ as claimed.
\end{proof}

%........................................................................................................%
\subsubsection{Relationship between \texorpdfstring{$\HT(s)$}{HT(s)} and \texorpdfstring{$\limHT$}{HT(1)}}
\label{apx:relation between HT(s) and HT(1)}
%........................................................................................................%

In this section we express $\HT(s)$ as a function of $s$ and $\limHT$, which is the main result of this appendix. The main idea is to relate $\frac{d}{ds} \HT(s)$ to $\HT(s)$. When we solve the resulting differential equation, the boundary condition at $s = 1$ gives the desired result.

First, note that by \Def{HT(s)}, $\HT(s)$ may be written as $\HT(s) = \bra{U} A(s)\ket{U}$, where
\begin{equation}
  A(s) := \sum_{k = 1}^{n-1}
          \frac{\ket{v_k(s)}\bra{v_k(s)}}{1 - \lambda_k(s)}.
  \label{eq:A(s)}
\end{equation}
The following property of $A(s)$ will be useful on several occasions.

\begin{proposition}\label{prop:AM-AU}
$A(s) \ket{M} = - \frac{\cos \theta(s)}{\sin \theta(s)} A(s) \ket{U}$.
\end{proposition}

\begin{proof}
Recall from \Prop{Multiplicity} that $\lambda_n(s) = 1$, so $A(s) \ket{v_n(s)} = 0$ by definition. If we substitute $\ket{v_n(s)} = \cos \theta(s) \ket{U} + \sin \theta(s) \ket{M}$ from \Prop{vn(s)} in this equation, we get the desired formula.
\end{proof}

\begin{lemma}\label{lem:HT(s) dot}
For $s < 1$, the derivative of $\HT(s)$ is related to $\HT(s)$ as
\begin{equation}
  \frac{d}{ds} \HT(s) = \frac{2(1-p_M)}{1-s(1-p_M)} \HT(s)
\end{equation}
where $p_M$ is the probability to pick a marked state from the stationary distribution $\pi$ of $P$.
\end{lemma}

\begin{proof}
Recall that $\HT(s) = \bra{U} A(s) \ket{U}$ where $A(s)$ may be written as
\begin{equation}
  A(s) = B(s)^{-1} - \Pi_n(s) \text{\: where \:}
  B(s) := I - D(s) + \Pi_n(s), \:
  \Pi_n(s) := \proj{v_n(s)}.
  \label{eq:A,B,Pi}
\end{equation}
Recall from \Apx{D(s) spectrum} that $\ket{v_n(s)}$ is the unique $(+1)$-eigenvector of $D(s)$ for $s \in \co{0}{1}$, thus $B(s)$ is indeed invertible when $s$ is in this range.

From now on we will not write the dependence on $s$ explicitly. We will also often use $\dot{f}(s)$ as a shorthand form of $\frac{d}{ds} f(s)$. Let us start with
\begin{equation}
  \frac{d}{ds} \HT\ofs = \bra{U} \dot{A}\ofs \ket{U}
  \label{eq:HT(s) dot}
\end{equation}
and expand $\dot{A}$ using \Eq{A,B,Pi}. To find $\frac{d}{ds} (B^{-1})$, take the derivative of both sides of $B^{-1}\ofs B\ofs = I$ and get $\frac{d}{ds} (B\ofs^{-1}) \cdot B\ofs + B\ofs^{-1} \cdot \frac{d}{ds} B\ofs = 0$. Thus $\frac{d}{ds} (B\ofs^{-1}) = -B\ofs^{-1} \dot{B\ofs} B\ofs^{-1}$ and
\begin{equation}
  \dot{A}\ofs = - B\ofs^{-1} \dot{B}\ofs B\ofs^{-1} - \dot{\Pi}_n\ofs.
\end{equation}
Notice from \Eq{A,B,Pi} that $\dot{B}\ofs = -\dot{D}\ofs + \dot{\Pi}_n\ofs$, thus $\dot{A}\ofs = - B\ofs^{-1} ( - \dot{D}\ofs + \dot{\Pi}_n\ofs ) B\ofs^{-1} - \dot{\Pi}_n\ofs$ and $\frac{d}{ds} \HT\ofs = h_1 + h_2 + h_3$ where
\begin{align}
  h_1 &:= \bra{U} B\ofs^{-1} \dot{D}\ofs B\ofs^{-1} \ket{U}, \\
  h_2 &:=-\bra{U} B\ofs^{-1} \dot{\Pi}_n\ofs B\ofs^{-1} \ket{U}, \\
  h_3 &:=-\bra{U} \dot{\Pi}_n\ofs \ket{U}.
\end{align}
Let us evaluate each of these terms separately.

To evaluate the first term $h_1$, we substitute $\dot{D}\ofs = \frac{1}{2(1-s)} \bigl\{ \Pi_M, I - D\ofs \bigr\}$ from \Prop{D(s) dot} and replace $I - D\ofs$ by $B\ofs - \Pi_n\ofs$ according to \Eq{A,B,Pi}:
\begin{align}
  2(1-s) h_1
   &= \bra{U} B\ofs^{-1} \{\Pi_M, B\ofs-\Pi_n\ofs\} B\ofs^{-1} \ket{U} \\
   &= \bra{U} B\ofs^{-1} \bigl(\{\Pi_M,B\ofs\}-\{\Pi_M,\Pi_n\ofs\}\bigr) B\ofs^{-1} \ket{U} \\
   &= \bra{U} \{B\ofs^{-1},\Pi_M\} \ket{U} - \bra{U} B\ofs^{-1} \{\Pi_M,\Pi_n\ofs\} B\ofs^{-1} \ket{U}.
\end{align}
Recall that $\Pi_M = \sum_{x \in M} \proj{x}$ is the projector onto the marked states. Thus $\Pi_M \ket{U} = 0$ and the first term vanishes. Note that $B\ofs$ has the same eigenvectors as $D\ofs$. In particular, $B\ofs^{-1} \ket{v_n\ofs} = \ket{v_n\ofs}$ and thus $B\ofs^{-1} \Pi_n\ofs = \Pi_n\ofs = \Pi_n\ofs B\ofs^{-1}$. Using this we can expand the anti-commutator in the second term: $B\ofs^{-1} \{\Pi_M,\Pi_n\ofs\} B\ofs^{-1} = B\ofs^{-1} \Pi_M \Pi_n\ofs + \Pi_n\ofs \Pi_M B\ofs^{-1}$. Since all three matrices in this expression are real and symmetric and $\ket{U}$ is also real, both terms of the anti-commutator have the same contribution, so we get
\begin{equation}
  2(1-s) h_1 = -2 \bra{U} B\ofs^{-1} \Pi_M \Pi_n\ofs \ket{U}.
\end{equation}
Recall from \Prop{vn(s)} that $\ket{v_n\ofs} = \cos \theta\ofs \ket{U} + \sin \theta\ofs \ket{M}$, so we see that $\Pi_M \Pi_n\ofs \ket{U} = \Pi_M \ket{v_n\ofs} \cdot \braket{v_n\ofs}{U} = \sin \theta\ofs \ket{M} \cdot \cos \theta\ofs$. Moreover, $B\ofs^{-1} = A\ofs + \Pi_n\ofs$ according to \Eq{A,B,Pi}, so
\begin{equation}
 2(1-s) h_1 = -2 \sin\theta\ofs \cos\theta\ofs \bra{U} (A\ofs + \Pi_n\ofs) \ket{M}.
\end{equation}
Recall from \Prop{AM-AU} that $\sin \theta \bra{U} A \ket{M} = \cos \theta \bra{U} A \ket{U}$. To simplify the second term, notice that $\bra{U} \Pi_n\ofs \ket{M} = \braket{U}{v_n\ofs} \cdot \braket{v_n\ofs}{M} = \cos\theta\ofs \cdot \sin\theta\ofs$. When we put this together, we get
\begin{equation}
  2(1-s) h_1 = 2 \cos^2\theta\ofs \bra{U} A\ofs \ket{U} - 2 \sin^2\theta\ofs \cos^2 \theta\ofs
\end{equation}
or simply
\begin{equation}
  h_1 = \frac{\cos^2\theta\ofs}{1-s} \bigl( \bra{U} A\ofs \ket{U} - \sin^2\theta\ofs \bigr).
  \label{eq:h1}
\end{equation}

Let us now consider the second term $h_2 = -\bra{U} B\ofs^{-1} \dot{\Pi}_n\ofs B\ofs^{-1} \ket{U}$. First, we compute $\dot{\Pi}_n\ofs = \ket{\dot{v}_n\ofs} \bra{v_n\ofs} + \ket{v_n\ofs} \bra{\dot{v}_n\ofs}$. Using $B\ofs^{-1} \ket{v_n\ofs} = \ket{v_n\ofs}$ we get $B\ofs^{-1} \dot{\Pi}_n\ofs B\ofs^{-1} = B\ofs^{-1} \ket{\dot{v}_n\ofs} \bra{v_n\ofs} + \ket{v_n\ofs} \bra{\dot{v}_n\ofs} B\ofs^{-1}$. Since $\braket{v_n\ofs}{U} = \cos \theta\ofs$ we have
\begin{equation}
  h_2 = -2 \bra{U} B\ofs^{-1} \ket{\dot{v}_n\ofs} \cos \theta\ofs
\end{equation}
where the factor two comes from the fact that all vectors involved are real and matrix $B\ofs^{-1}$ is real and symmetric. Let us compute
\begin{equation}
  \ket{\dot{v}_n\ofs} = \dot{\theta}\ofs \bigl( -\sin \theta\ofs \ket{U} + \cos \theta\ofs \ket{M} \bigr).
\end{equation}
Notice that $\braket{v_n\ofs}{\dot{v}_n\ofs} = 0$ and thus $\Pi_n\ofs \ket{\dot{v}_n\ofs} = 0$. By substituting $B\ofs^{-1} = A\ofs + \Pi_n\ofs$ from \Eq{A,B,Pi} we get
\begin{equation}
  h_2 = -2 \bra{U} A\ofs \ket{\dot{v}_n\ofs} \cos \theta\ofs.
\end{equation}
Next, we substitute $\ket{\dot{v}_n\ofs}$ and get
\begin{equation}
  h_2 = -2 \dot{\theta}\ofs \bigl(
           - \sin\theta\ofs \bra{U} A\ofs \ket{U}
           + \cos\theta\ofs \bra{U} A\ofs \ket{M}
           \bigr) \cos\theta\ofs.
\end{equation}
Now we use \Prop{AM-AU} to substitute $A\ofs \ket{M}$ by $A\ofs \ket{U}$:
\begin{equation}
  h_2 = -2 \dot{\theta}\ofs \biggl(
           - \sin\theta\ofs
           - \frac{\cos^2\theta\ofs}{\sin\theta\ofs}
           \biggr) \bra{U} A\ofs \ket{U} \cos\theta\ofs
      =  2 \dot{\theta}\ofs
           \frac{\cos\theta\ofs}{\sin\theta\ofs}
           \bra{U} A\ofs \ket{U}.
\end{equation}
Finally, we substitute $2 \dot{\theta}\ofs = \frac{\sin\theta\ofs \cos\theta\ofs}{1-s}$ from \Eq{theta(s) dot} and get
\begin{equation}
  h_2 = \frac{\cos^2\theta\ofs}{1-s} \bra{U} A\ofs \ket{U}.
  \label{eq:h2}
\end{equation}

For the last term $h_3 = -\bra{U} \dot{\Pi}_n\ofs \ket{U}$ we observe that $\braket{U}{\dot{v}_n\ofs} \braket{v_n\ofs}{U} = - \dot{\theta}\ofs \sin\theta\ofs \cdot \cos\theta\ofs$ thus $h_3 = 2 \dot{\theta}\ofs \sin\theta\ofs \cos\theta\ofs$ where the factor two comes from symmetry. After substituting $2\dot{\theta}\ofs$ from \Eq{theta(s) dot} we get
\begin{equation}
  h_3 = \frac{\cos^2\theta\ofs}{1-s} \sin^2\theta\ofs.
  \label{eq:h3}
\end{equation}

When we compare Eqs.~\EqRef{h1}, \EqRef{h2}, and~\EqRef{h3} we notice that $h_2 = h_1 + h_3$. Thus the derivative of the hitting time is $\frac{d}{ds} \HT\ofs = h_1 + h_2 + h_3 = 2 h_2$. Recall from \Def{HT(s)} that $\HT\ofs = \bra{U} A\ofs \ket{U}$. Thus
\begin{equation}
  \frac{d}{ds} \HT(s) = 2 \frac{\cos^2\theta(s)}{1-s} \HT(s).
  \label{eq:diff-eq}
\end{equation}
By substituting $\cos\theta(s)$ from \Eq{cos and sin theta} we get the desired result.
\end{proof}

We now prove the following theorem which relates $\HT(s)$ to $\limHT$.

\HTintermsoflimHT*

\begin{proof}
When the marked element is unique, $\limHT = \HT(P,M)$ by \Prop{Continuity}. This gives the second part.

We will prove the first part by solving the differential equation obtained in \Lem{HT(s) dot}. Consider \Eq{diff-eq} and recall from \Eq{theta(s) dot} that $2 \dot{\theta}\ofs = \frac{\sin\theta\ofs \cos\theta\ofs}{1-s}$. We can rewrite the coefficient in \Eq{diff-eq} as
\begin{equation}
  2 \frac{\cos^2\theta\ofs}{1-s}
  = 2 \cdot \frac{\sin\theta\ofs \cos\theta\ofs}{1-s}
      \cdot \frac{\cos\theta\ofs}{\sin\theta\ofs}
  = 4 \dot{\theta}\ofs \frac{\cos\theta\ofs}{\sin\theta\ofs}
  = 4 \frac{\frac{d}{ds} (\sin\theta\ofs)}{\sin\theta\ofs}.
\end{equation}
Then the differential equation becomes
\begin{equation}
  \frac{\frac{d}{ds}\HT(s)}
                   {\HT(s)}
  = 4 \frac{\frac{d}{ds}(\sin\theta(s))}
                        {\sin\theta(s)}.
\end{equation}
By integrating both sides we get
\begin{equation}
  \ln \, \abs{\HT(s)} = 4 \ln \, \abs{\sin\theta(s)} + C
\end{equation}
for some constant $C$. Recall from \Eq{cos and sin theta} that $\sin\theta(1) = 1$, so the boundary condition at $s=1$ gives us $C = \ln \, \abs{\limHT}$. Since all quantities are non-negative, we can omit the absolute value signs. After exponentiating both sides we get
\begin{equation}
  \HT(s) = \sin^4 \theta(s) \cdot \limHT.
\end{equation}
We get the desired expression when we substitute $\sin\theta(s)$ from \Eq{cos and sin theta}.
\end{proof}

In \Sect{Search algorithms} we consider several quantum search algorithms whose running time depends on $\HT(s)$ for some values of $s$. \Thm{HT} is a crucial ingredient in analysis of these algorithms: when the marked element is unique, it expresses $\HT(s)$ as a function of $s$ and the usual hitting time $\HT(P,M)$. In particular, we see that $\HT(s)$ is monotonically increasing as a function of $s$ and it reaches maximum value at $s = 1$ (some example plots of $\HT(s)$ are shown in \Fig{HT(s)}). This observation is crucial, for example, in the proof of \Thm{Search with known HT}.

\begin{figure}[th]
  \centering
%  \input{fig-HTs.tex}
% !TeX root = KMOR.tex

% Zoom factors
\def\zx{7.0}
\def\zy{5.0}

\begin{tikzpicture}
  [domain = 0:1, variable = \s, samples = 60,
   line width = 0.7pt,
   arc/.style = {-latex}]

  \def\d{0.1}
  \def\D{0.6}

  % Coordinate axis
  \draw[arc] (-\d,0) -- (\zx+\D,0) node[right] {$s$};
  \draw[arc] (0,-\d) -- (0,\zy+\D) node[above] {$\HT(s)$};

  \def\sm{0.3}

  % The p = 1  plot
  \draw (-\d,\zy) to (\zx,\zy);
  \draw[dashed] (\zx,-\d) to (\zx,\zy);
  \draw (\zx*\sm, \zy) + (-0.3,0.3) node {$p_M = 1.0$};

  % Other plots
  \foreach \p in {0.2,0.4,0.6,0.8} {
    \draw (\zx*\sm, { \zy * (\p / (1 - \sm * (1 - \p)))^2 }) + (-0.3,0.3) node {$p_M = \p$};
    \draw plot (\zx*\s, { \zy * (\p / (1 - \s * (1 - \p)))^2 });
  }

  \def\l{0.35}

  % Tick labels
  \draw (0  ,-\l) node {$0$};
  \draw (\zx,-\l) node {$1$};
  \draw (-\l,  0)+(+0.05,0) node {$0$};
  \draw (-\l,\zy)+(-0.80,0) node {$\limHT$};

\end{tikzpicture}

  \caption[The interpolated hitting time $\HT(s)$ as a function of $s$]{The interpolated hitting time $\HT(s)$ as a function of $s$ for several values of $p_M$ according to \Thm{HT}.}
  \label{fig:HT(s)}
\end{figure}
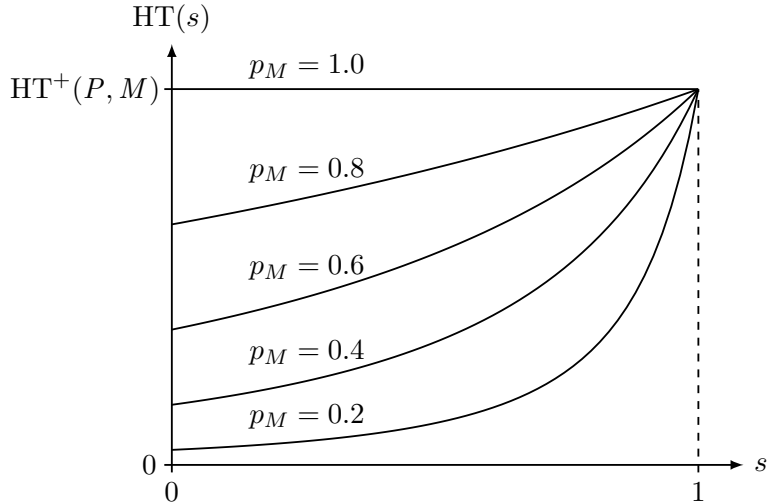

%%%%%%%%%%%%%%%%%%%%%%%%%%%%%%%%%%%%%%%%%%%%%%%%%%%%%%%%%%%%%%%%%%%%%%
\section{Spectrum and implementation of \texorpdfstring{$W(s)$}{W(s)}} \label{apx:W(s)}
%%%%%%%%%%%%%%%%%%%%%%%%%%%%%%%%%%%%%%%%%%%%%%%%%%%%%%%%%%%%%%%%%%%%%%

Szegedy~\cite{Sze} proposed a general method to map a random walk to a unitary operator that defines a quantum walk. The first step of Szegedy's construction is to map the rows of $P(s)$ to quantum states. Let $X$ be the state space of $P(s)$ and $\Hi := \spn \set{\ket{x} : x \in X}$ be a complex Euclidean space of dimension $n := \abs{X}$ with basis states labelled by elements of $X$. For every $x \in X$ we define the following state in $\Hi$:
 \begin{equation}
  \ket{p_x(s)} := \sum_{y \in X} \sqrt{P_{xy}(s)} \ket{y}.
\end{equation}
Notice that these states are correctly normalized, since $P(s)$ is row-stochastic. Following the approach of Szegedy~\cite{Sze}, we define a unitary operator $V(s)$ acting on $\Hi \x \Hi$ as
\begin{equation}
  V(s) \ket{x, \bar{0}} := \ket{x} \ket{p_x(s)} = \sum_{y \in X} \sqrt{P_{xy}(s)} \ket{x,y},
  \label{eq:V(s)}
\end{equation}
when the second register is in some reference state $\ket{\bar{0}} \in \Hi$, and arbitrarily otherwise. It will not be relevant to us how $V(s)$ is extended from $\Hi \otimes \ket{\bar{0}}$ to $\Hi \otimes \Hi$. The only constraint we impose is that $V(s)$ is continuous as a function of $s$, which is a reasonable assumption from a physical point of view.

Let \shift{} be the operation defined in \Eq{Shift}. Let $\Pi_0 := I \x \proj{\bar{0}}$ be the projector that keeps only the component containing the reference state $\ket{\bar{0}}$ in the second register and let $\reflex_{\X} := 2 \Pi_0 - I \x I$. The goal of this section is to find the spectral decomposition of the quantum walk operator corresponding to $P(s)$:
\begin{equation}
  W(s) := V(s)\ct \, \shift \, V(s) \cdot \reflex_{\X}
  \label{eq:W(s)}
\end{equation}
where $V(s) := V(P(s))$. Recall from \Apx{D(s) spectrum} that $\lambda_k(s)$ and $\ket{v_k(s)}$ are the eigenvalues and eigenvectors of the discriminant matrix $D(s)$ of $P(s)$. 

%------------------------------------------------------------------%
\subsection{Spectral decomposition of \texorpdfstring{$W(s)$}{W(s)}} \label{apx:Spectrum of W(s)}
%------------------------------------------------------------------%

In this section we determine the invariant subspaces of $W(s)$ and find its eigenvectors and eigenvalues. First, observe that on certain states $\shift$ acts as the swap gate.

\begin{proposition}\label{prop:Shift succeeds}
If $P$ is a Markov chain on graph $G$ then $\shift \, \ket{x, p_x(s)} = \ket{p_x(s), x}$, \ie, \shift{} always succeeds on states of the form $\ket{x, p_x(s)}$ for any $x \in X$.
\end{proposition}

\begin{proof}
From \Eq{V(s)} we get
\begin{align}
  \shift \, \ket{x, p_x(s)}
  &= \shift \, \sum_{y \in X} \sqrt{P_{xy}(s)} \ket{x, y} \\
  &= \sum_{y \in X} \sqrt{P_{xy}(s)} \ket{y, x} \\
  &= \ket{p_x(s), x},
\end{align}
where the second equality holds since $P(s)$ is a Markov chain on $G$ and thus $P_{xy}(s) = 0$ when $xy$ is not an edge of $G$.
\end{proof}

It follows from \Prop{Shift succeeds} that $\shift$ always succeeds when $V\ct(s) \, \shift \, V(s)$ acts on any state that has $\ket{\bar{0}}$ in the second register. In fact, we can say even more.

\begin{proposition}\label{prop:VSV and D}
If $P$ is a Markov chain on graph $G$ then the operator $V\ct(s) \, \shift \, V(s)$ acts as the discriminant matrix $D(s)$ (see \Apx{Discriminant}) when restricted to $\ket{\bar{0}}$ in the second register, \ie,
\begin{equation}
  \Pi_0 V\ct(s) \, \shift \, V(s) \Pi_0 = D(s) \x \proj{\bar{0}}.
\end{equation}
\end{proposition}

\begin{proof}
From \Eq{V(s)} and \Prop{Shift succeeds} we get
\begin{align}
  \bra{x,\bar{0}} V\ct(s) \, \shift \, V(s) \ket{y,\bar{0}}
  &= \bra{x, p_x(s)} \, \shift \, \ket{y, p_y(s)} \\
  &= \braket{x, p_x(s)}{p_y(s), y} \\
  &= \braket{p_x(s)}{y} \braket{x}{p_y(s)} \\
  &= \sqrt{P_{xy}(s) P_{yx}(s)} \\
  &= D_{xy}(s)
  \label{eq:VSV and D}
\end{align}
where last equality follows from \Eq{D(s)}.
\end{proof}

This suggests a close relationship between the operators $D(s)$ and $V\ct(s) \, \shift \, V(s)$. We want to extend this and relate the spectral decompositions of $D(s)$ and $W(s)$ from \Eq{W(s)}. Recall from \Eq{D(s) spectrum} that the spectral decomposition of $D(s)$ is $D(s) = \sum_{i=1}^n \lambda_i(s) \ket{v_i(s)} \bra{v_i(s)}$.

\begin{definition}\label{def:Subspaces}
We define the following subspaces of $\Hi \x \Hi$ in terms of the eigenvectors of $D(s)$ and the operator $V\ct(s) \, \shift \, V(s)$:
\begin{align}
  \mathcal{B}_k(s) &:= \spn \set{\ket{v_k(s),\bar{0}}, V\ct(s) \, \shift \, V(s) \ket{v_k(s),\bar{0}}},
                       \quad k \in \set{1, \dotsc, n-1}, \\[2pt]
  \mathcal{B}_n(s) &:= \spn \set{\ket{v_n(s),\bar{0}}}, \\
  \mathcal{B}^\perp(s) &:= \textstyle \bigl( \bigoplus_{k=1}^n \mathcal{B}_k(s) \bigr)^\perp.
\end{align}
\end{definition}

Let us first understand how $V\ct(s) \, \shift \, V(s)$ acts on vectors defining the subspaces in \Def{Subspaces}.
Let us consider $s < 1$ and $k < n$. Then $\lambda_k(s) \neq 1$ by \Prop{Multiplicity}. By unitarity of $V\ct(s) \, \shift \, V(s)$ and \Prop{VSV and D},
\begin{equation}
  V\ct(s) \, \shift \, V(s) \ket{v_k(s), \bar{0}}
  = \lambda_k(s) \ket{v_k(s), \bar{0}}
  + \sqrt{1-\lambda_k(s)^2} \ket{v_k(s), \bar{0}}^\perp
    \label{eq:def-vk-perp}
\end{equation}
for some unit vector $\ket{v_k(s), \bar{0}}^\perp$ orthogonal to $\ket{v_k(s), \bar{0}}$ and lying in the subspace $\mathcal{B}_k(s)$. In particular, $\mathcal{B}_k(s)$ is two-dimensional. Note that $\ket{v_k(s), \bar{0}}^\perp$ depends on how the operator $V(s)$, defined in \Eq{V(s)}, is extended to the rest of the space $\Hi \x \Hi$.

Let us also find how $V\ct(s) \, \shift \, V(s)$ acts on $\ket{v_k(s), \bar{0}}^\perp$. If we apply $V\ct(s) \, \shift \, V(s)$ to both sides of \Eq{def-vk-perp}, we get
\begin{equation}
  \ket{v_k(s), \bar{0}}
  = \lambda_k(s) V\ct(s) \, \shift \, V(s) \ket{v_k(s), \bar{0}}
  + \sqrt{1-\lambda_k(s)^2} V\ct(s) \, \shift \, V(s) \ket{v_k(s), \bar{0}}^\perp.
\end{equation}
We regroup the terms and substitute \Eq{def-vk-perp}:
\begin{align}
  &\sqrt{1-\lambda_k(s)^2} V\ct(s) \, \shift \, V(s) \ket{v_k(s), \bar{0}}^\perp \\
  &= \ket{v_k(s), \bar{0}}
  - \lambda_k(s) V\ct(s) \, \shift \, V(s) \ket{v_k(s), \bar{0}} \\
  &= \ket{v_k(s), \bar{0}}
  - \lambda_k(s) \Bigl(
      \lambda_k(s) \ket{v_k(s), \bar{0}} + \sqrt{1-\lambda_k(s)^2} \ket{v_k(s), \bar{0}}^\perp
    \Bigr).
\end{align}
After cancellation we get
\begin{equation}
  V\ct(s) \, \shift \, V(s) \ket{v_k(s), \bar{0}}^\perp
  = \sqrt{1-\lambda_k(s)^2} \ket{v_k(s), \bar{0}}
  - \lambda_k(s) \ket{v_k(s), \bar{0}}^\perp.
  \label{eq:vk-perp}
\end{equation}

\begin{proposition}\label{prop:Subspaces}
Subspaces $\mathcal{B}_1(s), \dotsc, \mathcal{B}_n(s)$, and $\mathcal{B}^\perp(s)$ are mutually orthogonal and invariant under $W(s)$ for all $s \in [0,1]$.
\end{proposition}

\begin{proof}
Clearly, $\mathcal{B}^\perp(s)$ is orthogonal to the other subspaces. Vectors $\ket{v_k(s),\bar{0}}$ are also mutually orthogonal for $k \in \set{1,\dotsc,n}$, since they form an orthonormal basis of $\Hi \x \ket{\bar{0}}$. Finally, note from \Prop{VSV and D} that
\begin{equation}
  \bra{v_j(s), \bar{0}} \cdot V\ct(s) \, \shift \, V(s) \ket{v_k(s), \bar{0}}
= \bra{v_j(s)} D(s) \ket{v_k(s)}
= \delta_{jk} \lambda_k(s),
\end{equation}
so $V\ct(s) \, \shift \, V(s) \ket{v_k(s), \bar{0}}$ is orthogonal to $\ket{v_j(s), \bar{0}}$ for any $j \neq k$. Thus all of the above subspaces are mutually orthogonal.

Let us show that these subspaces are invariant under $W(s)$. From the definition of $W(s)$ in \Eq{W(s)} we see that it suffices to check the invariance of each subspace under $V\ct(s) \, \shift \, V(s)$ and $\Pi_0$ separately.

First, let us argue the invariance under $V\ct(s) \, \shift \, V(s)$. Since $\shift^2$ acts as identity according to \Eq{Shift}, then so does $V\ct(s) \, \shift \, V(s)$ and hence $\mathcal{B}_k(s)$ is invariant under $V\ct(s) \, \shift \, V(s)$ for any $k < n$. Next, $\mathcal{B}_n(s)$ is invariant, since $V\ct(s) \, \shift \, V(s)$ acts trivially on $\ket{v_n(s), \bar{0}}$ by \Prop{VSV and D}. Finally, $\mathcal{B}^\perp(s)$ is invariant, since it is the orthogonal complement of invariant subspaces.

Let us now show the invariance under $\Pi_0$. First, let us argue that
\begin{equation}
  \braket{v_j(s), \bar{0}}{v_k(s), \bar{0}}^\perp = 0, \quad \forall j \in \set{1, \dotsc, n}.
\end{equation}
These vectors lie in subspaces $\mathcal{B}_j(s)$ and $\mathcal{B}_k(s)$ that are mutually orthogonal when $j \neq k$. For $j = k$ this holds by definition of $\ket{v_k(s), \bar{0}}^\perp$. Since $\spn \set{\ket{v_k(s), \bar{0}}}_{k=1}^n = \Hi \x \ket{\bar{0}}$, we conclude that
\begin{equation}
  \Pi_0 \ket{v_k(s), \bar{0}}^\perp = 0.
  \label{eq:Pi0}
\end{equation}
From \Eq{def-vk-perp} we get
\begin{equation}
  \Pi_0 V\ct(s) \, \shift \, V(s) \ket{v_k(s), \bar{0}} = \lambda_k(s) \ket{v_k(s), \bar{0}},
\end{equation}
hence $\mathcal{B}_k(s)$ is invariant under $\Pi_0$ for $k < n$. Next, $\mathcal{B}_n(s)$ is invariant since $\Pi_0 \ket{v_n(s), \bar{0}} = \ket{v_n(s), \bar{0}}$. Finally, $\mathcal{B}^\perp(s)$ is invariant by being the orthogonal complement of invariant subspaces.
\end{proof}

The following lemma by Szegedy~\cite{Sze} provides the spectral decomposition of $W(s)$ in terms of that of $D(s)$. Note that we can guarantee that all eigenvalues of $D(s)$ are in $[0,1]$ via \Prop{Lazy}.

\lemWspectrum*

\begin{proof}
Recall Eqs.~\EqRef{def-vk-perp} and~\EqRef{vk-perp}:
\begin{IEEEeqnarray}{lCl}
  V\ct(s) \, \shift \, V(s) \cdot \ket{v_k(s), \bar{0}}
 &=&\lambda_k(s) \ket{v_k(s), \bar{0}} + \sqrt{1-\lambda_k(s)^2} \ket{v_k(s), \bar{0}}^\perp,
    \label{eq:V Shift V vk} \\
  V\ct(s) \, \shift \, V(s) \cdot \ket{v_k(s), \bar{0}}^\perp
 &=&\sqrt{1-\lambda_k(s)^2} \ket{v_k(s), \bar{0}}
  - \lambda_k(s) \ket{v_k(s), \bar{0}}^\perp.
    \label{eq:V Shift V vk perp}
\end{IEEEeqnarray}
Clearly, $\reflex_{\X} \ket{v_k(s), \bar{0}} = \ket{v_k(s), \bar{0}}$ from \Eq{refX}, and recall from \Eq{Pi0} that $\Pi_0 \ket{v_k(s), \bar{0}}^\perp = 0$, so $\reflex_{\X} \ket{v_k(s), \bar{0}}^\perp = - \ket{v_k(s), \bar{0}}^\perp$. Thus, Eqs.~\EqRef{V Shift V vk} and~\EqRef{V Shift V vk perp} give us
\begin{IEEEeqnarray}{lCl}
  W(s) \cdot \ket{v_k(s), \bar{0}}
 &=&\lambda_k(s) \ket{v_k(s), \bar{0}} + \sqrt{1-\lambda_k(s)^2} \ket{v_k(s), \bar{0}}^\perp,
    \label{eq:W(s) vk} \\
  W(s) \cdot \ket{v_k(s), \bar{0}}^\perp
 &=& - \sqrt{1-\lambda_k(s)^2} \ket{v_k(s), \bar{0}} + \lambda_k(s) \ket{v_k(s), \bar{0}}^\perp.
    \label{eq:W(s) vk perp}
\end{IEEEeqnarray}
Recall from \Prop{Subspaces} that subspaces $\mathcal{B}_k(s)$ are mutually orthogonal and invariant under $W(s)$. In fact, $W(s)$ acts in the basis $\set{\ket{v_k(s), \bar{0}}, \ket{v_k(s), \bar{0}}^\perp}$ of $\mathcal{B}_k(s)$ as
\begin{equation}
  \mx{
            \lambda_k(s)    & -\sqrt{1-\lambda_k(s)^2} \\
    \sqrt{1-\lambda_k(s)^2} &          \lambda_k(s)
  }
  = \lambda_k(s) I + i \sqrt{1-\lambda_k(s)^2} \, \sigma_y
  \label{eq:Orthogonal rotation}
\end{equation}
where $\sigma_y := \smx{0 & -i \\ i & 0}$ is the Pauli $y$ matrix. The matrix in \Eq{Orthogonal rotation} has the same eigenvectors as $\sigma_y$ and its eigenvalues are given by
\begin{equation}
  \lambda_k(s) \pm i \sqrt{1-\lambda_k(s)^2} = e^{\pm i \varphi_k(s)}.
\end{equation}
This shows \Eq{Bk}. To obtain \Eq{Bn}, we use \Prop{VSV and D}:
\begin{equation}
  \bra{v_n(s), \bar{0}}
  \cdot V\ct(s) \, \shift \, V(s) \cdot
  \ket{v_n(s), \bar{0}} = 1,
\end{equation}
so $\ket{v_n(s), \bar{0}}$ is an eigenvector of $W(s)$ with eigenvalue $1$.
\end{proof}

%------------------------------------------------------------%
\subsection{Quantum circuit for \texorpdfstring{$W(s)$}{W(s)}}
\label{apx:Circuit for W(s)}
%------------------------------------------------------------%

Recall that $\update(P)$ can be used to implement the quantum walk operator $W(P)$. However, we would also like to be able to implement the quantum analogue of $P(s)$ for any $s \in [0,1]$. Recall from \Eq{W(s)} that it is given by
\begin{equation}
  W(s) = V(s)\ct \, \shift \, V(s) \cdot \reflex_{\X}.
\end{equation}
We know how to implement $\shift$ and $\reflex_{\X}$, so we only need to understand how to implement $V(s)$ using $V(P)$. Recall from \Eq{V(P)} that
\begin{equation}
  V(s) \ket{x} \ancilla = \ket{x} \ket{p_x(s)} = \ket{x} \sum_{y \in X} \sqrt{P_{xy}(s)} \ket{y}.
\end{equation}

In the following lemma, we assume that we know $p_{xx}$ for every $x$. This is reasonable since in practice the probability of self-loops is known. In many cases, it is even independent of $x$. For the rest of this chapter, we assume that this is not an obstacle (we can assume that one call to $\update(P)$ allows to learn $p_{xx}$ for any $x$).

\begin{lemma}\label{lem:Update}
Assuming that $p_{xx}$ is known for every $x$, \Itp$(P,M,s)$ implements $V(s)$ with quantum complexity $2 \checkingcost + \updatecost$. Thus, $\update(P(s))$ has quantum complexity of order $\checkingcost + \updatecost$.
\end{lemma}

\begin{proof}
We explain only how to implement $V(s)$ using one call to $V(P)$ and two calls to $\check(M)$. The algorithm for $V(s)\ct$ is obtained from the reverse algorithm.

Our algorithm uses four registers: $\Reg_1$, $\Reg_2$, $\Reg_3$, $\Reg_4$. The first two registers have underlying state space $\Hi$ each, but the last two store a qubit in $\C^2$ each. Register $\Reg_3$ is used to store if the current vertex $x$ is marked, but $\Reg_4$ is used for performing rotations. Let
\begin{equation}
  R_\alpha := \mx{\cos \alpha & -\sin \alpha \\ \sin \alpha & \cos \alpha}
\end{equation}
denote the rotation by angle $\alpha$. An algorithm for implementing the transformation $\ket{x} \ancilla \mapsto \ket{x} \ket{p_x(s)}$ is given below.

\begin{algobox}{0.9}
\Itptxt$(P,M,s)$\label{alg:Itp}
\begin{enumerate}
  \item Let the initial state be $\ket{x} \ancilla \ket{0} \ket{0}$.
  \item Apply $\check(M)$ to $\Reg_1 \Reg_3$ (then $\Reg_3 = 1$ if and only if $x \in M$).
  \item If $\Reg_3 = 0$, apply $V(P)$ to $\Reg_1 \Reg_2$ and get
        $\ket{x} \ket{p_x} \ket{0} \ket{0}$.
  \item Otherwise:\label{step:marked}
  \begin{enumerate}
    \item The state is $\ket{x} \ancilla \ket{1} \ket{0}$ where $x \in M$.
    \item Apply $R_\alpha$ with $\alpha = \arcsin \sqrt{s}$ on $\Reg_4$:
          $\ket{x} \ancilla \ket{1} (\sqrt{1-s} \ket{0} + \sqrt{s} \ket{1})$.
    \item If $\Reg_4 = 0$, apply $V(P)$ on $\Reg_1 \Reg_2$.
          Otherwise, use CNOT to copy $\Reg_1$ to $\Reg_2$ in the standard basis:
          $\ket{x} (\sqrt{1-s} \ket{p_x} \ket{1} \ket{0} + \sqrt{s} \ket{x} \ket{1} \ket{1})$.
    \item If $\Reg_1 = \Reg_2$, apply $R_\alpha$ with\label{step:marked-d}
          $\alpha = -\arcsin \sqrt{s/((1-s)P_{xx}+s)}$ to $\Reg_4$.
          Otherwise, do nothing: $\ket{x} \ket{p_x(s)} \ket{1} \ket{0}$.
  \end{enumerate}
  \item Apply $\check(M)$ to $\Reg_1 \Reg_3$ to uncompute $\Reg_3$ and get\label{step:uncompute}
        $\ket{x} \ket{p_x(s)} \ket{0} \ket{0}$.
\end{enumerate}
\end{algobox}

Recall from \Eq{P(s)} that $P(s)$ has the following block structure:
\begin{equation}
  P(s) = \mx{
    P_{UU}      & P_{UM} \\
    (1-s)P_{MU} & (1-s)P_{MM} + s I
  }.
  \label{eq:P(s) block form}
\end{equation}
We will analyze the cases $x \in M$ and $x \in U$ separately. Then the general case will hold by linearity.

If $x \in U$ then the corresponding row of $P(s)$ does not depend on $s$, so $\ket{p_x(s)} = \ket{p_x}$. In this case \step{marked} of the above algorithm is never executed and the remaining steps effectively apply $V(P)$ to produce the correct state.

When $x \in M$ the algorithm is more involved. Let us analyze only \step{marked} where most of the work is done. During this step the state gets transformed as follows:
\begin{align}
            \ket{x} \ancilla \ket{1} \ket{0}
  & \mapsto \ket{x} \ancilla \ket{1} (\sqrt{1-s} \ket{0} + \sqrt{s} \ket{1}) \\
  & \mapsto \ket{x} \bigl( \sqrt{1-s} \ket{p_x} \ket{1} \ket{0} + \sqrt{s} \ket{x} \ket{1} \ket{1} \bigr) \\
  & \mapsto \ket{x} \ket{p_x(s)} \ket{1} \ket{0}.
\end{align}
The first two transformations are straightforward, so let us focus only on the last one which corresponds to \step{marked-d}. The state at the beginning of this step is
\begin{align}
  &  \ket{x} \bigl( \sqrt{1-s} \ket{p_x} \ket{1} \ket{0} + \sqrt{s} \ket{x} \ket{1} \ket{1} \bigr) \\
  &= \ket{x} \Biggl[
       \sqrt{1-s} \sum_{y \in X \setminus \set{x}} \sqrt{P_{xy}} \ket{y} \ket{1} \ket{0}
     + \ket{x} \ket{1} \Bigl( \sqrt{(1-s) P_{xx}} \ket{0} + \sqrt{s} \ket{1} \Bigr)
     \Biggr].
  \label{eq:marked-d}
\end{align}
Note from the second row of matrix $P(s)$ in \Eq{P(s) block form} that all its elements have acquired a factor of $1-s$, except the diagonal ones. Thus in \step{marked-d} we perform a rotation only when $\Reg_1 = \Reg_2$. This rotation affects only the second half of the state in \Eq{marked-d} and transfers all amplitude to $\ket{0}$ in the last register:
\begin{equation}
  \ket{x} \Biggl[
       \sqrt{1-s} \sum_{y \in X \setminus \set{x}} \sqrt{P_{xy}} \ket{y}
     + \sqrt{(1-s) P_{xx} + s} \ket{x}
     \Biggr] \ket{1} \ket{0}
  = \ket{x} \ket{p_x(s)} \ket{1} \ket{0}.
\end{equation}
Finally, \step{uncompute} uncomputes $\Reg_3$ to $\ket{0}$ and the final state is $\ket{x} \ket{p_x(s)} \ket{0} \ket{0}$ as desired.
\end{proof}

%%%%%%%%%%%%%%%%%%%%%%%%%%%%%%%%%%%%%%%%%%%%%%%%%%%%%%%%%%%%%%%%%%%%%
\section{An explicit formula for \texorpdfstring{$\limHT$}{HT+(P,M)}} \label{apx:HT comparison}
%%%%%%%%%%%%%%%%%%%%%%%%%%%%%%%%%%%%%%%%%%%%%%%%%%%%%%%%%%%%%%%%%%%%%

Recall from \Def{HT(s)} that $\limHT$ is defined as the $s \to 1$ limit of $\HT(s)$. In this appendix we derive an alternative expression for $\limHT$. This formula explicitly expresses $\limHT$ in terms of the Markov chain $P$ and its stationary distribution $\pi$, and makes it easier to evaluate this quantity and compare it to the regular hitting time $\HT(P,M)$.

Let us define unit vectors $\ket{\tilde{U}} \in \R^{\abs{U}}$ and $\ket{\tilde{M}} \in \R^{\abs{M}}$ as follows:
\begin{align}
  \ket{\tilde{U}} &:= \sqrt{\tilde{\pi}_U\tp}, &
  \ket{\tilde{M}} &:= \sqrt{\tilde{\pi}_M\tp},
  \label{eq:tildeUM}
\end{align}
where $\tilde{\pi}_U$ and $\tilde{\pi}_M$ are defined in \Eq{piUM} in terms of the stationary distribution $\pi = (\pi_U \; \pi_M)$ of $P$. Note from \Eq{U} that $\ket{\tilde{U}}$ and $\ket{\tilde{M}}$ are the restrictions of $\ket{U}$ and $\ket{M}$ to the unmarked and marked subspaces. Furthermore, let
\begin{equation}
  \mx{D_{UU} & D_{UM} \\
      D_{MU} & D_{MM}} :=
  \mx{\sqrt{P_{UU} \circ P_{UU}\tp} & \sqrt{P_{UM} \circ P_{MU}\tp} \\
      \sqrt{P_{MU} \circ P_{UM}\tp} & \sqrt{P_{MM} \circ P_{MM}\tp} }
  \label{eq:apx D blocks}
\end{equation}
be the blocks of the discriminant matrix $D(P)$ of $P$ (see \Def{Discriminant}).

\begin{lemma}\label{lem:HT comparison}
If $\HT(P,M)$ is the hitting time of $P$ (see \Def{HT}) and $\limHT$ is the extended hitting time (see \Def{HT(s)}) then
\begin{align}
  \HT(P,M) &= \bra{\tilde{U}} (I - D_{UU})^{-1} \ket{\tilde{U}},     \label{eq:HT block}  \\
  \limHT   &= \bra{\tilde{U}} (I - D_{UU} - S)^{-1} \ket{\tilde{U}}, \label{eq:HT+ block}
\end{align}
where
\begin{equation}
  S := D_{UM} \Biggl[
         (I - D_{MM})^{-1}
       - \frac{(I - D_{MM})^{-1} \proj{\tilde{M}} (I - D_{MM})^{-1}}
              {\bra{\tilde{M}} (I - D_{MM})^{-1} \ket{\tilde{M}}}
       \Biggr] D_{MU}.
  \label{eq:apx S def}
\end{equation}
Vectors $\ket{\tilde{U}}$ and $\ket{\tilde{M}}$ are defined in \Eq{tildeUM} and matrices $D_{UU}, D_{UM}, D_{MU}, D_{MM}$ in \Eq{apx D blocks}.
\end{lemma}

\begin{proof}
Let us first derive \Eq{HT block}. Recall from \Eq{HT(P,M) series} that $\HT(P,M)$ can be written as
\begin{equation}
  \HT(P,M) = \sum_{t=0}^\infty \bra{U} D(1)^t \ket{U},
  \label{eq:apx HT(P,M) series}
\end{equation}
where $D(1)$ is the discriminant matrix of $P(1) = P'$. Recall from \Eq{D(1) circ} that
\begin{equation}
  D(1) = \mx{\sqrt{P_{UU} \circ P_{UU}\tp} & 0 \\ 0 & I}.
  \label{eq:apx D(1)}
\end{equation}
Since $D(1)$ is block diagonal and $\ket{U}$ acts only on the unmarked states $U$, we can restrict each term in \Eq{apx HT(P,M) series} to the unmarked subspace and bring the summation inside:
\begin{equation}
  \HT(P,M) = \bra{\tilde{U}} \sum_{t=0}^\infty D(1)_{UU}^t \ket{\tilde{U}}.
  \label{eq:apx HT(P,M) block series}
\end{equation}
Recall from \Eq{D(s) blocks} that the $UU$ block of $D(s)$ is independent of $s$, hence $D(1)_{UU} = D_{UU}$, the $UU$ block of $D(0)$ given in \Eq{apx D blocks}. Recall from \Prop{Invertible} that $I - P_{UU}$ is invertible. Furthermore, due to \Prop{PUU limit} we can write $(I - P_{UU})^{-1} = \sum_{t=0}^\infty P_{UU}^t$. As $D_{UU}$ and $P_{UU}$ are similar according to \Eq{DUU}, $I - D_{UU}$ is also invertible and $(I - D_{UU})^{-1} = \sum_{t=0}^\infty D_{UU}^t$. If we substitute this in \Eq{apx HT(P,M) block series}, we get \Eq{HT block} and thus prove the first half of the lemma.

For the second half, recall from \Eq{HT(s) series} that for $s \in [0,1)$,
\begin{equation}
  \HT(s) = \sum_{k=1}^{n-1} \frac{\abs{\braket{v_k(s)}{U}}^2}{1-\lambda_k(s)},
  \label{eq:apx HT(s) fractions}
\end{equation}
where $\lambda_k(s)$ and $\ket{v_k(s)}$ are the eigenvalues and eigenvectors of the discriminant matrix $D(s)$. By \Prop{Multiplicity}, for any $s \in [0,1)$, $\lambda_n(s) = 1$ and $\lambda_k(s) < 1$ for all $k \neq n$. Let $\Pi_n(s) := \proj{v_n(s)}$, where $\ket{v_n(s)}$ is given by \Prop{vn(s)}:
\begin{equation}
  \ket{v_n(s)} = \cos \theta(s) \ket{U} + \sin \theta(s) \ket{M}.
  \label{eq:apx vn(s)}
\end{equation}
With this in mind, we can rewrite \Eq{apx HT(s) fractions} as follows:
\begin{align}
  \HT(s)
  &= \bra{U} \Biggl[ \sum_{k=1}^{n-1} \sum_{t=0}^{\infty} \lambda_k^t(s) \proj{v_k(s)} \Biggr] \ket{U} \\
  &= \bra{U} \sum_{t=0}^\infty \bigl( D^t(s) - \Pi_n(s) \bigr) \ket{U} \\
  &= \bra{U} \Biggl[ I + \sum_{t=1}^\infty \bigl( D(s) - \Pi_n(s) \bigr)^t - \Pi_n(s) \Biggr] \ket{U} \\
  &= \bra{U} \Bigl[ \bigl( I - D(s) + \Pi_n(s) \bigr)^{-1} - \Pi_n(s) \Bigr] \ket{U} \\
  &= \bra{U} \bigl( I - D(s) + \Pi_n(s) \bigr)^{-1} \ket{U} - \cos^2 \theta(s), \label{eq:apx HT(s) matrix inverse}
\end{align}
where the last equality follows from \Eq{apx vn(s)}.

Our goal is to compute $\lim_{s \to 1} \HT(s)$. Recall from \Prop{Multiplicity} that $D(1)$ has eigenvalue $1$ with multiplicity $\abs{M}$. Thus, if $\abs{M} > 1$, the matrix $I - D(s) + \Pi_n(s)$ in \Eq{apx HT(s) matrix inverse} is not invertible at $s = 1$, hence we cannot compute the limit by simply substituting $s = 1$. Let us rewrite this expression before we take the limit.

Note that the discriminant matrix $D(s)$ at $s = 0$ agrees with $D(P)$. Using \Eq{D(s) blocks} that relates $D(s)$ and $D(P)$, we can write
\begin{equation}
  I - D(s)
   = \mx{         I - D_{UU} & -\sqrt{1-s} D_{UM} \\
         - \sqrt{1-s} D_{MU} & (1-s) (I - D_{MM}) },
\end{equation}
where $\smx{D_{UU} & D_{UM} \\ D_{MU} & D_{MM}}$ are the blocks of $D(P)$ given in \Eq{apx D blocks}. Next, note that
\begin{equation}
  \ket{v_n(s)} = \mx{
    \cos \theta(s) \ket{\tilde{U}} \\
    \sin \theta(s) \ket{\tilde{M}}
  },
\end{equation}
so we can write
\begin{equation}
  \Pi_n(s)
   = \mx{               \cos^2 \theta(s) \ket{\tilde{U}} \bra{\tilde{U}} &
         \cos \theta(s) \sin   \theta(s) \ket{\tilde{U}} \bra{\tilde{M}} \\
         \cos \theta(s) \sin   \theta(s) \ket{\tilde{M}} \bra{\tilde{U}} &
                        \sin^2 \theta(s) \ket{\tilde{M}} \bra{\tilde{M}} }.
\end{equation}
Putting the two equations together, we can write $I - D(s) + \Pi_n(s)$ as
\begin{equation}
  \mx{
    I - D_{UU} + \cos^2 \theta(s) \ket{\tilde{U}} \bra{\tilde{U}} &
      - \sqrt{1-s} D_{UM} + \cos   \theta(s) \sin \theta(s) \ket{\tilde{U}} \bra{\tilde{M}} \\
      - \sqrt{1-s} D_{MU} + \cos   \theta(s) \sin \theta(s) \ket{\tilde{M}} \bra{\tilde{U}} &
        (1-s)(I - D_{MM}) + \sin^2 \theta(s) \ket{\tilde{M}} \bra{\tilde{M}}
  }.
  \label{eq:apx blocks}
\end{equation}
In \Eq{apx HT(s) matrix inverse} we need only the upper left block of the inverse of the above matrix, since $\ket{U}$ is non-zero only on the $U$ block. According to the block-wise inversion formula,
\begin{equation}
  \mx{A & B \\ B\tp & C}^{-1} = 
  \mx{(A - B C^{-1} B\tp)^{-1} & \ldots\quad \\ \ldots\quad & \ldots\quad}.
  \label{eq:apx Block inverse}
\end{equation}
Thus, \Eq{apx HT(s) matrix inverse} becomes
\begin{equation}
  \HT(s)
  = \bra{\tilde{U}} \bigl( A(s) - B(s) C(s)^{-1} B(s)\tp \bigr)^{-1} \ket{\tilde{U}}
  - \cos^2 \theta(s),
  \label{eq:apx HT(s) and ABC original}
\end{equation}
where $A(s)$, $B(s)$, and $C(s)$ are the blocks in \Eq{apx blocks}. We can further rewrite this as follows:
\begin{equation}
  \HT(s)
  = \bra{\tilde{U}}
    \biggl[
      A(s) - \frac{B(s)}{\sqrt{1-s}} \biggl(\frac{C(s)}{1-s}\biggr)^{-1} \frac{B(s)\tp}{\sqrt{1-s}}
    \biggr]^{-1} \ket{\tilde{U}}
  - \cos^2 \theta(s),
  \label{eq:apx HT(s) and ABC}
\end{equation}
where the extra factors will allows us to deal with the fact that $C(1)$ is singular.

Now we can compute $\lim_{s \to 1} \HT(s)$ for each piece of \Eq{apx HT(s) and ABC} separately. Note from \Eq{cos and sin theta} that $\cos^2 \theta(s)$ vanishes as $s \to 1$. Similarly, we also get that
\begin{align}
  A' &:= \lim_{s \to 1} A(s) = I - D_{UU}, \label{eq:apx A'} \\
  B' &:= \lim_{s \to 1} \frac{B(s)}{\sqrt{1-s}}
       = -D_{UM} + \sqrt{\frac{1-p_M}{p_M}} \ket{\tilde{U}} \bra{\tilde{M}}. \label{eq:apx B'}
\end{align}
Finally, notice that $\lim_{s \to 1} C(s)/(1-s)$ does not exist. Nevertheless, the limit of the inverse exists (in particular, it is a singular matrix) and we can compute it using the Sherman--Morrison formula:
\begin{equation}
  \bigl( X + \proj{\psi} \bigr)^{-1}
  = X^{-1} - \frac{X^{-1} \proj{\psi} X^{-1}}{1 + \bra{\psi} X^{-1} \ket{\psi}}.
\end{equation}
For $s < 1$, we get
\begin{align}
  \biggl(\frac{C(s)}{1-s}\biggr)^{-1}
  &= \biggl(
       I - D_{MM} + \frac{\sin^2 \theta(s)}{1-s} \ket{\tilde{M}} \bra{\tilde{M}}
     \biggr)^{-1} \\
  &= (I - D_{MM})^{-1}
   - \frac{(I - D_{MM})^{-1} \proj{\tilde{M}} (I - D_{MM})^{-1}}
     {\frac{1-s}{\sin^2 \theta(s)} + \bra{\tilde{M}} (I - D_{MM})^{-1} \ket{\tilde{M}}},
\end{align}
so the limit is
\begin{equation}
  C':= \lim_{s \to 1} \biggl(\frac{C(s)}{1-s}\biggr)^{-1}
     = (I - D_{MM})^{-1}
     - \frac{(I - D_{MM})^{-1} \proj{\tilde{M}} (I - D_{MM})^{-1}}
            {\bra{\tilde{M}} (I - D_{MM})^{-1} \ket{\tilde{M}}}.
  \label{eq:apx C'}
\end{equation}

Let $S(s) := B(s) C(s)^{-1} B(s)\tp$ be the matrix that appears in \Eq{apx HT(s) and ABC original}. Since it also appears in \Eq{apx HT(s) and ABC}, we find that
\begin{equation}
  S' := \lim_{s \to 1} S(s) = B' C' {B'}\tp
  \label{eq:apx S' def}
\end{equation}
by substituting $B'$ and $C'$ from Eqs.~\EqRef{apx B'} and \EqRef{apx C'}, respectively. Note from \Eq{apx C'} that $C' \ket{\tilde{M}} = 0$, so \Eq{apx S' def} simplifies to
\begin{equation}
  S' = D_{UM} C' D_{MU} \label{eq:apx S'}
\end{equation}
after we substitute $B'$ from \Eq{apx B'}. Note that $S'$ agrees with \Eq{apx S def} and that
\begin{equation}
  \limHT = \lim_{s \to 1} \HT(s) = \bra{\tilde{U}} (A' - S')^{-1} \ket{\tilde{U}},
  \label{eq:apx HT(1)}
\end{equation}
where $A'$ and $S'$ are given in Eqs.~(\ref{eq:apx A'}) and~(\ref{eq:apx S'}), respectively. This completes the proof.
\end{proof}

% B I B L I O G R A P H Y

\bibliographystyle{alphaurl}
\newcommand{\Proc}{Proceedings of the}
\bibliography{References}

\newcommand{\etalchar}[1]{$^{#1}$}
\begin{thebibliography}{HMdW03}

\bibitem[AA05]{AaronsonA05}
Scott Aaronson and Andris Ambainis.
\newblock Quantum search of spatial regions.
\newblock {\em Theory of Computing}, 1(4):47--79, 2005.
\newblock \href {http://arxiv.org/abs/quant-ph/0303041}
  {\path{arXiv:quant-ph/0303041}}, \href
  {http://dx.doi.org/10.4086/toc.2005.v001a004}
  {\path{doi:10.4086/toc.2005.v001a004}}.

\bibitem[ABN{\etalchar{+}}11]{ABNOR}
Andris Ambainis, Arturs Backurs, Nikolajs Nahimovs, Raitis Ozols, and Alexander
  Rivosh.
\newblock Search by quantum walks on two-dimensional grid without amplitude
  amplification.
\newblock 2011.
\newblock \href {http://arxiv.org/abs/1112.3337} {\path{arXiv:1112.3337}}.

\bibitem[AKR05]{AKR}
Andris Ambainis, Julia Kempe, and Alexander Rivosh.
\newblock Coins make quantum walks faster.
\newblock In {\em \Proc{} 16th ACM-SIAM Symposium on Discrete Algorithms
  (SODA'05)}, pages 1099--1108. SIAM, 2005.
\newblock URL: \url{http://dl.acm.org/citation.cfm?id=1070432.1070590}, \href
  {http://arxiv.org/abs/quant-ph/0402107} {\path{arXiv:quant-ph/0402107}}.

\bibitem[Amb07]{Ambainis04}
Andris Ambainis.
\newblock Quantum walk algorithm for element distinctness.
\newblock {\em SIAM J. Comput.}, 37(1):210--239, 2007.
\newblock \href {http://arxiv.org/abs/quant-ph/0311001}
  {\path{arXiv:quant-ph/0311001}}, \href
  {http://dx.doi.org/10.1137/S0097539705447311}
  {\path{doi:10.1137/S0097539705447311}}.

\bibitem[B{\v{S}}06]{BuhrmanS06}
Harry Buhrman and Robert {\v{S}}palek.
\newblock Quantum verification of matrix products.
\newblock In {\em \Proc{} 17th ACM-SIAM Symposium on Discrete Algorithms
  (SODA'06)}, pages 880--889. ACM, 2006.
\newblock \href {http://arxiv.org/abs/quant-ph/0409035}
  {\path{arXiv:quant-ph/0409035}}, \href
  {http://dx.doi.org/10.1145/1109557.1109654}
  {\path{doi:10.1145/1109557.1109654}}.

\bibitem[CEMM98]{CEMM98}
Richard Cleve, Artur Ekert, Chiara Macchiavello, and Michele Mosca.
\newblock Quantum algorithms revisited.
\newblock {\em Proceedings of the Royal Society of London. Series A:
  Mathematical, Physical and Engineering Sciences}, 454(1969):339--354, 1998.
\newblock \href {http://arxiv.org/abs/quant-ph/9708016}
  {\path{arXiv:quant-ph/9708016}}, \href
  {http://dx.doi.org/10.1098/rspa.1998.0164}
  {\path{doi:10.1098/rspa.1998.0164}}.

\bibitem[CG04a]{childs2}
Andrew~M. Childs and Jeffrey Goldstone.
\newblock Spatial search and the {D}irac equation.
\newblock {\em Phys. Rev. A}, 70(4):042312, 2004.
\newblock \href {http://arxiv.org/abs/quant-ph/0405120}
  {\path{arXiv:quant-ph/0405120}}, \href
  {http://dx.doi.org/10.1103/PhysRevA.70.042312}
  {\path{doi:10.1103/PhysRevA.70.042312}}.

\bibitem[CG04b]{childs1}
Andrew~M. Childs and Jeffrey Goldstone.
\newblock Spatial search by quantum walk.
\newblock {\em Phys. Rev. A}, 70(2):022314, 2004.
\newblock \href {http://arxiv.org/abs/quant-ph/0306054}
  {\path{arXiv:quant-ph/0306054}}, \href
  {http://dx.doi.org/10.1103/PhysRevA.70.022314}
  {\path{doi:10.1103/PhysRevA.70.022314}}.

\bibitem[FGGS00]{Farhi}
Edward Farhi, Jeffrey Goldstone, Sam Gutmann, and Michael Sipser.
\newblock Quantum computation by adiabatic evolution.
\newblock 2000.
\newblock \href {http://arxiv.org/abs/quant-ph/0001106}
  {\path{arXiv:quant-ph/0001106}}.

\bibitem[FRPU94]{FeigeRPU94}
Uriel Feige, Prabhakar Raghavan, David Peleg, and Eli Upfal.
\newblock Computing with noisy information.
\newblock {\em SIAM J. Comput.}, 23(5):1001--1018, 1994.
\newblock \href {http://dx.doi.org/10.1137/S0097539791195877}
  {\path{doi:10.1137/S0097539791195877}}.

\bibitem[GS97]{GrinsteadSnell}
Charles~M. Grinstead and J.~Laurie Snell.
\newblock {\em Introduction to Probability}.
\newblock 2nd ed. American Mathematical Society, 1997.
\newblock URL: \url{http://books.google.com/books?id=14oq4uWGCkwC}.

\bibitem[HJ90]{HornJohnson}
Roger~A. Horn and Charles~R. Johnson.
\newblock {\em Matrix Analysis}.
\newblock Cambridge University Press, 1990.
\newblock URL: \url{http://books.google.com/books?id=PlYQN0ypTwEC}.

\bibitem[HMdW03]{HMdW03}
Peter H{\o}yer, Michele Mosca, and Ronald de~Wolf.
\newblock Quantum search on bounded-error inputs.
\newblock In {\em Proceedings of the 30th International Colloquium on Automata,
  Languages and Programming (ICALP'03)}, volume 2719 of {\em Lecture Notes in
  Computer Science}, pages 291--299. Springer, 2003.
\newblock \href {http://arxiv.org/abs/quant-ph/0304052}
  {\path{arXiv:quant-ph/0304052}}, \href
  {http://dx.doi.org/10.1007/3-540-45061-0_25}
  {\path{doi:10.1007/3-540-45061-0_25}}.

\bibitem[KB06]{KB}
Hari Krovi and Todd~A. Brun.
\newblock Hitting time for quantum walks on the hypercube.
\newblock {\em Phys. Rev. A}, 73(3):032341, 2006.
\newblock \href {http://arxiv.org/abs/quant-ph/0510136}
  {\path{arXiv:quant-ph/0510136}}, \href
  {http://dx.doi.org/10.1103/PhysRevA.73.032341}
  {\path{doi:10.1103/PhysRevA.73.032341}}.

\bibitem[Kem05]{Kempe}
Julia Kempe.
\newblock Discrete quantum walks hit exponentially faster.
\newblock {\em Prob. Th. Rel. Fields}, 133(2):215--235, 2005.
\newblock \href {http://arxiv.org/abs/quant-ph/0205083}
  {\path{arXiv:quant-ph/0205083}}, \href
  {http://dx.doi.org/10.1007/s00440-004-0423-2}
  {\path{doi:10.1007/s00440-004-0423-2}}.

\bibitem[Kit95]{Kitaev95}
Alexei Kitaev.
\newblock Quantum measurements and the {A}belian {S}tabilizer {P}roblem.
\newblock 1995.
\newblock \href {http://arxiv.org/abs/quant-ph/9511026}
  {\path{arXiv:quant-ph/9511026}}.

\bibitem[KMOR10]{KMOR}
Hari Krovi, Fr{\'e}d{\'e}ric Magniez, Maris Ozols, and J{\'e}r{\'e}mie Roland.
\newblock Finding is as easy as detecting for quantum walks.
\newblock In {\em Automata, Languages and Programming}, volume 6198 of {\em
  Lecture Notes in Computer Science}, pages 540--551. Springer, 2010.
\newblock \href {http://arxiv.org/abs/1002.2419v1} {\path{arXiv:1002.2419v1}},
  \href {http://dx.doi.org/10.1007/978-3-642-14165-2_46}
  {\path{doi:10.1007/978-3-642-14165-2_46}}.

\bibitem[KOR10]{KOR}
Hari Krovi, Maris Ozols, and J{\'e}r{\'e}mie Roland.
\newblock Adiabatic condition and the quantum hitting time of {M}arkov chains.
\newblock {\em Phys. Rev. A}, 82(2):022333, Aug 2010.
\newblock \href {http://arxiv.org/abs/1004.2721v1} {\path{arXiv:1004.2721v1}},
  \href {http://dx.doi.org/10.1103/PhysRevA.82.022333}
  {\path{doi:10.1103/PhysRevA.82.022333}}.

\bibitem[KS60]{KemenySnell}
John~G. Kemeny and J.~Laurie Snell.
\newblock {\em Finite Markov Chains}.
\newblock Undergraduate Texts in Mathematics. Springer, 1960.
\newblock URL: \url{http://books.google.com/books?id=0bTK5uWzbYwC}.

\bibitem[KS07]{KoralovSinai}
Leonid~B. Koralov and Yakov~G. Sinai.
\newblock {\em Theory of Probability and Random Processes}.
\newblock Springer, 2007.
\newblock URL: \url{http://books.google.com/books?id=tlWOphOFRgwC}.

\bibitem[Mey00]{Meyer}
Carl~D. Meyer.
\newblock {\em Matrix Analysis and Applied Linear Algebra}, volume~1.
\newblock SIAM, 2000.
\newblock URL: \url{http://books.google.com/books?id=Zg4M0iFlbGcC}.

\bibitem[MN07]{MagniezN05}
Fr{\'e}d{\'e}ric Magniez and Ashwin Nayak.
\newblock Quantum complexity of testing group commutativity.
\newblock {\em Algorithmica}, 48(3):221--232, 2007.
\newblock \href {http://arxiv.org/abs/quant-ph/0506265}
  {\path{arXiv:quant-ph/0506265}}, \href
  {http://dx.doi.org/10.1007/s00453-007-0057-8}
  {\path{doi:10.1007/s00453-007-0057-8}}.

\bibitem[MNRS07]{MNRS}
Fr{\'e}d{\'e}ric Magniez, Ashwin Nayak, J{\'e}r{\'e}mie Roland, and Miklos
  Santha.
\newblock Search via quantum walk.
\newblock In {\em \Proc{} 39th ACM Symposium on Theory of Computing (STOC'07)},
  pages 575--584. ACM Press, 2007.
\newblock \href {http://arxiv.org/abs/quant-ph/0608026}
  {\path{arXiv:quant-ph/0608026}}, \href
  {http://dx.doi.org/10.1145/1250790.1250874}
  {\path{doi:10.1145/1250790.1250874}}.

\bibitem[MNRS12]{MNRS2}
Fr{\'e}d{\'e}ric Magniez, Ashwin Nayak, Peter Richter, and Miklos Santha.
\newblock On the hitting times of quantum versus random walks.
\newblock {\em Algorithmica}, 63(1):91--116, 2012.
\newblock \href {http://arxiv.org/abs/0808.0084} {\path{arXiv:0808.0084}},
  \href {http://dx.doi.org/10.1007/s00453-011-9521-6}
  {\path{doi:10.1007/s00453-011-9521-6}}.

\bibitem[MSS07]{MagniezSS05}
Fr{\'e}d{\'e}ric Magniez, Miklos Santha, and Mario Szegedy.
\newblock Quantum algorithms for the triangle problem.
\newblock {\em SIAM J. Comput.}, 37(2):413--424, 2007.
\newblock \href {http://arxiv.org/abs/quant-ph/0310134}
  {\path{arXiv:quant-ph/0310134}}, \href {http://dx.doi.org/10.1137/050643684}
  {\path{doi:10.1137/050643684}}.

\bibitem[SKW03]{ShenviKW03}
Neil Shenvi, Julia Kempe, and Birgitta~K. Whaley.
\newblock Quantum random-walk search algorithm.
\newblock {\em Phys. Rev. A}, 67(5):052307, May 2003.
\newblock \href {http://arxiv.org/abs/quant-ph/0210064}
  {\path{arXiv:quant-ph/0210064}}, \href
  {http://dx.doi.org/10.1103/PhysRevA.67.052307}
  {\path{doi:10.1103/PhysRevA.67.052307}}.

\bibitem[Sze04a]{Sze}
Mario Szegedy.
\newblock Quantum speed-up of {M}arkov chain based algorithms.
\newblock In {\em \Proc{} 45th IEEE Symposium on Foundations of Computer
  Science (FOCS'04)}, pages 32--41. IEEE Computer Society Press, 2004.
\newblock \href {http://dx.doi.org/10.1109/FOCS.2004.53}
  {\path{doi:10.1109/FOCS.2004.53}}.

\bibitem[Sze04b]{Sze-arXiv}
Mario Szegedy.
\newblock Spectra of quantized walks and a $\sqrt{\delta\varepsilon}$-rule.
\newblock 2004.
\newblock \href {http://arxiv.org/abs/quant-ph/0401053}
  {\path{arXiv:quant-ph/0401053}}.

\bibitem[Tul08]{Tulsi}
Avatar Tulsi.
\newblock Faster quantum-walk algorithm for the two-dimensional spatial search.
\newblock {\em Phys. Rev. A}, 78(1):012310, 2008.
\newblock \href {http://arxiv.org/abs/0801.0497} {\path{arXiv:0801.0497}},
  \href {http://dx.doi.org/10.1103/PhysRevA.78.012310}
  {\path{doi:10.1103/PhysRevA.78.012310}}.

\bibitem[VKB08]{VKB}
Martin Varbanov, Hari Krovi, and Todd~A. Brun.
\newblock Hitting time for the continuous quantum walk.
\newblock {\em Phys. Rev. A}, 78(2):022324, Aug 2008.
\newblock \href {http://arxiv.org/abs/0803.3446} {\path{arXiv:0803.3446}},
  \href {http://dx.doi.org/10.1103/PhysRevA.78.022324}
  {\path{doi:10.1103/PhysRevA.78.022324}}.

\end{thebibliography}

\end{document}